%% file: GJL24_arxiv.tex
\documentclass[11pt]{article}

\usepackage{fullpage}
\usepackage[usenames,dvipsnames]{xcolor}
\usepackage[colorlinks,citecolor=blue,linkcolor=BrickRed]{hyperref}
\usepackage{makeidx}
\usepackage{algorithm}
\usepackage{algorithmic}
\usepackage{graphicx,tipa,subfigure}
\usepackage{arcs,lmodern,fix-cm}
\usepackage{times}
\usepackage{amsfonts,latexsym,graphicx,epsfig,amssymb,color}
\usepackage{mathdots,amstext,setspace}
\usepackage{amsmath,amstext,setspace,enumerate,amsthm}
\usepackage{rotating}
\usepackage{verbatim}
\usepackage{mathabx}
\usepackage{soul}
\usepackage{diagbox}
\usepackage{makecell}
\usepackage{svg}
\usepackage{caption}
\usepackage{subcaption}

\usepackage{tikz}
\usetikzlibrary{arrows.meta, decorations.markings, quotes}
\usetikzlibrary{shapes.geometric}

\tikzset{currarrow/.style={-Stealth}} 

\newtheorem{theorem}{Theorem}[section]

\newtheorem{lemma}[theorem]{Lemma}
\newtheorem{corollary}[theorem]{Corollary}

\newtheorem{observation}[theorem]{Observation}

\newtheorem{example}[theorem]{Example}

\newcommand{\reals}{\mathbb{R}}

\newcommand{\sinn}[1]{\sin\left(#1\right)}

\newcommand{\ignore}[1]{}

\newcommand{\pe}[2]{$\sc{PE}^{#1}_{#2}$}
\newcommand{\de}[1]{$\sc{DE}_{#1}$}
\newcommand{\wpe}[2]{$\sc{PE}^{#2}({#1})$}
\newcommand{\wde}[1]{$\sc{DE}({#1})$}

\begin{document}

\title{
Multi-Agent Search-Type Problems on Polygons
}

\ignore{
\author{
Konstantinos Georgiou
\inst{1}
\thanks{Research supported in part by NSERC Discovery grant.}
\and
Caleb Jones
\and
Jesse Lucier
}
\institute{Department of Mathematics, Toronto Metropolitan University, Toronto, Ontario, Canada.}
}

\author{
Konstantinos Georgiou\footnotemark[1] ~\footnotemark[5] 
\and
Caleb Jones\footnotemark[1] ~\footnotemark[5]
\and 
Jesse Lucier\footnotemark[1] ~\footnotemark[5]
}

\def\thefootnote{\fnsymbol{footnote}}
\footnotetext[5]{Research supported in part by NSERC of Canada.}
\footnotetext[1]{
Department of Mathematics, 
Toronto Metropolitan University, 
Toronto, ON, Canada, \texttt{konstantinos,caleb.w.jones,jesse.lucier@torontomu.ca}
}

\maketitle

\begin{abstract}

We present several advancements in search-type problems for fleets of mobile agents operating in two dimensions under the wireless model. Potential hidden target locations are equidistant from a central point, forming either a disk (infinite possible locations) or regular polygons (finite possible locations). Building on the foundational disk evacuation problem, the disk priority evacuation problem with $k$ Servants, and the disk $w$-weighted search problem, we make improvements on several fronts.
First we establish new upper and lower bounds for the $n$-gon priority evacuation problem with $1$ Servant for $n \leq 13$, and for $n_k$-gons with $k=2, 3, 4$ Servants, where $n_2 \leq 11$, $n_3 \leq 9$, and $n_4 \leq 10$, offering tight or nearly tight bounds. The only previous results known were a tight upper bound for $k=1$ and $n=6$ and lower bounds for $k=1$ and $n \leq 9$. 
Second, our work improves the best lower bound known for the disk priority evacuation problem with $k=1$ Servant from $4.46798$ to $4.64666$ and for $k=2$ Servants from $3.6307$ to $3.65332$. 
Third, we improve the best lower bounds known for the disk $w$-weighted group search problem, significantly reducing the gap between the best upper and lower bounds for $w$ values where the gap was largest. These improvements are based on nearly tight upper and lower bounds for the $11$-gon and $12$-gon $w$-weighted evacuation problems, while previous analyses were limited only to lower bounds and only to $7$-gons.

\vspace{0.5cm}
\noindent
{\bf Mobile Agents, Evacuation, Priority Evacuation, Disk, Polygons, Wireless Model}
\end{abstract}

\tableofcontents

\section{Introduction}

Search theory is a branch of operations research that focuses on determining optimal strategies for locating targets amid uncertainty and limited information, with search and rescue missions being a prime example. Traditional \emph{search} requires an agent to identify the target quickly, while \emph{evacuation} missions need the entire fleet to gather at the hidden target. In some scenarios, only a specific agent must reach the target, or agents with different significance factors influence the solution's quality. Theoretical foundations for these problems date back to the 1960s and have been revitalized by recent advancements in robotics. The field now integrates tools from online algorithms, mobile agent distributed computing, operations research, discrete mathematics, and optimization.

Theoretically, autonomous mobile agents are modeled as volumeless entities navigating discrete domains (networks or graphs) or continuous domains (2D or 3D Euclidean spaces). In continuous domains, possible target locations shape the geometric search areas, including lines, disks, circles, squares, rectangles, and triangles. Significant progress has been made, yielding elegant distributed algorithms for mobile agent computing and technical arguments for impossibility results. However, achieving tight upper and lower bounds remains challenging, with many problems still lacking optimal solutions.

This work is the first to systematically explore a continuous search space with discrete target locations. We study \emph{priority evacuation}, an asymmetric search problem where a distinguished agent, the Queen, must reach the hidden item, aided by other agents, the Servants. The possible target locations form regular $n$-gons, and we provide upper and lower bounds, often tight, for various $n$ values and numbers of Servants. This problem has been previously studied for the line with $1$ Servant and the disk with any number of Servants. As a result, we improve the best lower bounds for the disk evacuation problem with $1$ and $2$ Servants. Additionally, we improve the best lower bounds known for the so-called $w$-weighted search problem on a disk. These improvements are based on new, and nearly tight bounds for the $w$-weighted search problem on $11$-gons and $12$-gons.

\subsection{Related Work}

The field of search problems dates back to the 1960s and has since developed a robust theoretical framework, detailed in various books and surveys~\cite{ahlswede1987search,alpern2013search,AlpGal03,czyzowicz2019groupkos,11340,hohzaki2016search}. Initially, the focus was on optimizing objectives for single searchers~\cite{beck1964linear,kleinberg1994line}, but with the rise of robotic fleets, multi-searcher problems have gained prominence~\cite{CzyzowiczGGKMP14,pattanayak2018evacuating}.

The linear search problem is a classic example studied extensively for both single~\cite{baezayates1993searching,bose2017general} and multiple searchers~\cite{ChrobakGGM15}. Variations aim at minimizing the weighted average of search completion times~\cite{GLweightedLine2023}. Other one-dimensional search settings include searching along rays~\cite{BrandtFRW20}, anomalous terrains~\cite{CzyzowiczKKNOS17}, graphs~\cite{AngelopoulosDL19}, and for multiple objects~\cite{Borowiecki0DK16,CzyzowiczDGKM16}.

In the past decade, attention has shifted to two-dimensional search problems, exploring domains like polygons~\cite{FeketeGK10}, disks~\cite{CzyzowiczGGKMP14}, planes~\cite{feinerman2017ants}, regular polygons~\cite{czyzowicz2020priority123}, equilateral triangles and squares~\cite{BagheriNO19,ChuangpishitMNO20,CzyzowiczKKNOS15,GJ22-triangle-algosensros}, arbitrary triangles~\cite{georgiou2022triangle}, and $\ell_p$ unit disks~\cite{GLLKllp2023}. Variations involve different communication models and searcher specifications, such as face-to-face communication~\cite{brandt2017collaboration,CGKNOV20,disser2019evacuating}, varying searcher speeds~\cite{BampasCGIKKP19}, searching for multiple exits~\cite{pattanayak2018evacuating}, and diverse communication capabilities~\cite{czyzowicz2021groupevac,GGK2022asym,georgiou2022asymmetricevacuation}. Fault tolerance has also been studied, addressing issues like faulty agents and Byzantine faults in~\cite{behrouz2023byzantine,BGMP2022pfaulty,CzyzowiczGGKKRW17,czyzowicz2021searchbyz} and in~\cite{czyzowicz2021searchnew,CzyzowiczKKNO19,GeorgiouKLPP19,Sun2020}.

Research has also explored non-standard objectives, including multi-objective search problems~\cite{chuangpishit2020multi}, competitive algorithmic approaches~\cite{0001DJ19}, and trade-offs between information and cost~\cite{MillerP15}. Studies have also considered time and energy 
trade-offs~\cite{czyzowiczICALP,czyzowicz2021energy}, and search-and-fetch problems in two dimensions~\cite{kranakis2019search,GEORGIOU202118fetch1}.

\subsection{Closely Related Work, Improvements and Significance of New Contributions}

Our work extends and improves a series of studies on search-type problems where possible hidden target locations lie on a unit disk, assuming instantaneous communication between agents under the wireless model. The initial study designed search trajectories to minimize the time for the last agent to reach the target, known as the disk evacuation problem~\cite{CzyzowiczGGKMP14}. This foundational study yielded mostly optimal results.

Subsequent research focused on designing search trajectories for a distinguished agent to evacuate quickly, with other agents, termed Servants, assisting. These disk priority evacuation problems were examined for $k=1,2,3$ Servants~\cite{czyzowicz2020priority123} and for $k\geq 4$ Servants~\cite{czyzowicz2020priority4}. Notably, significant gaps existed between the upper and lower bounds. The lower bound of 4.38962 for $k=1$ Servant relied on lower bounds for evacuating the significant agent from a $6$-gon, with the $n$-gon (as $n$ approaches infinity) approximating the disk priority evacuation problem.\footnote{The possible target placements on a disk are uncountable many, whereas for every $n$, the vertices of an $n$-gon are finite many. However, for every $\epsilon>0$, there is large enough $n$  so that every target placement on the disk is no more than $\epsilon$ away from a target placement on the $n$-gon.}
A tight upper bound for the $6$-gon priority evacuation was also provided.

An extension introduced in~\cite{GW24-iwoca} considered the $w$-weighted group search on a disk with two agents ($w \in [0,1]$), where $w=0$ corresponds to the disk priority evacuation problem with $k=1$ Servant. The paper provided upper and lower bounds, with gaps diminishing as $w \rightarrow 0$. Their lower bound relied on $w$-weighted group search bounds on $n$-gons with $n \leq 7$, but no upper bounds for the discrete domain were reported. They also provided lower bounds for the $n$-gon priority evacuation problem with $n \leq 9$, improving the~\cite{czyzowicz2020priority123} lower bound to 4.46798 for the disk priority evacuation problem with $k=1$ Servant.

We derive a number of improved results:
(i) We provide upper and lower bounds for the $n$-gon priority evacuation problem with one Servant for $n \leq 13$, achieving nearly tight or tight results (previously only $n=6$ upper bounds and $n \leq 9$ lower bounds were known).
(ii) We study priority evacuation on $n_k$-gons for $k=2,3,4$ Servants, with $n_2 \leq 11$, $n_3 \leq 9$, and $n_4 \leq 10$, offering tight or nearly tight bounds. No prior results were known.
(iii) We improve the best lower bound for the disk priority evacuation problem with $1$ Servant from 4.46798~\cite{GW24-iwoca} to 4.64666.
(iv) We present the first improved lower bound for the disk evacuation problem with $2$ Servants from 3.6307~\cite{czyzowicz2020priority123} to 3.65332.
(v) We improve the lower bounds for the disk $w$-weighted group search problem of~\cite{GW24-iwoca} for all $w \in [0,0.7]$.
(vi) Motivated by the fact that the previous improved lower bounds were obtained as corollaries to the 
$w$-weighted search problem on $n$-gons for $n=12$, we are the first to provide upper and lower bounds for the $11$-gon and the $12$-gon $w$-weighted evacuation problem for all $w \in [0,1]$, achieving tight or nearly tight results.

Our priority evacuation upper bounds are accompanied with search trajectories that resemble the best trajectories known for disk evacuation problems. Despite remaining gaps, this gives further indication that the known search trajectories may be optimal.

\subsection{Paper Organization}
\label{sec: paper organization}

In Section~\ref{sec: Notation and Terminology}, we introduce key notation and terminology. 
The formal definition of the problems we study, along with a formal description of our results can be found in Section~\ref{sec: definition and results}. These results pertain to the $n$-gon priority evacuation problem with $k$ Servants \pe{n}{k}, the disk priority evacuation problem with $k$ Servants \de{k}, the $n$-gon $w$-weighted search problem \wpe{w}{n}, and the disk $w$-weighted search problem \wde{w}. The proofs of our results for all these problems appear in the subsequent sections. 
Indeed, in Section~\ref{sec: formulations} we present formulations and relaxations for solving \pe{n}{k} (the same machinery is also applied later to~\wpe{w}{n}). 
Then, in Section~\ref{sec: lower bounds to pe} we discuss lower bounds on \pe{n}{k}, and in Section~\ref{sec: lower bounds to pe} upper bounds on the same problem. 
Equipped with the lower bounds on~\pe{n}{k}, we show improved lower bounds on~\de{k}, $k=1,2$ in Section~\ref{sec: lower bounds to de}. Then, in Section~\ref{sec: w-weighted results} we justify the reported upper and lower bounds for problems \wpe{w}{n}~and \wde{w}. Finally, in Section~\ref{sec: discussion}, we conclude with some open questions. 
Due to space limitations, many details have been moved to the appendix, and proper citations can be found for any omissions.

\ignore{
\pe{n}{k}		n-gon priority evacuation with k Servants \\
\de{k}		disk evacuation with k Servants \\
\wde{w}		weighted disk evacuation with 1 servant, with weight $w=0.5$ \\
\wpe{w}{n}	n-gon weighted evacuation with 1 servant and weight $w=0.5$ \\
\wde{0}=\de{1}
\wpe{0}{n} = \pe{n}{1}
}

\section{Preliminaries}
\label{sec: preliminaries}

\subsection{Notation \& Terminology}
\label{sec: Notation and Terminology}
We use $\|\cdot\|_2$ to denote the Euclidean norm over $\mathbb{R}^2$. For $n\in \mathbb{N}$, let $\mathcal{P}_n$ denote the set of all permutations of $\{1,\ldots, n\}$. We define $\mathcal{B}^n_k$ as the set of all $n$-dimensional $(k+1)$-ary strings, i.e., $\{0,\ldots,k\}^n$, so $\mathcal{B}^n_1$ is the set of $n$-dimensional binary strings.

We use the term \emph{unit speed trajectory} to refer to a continuous and differentiable function $\tau: \mathbb{R}_+ \to \mathbb{R}^2$ that induces a speed of at most 1. Specifically, if $\tau(t)=(\tau_1(t), \tau_2(t))$, then $\tau(t)$ is 1-Lipschitz continuous, meaning it satisfies $\| \tau(t_1)- \tau(t_2) \|_2 \leq |t_1 - t_2|$ for all $t_1, t_2 \in \mathbb{R}_+$.

In this work, $n$-gons exclusively refer to regular $n$-polygons inscribed in circles of radius 1. For convenience, we consider $n$-gons and the circle as embedded in the 2-dimensional Euclidean space. Thus, for a fixed $n$, the vertices of the $n$-gon are represented as
$
V^n_i := \left( \cos\left(2i\pi/n\right), \sin\left(2i\pi/n\right) \right), 
i=1,\ldots, n.
$
It follows that for all $i,j \in \{1,\ldots,n\}$, we have
$
\|V^n_i - V^n_j\|_2 = 2 \sin\left( \tfrac{\pi}{n} \cdot \mathrm{mod}(|i-j|,n) \right),
$
and that the edges of the $n$-gon have length $2 \sin\left( \tfrac{\pi}{n} \right)$.

Parameter $n$ will denote the number of vertices of the $n$-gon, and $k$ will denote the number of Servants in our search problem involving $k+1$ mobile agents. The following set will be useful in our later formulations:
$
\mathcal{X}^n_k := \{0,\ldots,k\} \times \{1,\ldots,n\},
$
which corresponds to pairs of agents and $n$-gon vertices in the multi-agent search problem we are considering.

\subsection{Problem Definition \& Main Contributions Made Formal}
\label{sec: definition and results}

\paragraph{Polygon Priority Evacuation:}
We study the \emph{Polygon Priority Evacuation Problem \pe{n}{k}} on $n$-gons with $k$ Servants, a multi-agent search problem where the hidden target lies in a discrete domain, specifically $n$-gons. In this search problem, the host space is the Euclidean 2-dimensional space (modeled for convenience as a Cartesian plane), and the searchers are $k+1$ unit speed \emph{agents}. Among these agents, one is labeled $0$ and is distinguished as the \emph{Queen}, while the other agents, labeled $1,2,\ldots,k$, are called \emph{Servants}. We consider a regular $n$-gon with vertices $V^n_i$, $i=1,\ldots,n$, inscribed in a unit radius disk and centered at the origin $O=(0,0)$. The orientation of the $n$-gon is known to the agents, i.e. to the algorithm.

Agent movements are determined by unit speed trajectories $\tau_i:\mathbb{R}_+ \to \mathbb{R}^2$, where the initial placements of the agents $\tau_i(0)$ are algorithmic choices, $i=0,\ldots,k$. The agents' movements are \textit{feasible} if for each polygon vertex $j \in \{1,\ldots,n\}$, there exists an agent $i \in \{0,\ldots,k\}$ and time $t'=t'(j) \in \mathbb{R}_+$ such that $\tau_i(t') = V^n_{j}$ (i.e., every polygon vertex is eventually visited by some agent). For each polygon vertex $j \in \{1,\ldots,n\}$, we denote by $T_j$ the smallest such $t'$, calling it the \emph{visitation time of vertex $j$}. The cost of the feasible solution $\{\tau_i\}_{i \in \{0,\ldots,k\}}$, referred to as the \textit{priority evacuation cost}, is defined as
$
\max_{j \in \{1,\ldots,n\} }
\left\{
T_j +\| V^n_j - \tau_0( T_j ) \|_2
\right\}
$
and the objective is to minimize that cost.

Next, we provide a high-level explanation of the above model. By definition, the agents' specifications correspond to the so-called wireless communication model, which allows them to share information instantaneously. All $k+1$ agents contribute to searching for a hidden target, often referred to as the \emph{exit}, located at one of the vertices of the $n$-gon. 
The optimal solution to the problem, should the position of the exit be known, is $1$. However, in this online problem, the hidden exit is identified only when any of the $k+1$ agents visit the corresponding $n$-gon vertex (unknown to the agents), and subsequently, all agents are notified accordingly. For an exit placement, the cost of a solution is given by the time that the Queen \emph{evacuates}, i.e., when she reaches the hidden exit, ignoring thereafter the whereabouts of the $k$ Servants. 
Compatible with worst-case analysis, the overall performance of a feasible solution is defined as the worst-case Queen evacuation time over all exit placements.

In this work, we provide upper and lower bounds on \pe{n}{k} with $k=1,2,3,4$ Servants over $n$-gons, and for various values of $n$. All our results are summarized in the next theorem, which is our first main contribution.

\begin{theorem}\footnote{Our positive and negative results provide only the first 5 digits of our computations, even though our numerical evaluations extend to at least 10 digits of accuracy. Often, we also have closed-form expressions, involving algebraic and trigonometric operations, that describe these numbers. However, for larger values of $n$, these expressions become too extensive to be informative, and hence we omit them.}
\label{thm: n-gon k Servants upper and lower bounds}
For $k=1,2,3,4$, and for various values of $n$, 
Priority Evacuation Problem \pe{n}{k} can be solved in time $u^n_k$.
Also, no algorithm for the problem has evacuation cost less than $l^n_k$.
The upper bounds $u^n_k$, and the lower bounds $l^n_k$ appear in Table~\ref{tab: summary of polygon evacuation results}.
Moreover, for $n$-gons $n=12,13$, we have 
$u^{12}_1 =3.38511$, $l^{12}_1 =3.38486$,  and $u^{13}_1 =3.36362$, $l^{13}_1 =3.36361$. 
\begin{table}[h]
\centering
\begin{small}
\begin{tabular}{|c|c|c|c|c|c|c|c|c|c|c|c|}
\hline
\diagbox{$k$}{$n$} & 3 & 4 & 5 & 6 & 7 & 8 & 9 & 10 & 11  \\ \hline
1 & \makecell{ 1.73205\\ 1.73205 \\ \textbf{1.}} & \makecell{ 2.14626\\ 2.12132 \\ \textbf{1.01175}} & \makecell{ 2.71441\\ 2.71441 \\ \textbf{1.} } & \makecell{ 2.86603\\ 2.86602 \\ \textbf{1.} } & \makecell{ 2.97391\\ 2.95125 \\ \textbf{1.00768} } & \makecell{ 3.02649\\ 3.00320 \\ \textbf{1.00776} } & \makecell{ 3.21891\\ 3.21891 \\ \textbf{1.} } & \makecell{ 3.21549\\ 3.18712 \\ \textbf{1.0089} } & \makecell{ 3.35919\\ 3.35577 \\ \textbf{1.00102} }  \\ \hline
2 & \makecell{ 1.00000\\ 1.00000 \\ \textbf{1.} } & \makecell{ 1.70711\\ 1.70710\\ \textbf{1.} } & \makecell{ 1.90211\\ 1.90211 \\ \textbf{1.} } & \makecell{ 2.00000\\  2.00000\\ \textbf{1.} } & \makecell{ 2.14027\\ 2.08348 \\ \textbf{1.02726} } & \makecell{ 2.25951\\ 2.23784 \\ \textbf{1.00968} } & \makecell{ 2.37176\\ 2.35288 \\ \textbf{1.00802} } & \makecell{ 2.38956\\ 2.37810 \\ \textbf{1.00482} } & \makecell{ 2.50211\\ 2.46291 \\ \textbf{1.01592} }  \\ \hline
3 & \makecell{ 1.00000\\ 1.00000 \\ \textbf{1.} } & \makecell{ 1.00000\\ 1.00000 \\ \textbf{1.} } & \makecell{ 1.55017\\ 1.53884 \\ \textbf{1.00736} } & \makecell{ 2.00000\\  1.86602 \\ \textbf{1.0718} } & \makecell{ 1.86777\\ 1.84269 \\ \textbf{1.01361} } & \makecell{ 1.91342\\ 1.91341 \\ \textbf{1.} } & \makecell{ 1.91362\\ 1.85083 \\ \textbf{1.03393} } & \makecell{ NA \\ NA \\ \textbf{NA} } & \makecell{ NA \\ NA \\ \textbf{NA} }  \\ \hline
4 & \makecell{ 1.00000\\ 1.00000 \\ \textbf{1.} } & \makecell{ 1.00000\\ 1.00000 \\ \textbf{1.} } & \makecell{ 1.00000\\ 1.00000 \\ \textbf{1.} } & \makecell{ 1.5000\\  1.50000 \\ \textbf{1.} } & \makecell{ 1.64960\\ 1.64959 \\ \textbf{1.} } & \makecell{ 1.76537\\ 1.68924 \\ \textbf{1.04507} } & \makecell{ 1.68404\\ 1.66884 \\ \textbf{1.00911} } & \makecell{ 1.65153\\ 1.61803 \\ \textbf{1.02070} } & \makecell{ NA \\ NA \\ \textbf{NA} }  \\ \hline
\end{tabular}
\caption{Summary of our upper and lower bound results for \pe{n}{k}. 
Every entry is populated with the upper bound $u^n_k$, the lower bound known $l^n_k$ , and the corresponding optimality gap, i.e. $u^n_k/l^n_k$. Therefore a reported gap of $1.0$ corresponds to an optimal result. For $n=12,13$, we only have results for \pe{n}{1}, and we report them separately.
However, for these $n$-gons, the optimality gaps are $1.00007$ and $1.0$, respectively. 
The only values known before were $u^6_1$, and $l^n_1$ for $n\leq 9$. 
}
\label{tab: summary of polygon evacuation results}
\end{small}
\end{table}
\end{theorem}

\paragraph{Disk Priority Evacuation:}
Our second contribution pertains to \emph{improved} lower bounds for the \emph{disk priority evacuation problem with $k$-Servants}, which we denote by \de{k}. 
The problem \de{k} was first considered in~\cite{czyzowicz2020priority123} for $k=1,2,3$ Servants, and in~\cite{czyzowicz2020priority4} for $k \geq 4$ Servants. For completeness, we include the definition of the problem by adapting the description of \pe{n}{k}. The primary difference between the two problems is that in the disk evacuation problem, the hidden exit can lie \emph{anywhere} on the perimeter of the unit radius disk. 
Moreover, the agents' starting positions are at the origin, i.e., for the feasible unit speed trajectories $\tau_i:\mathbb{R}_+ \to \mathbb{R}^2$, we have $\tau_i(0)=O=(0,0)$, for $i=0,\ldots,k$. We quantify the positioning of the exit by some $\theta \in [0,2\pi]$, corresponding to the exit placement
$
P_\theta := \left( \cos(\theta), \sin(\theta) \right).
$
Then, for a feasible solution $\{\tau_i\}_{i=0,\ldots,k}$, we require that for each $\theta$, there exists an agent $i \in \{0,\ldots,k\}$ and a time $t' \in \mathbb{R}_+$ such that $\tau_i(t') = P_\theta$. For each exit placement $\theta$, we denote by $T({\theta})$ the smallest corresponding time $t'$. The evacuation cost of the feasible solution $\{\tau_i\}_{i \in \{0,\ldots,k\}}$ is then defined as
$
\sup_{\theta \in [0,2\pi)}
\left\{
T({\theta}) + \| P_{\theta} - \tau_0(T({\theta})) \|_2
\right\}
$
and the objective is to minimize this cost.

Our second main contribution are new lower bounds for \de{k} which are summarized in the next statement. 
\begin{theorem}
\label{thm: improved lower bounds summary}
No algorithm for \de{k}~has evacuation cost
less than $4.64666$ for $k=1$, 
and less than $3.65332$ for $k=2$. 
\end{theorem}

For an informative perspective, we summarize in Table~\ref{tab:previous positive and negative disk evac results}
all previous best upper and lower bounds known for \de{k}.
\begin{table}[h]
\centering
\begin{tabular}{|c|c|c|c|c|}
\hline
		$k$		& $1$ 						& $2$ 						& $3$ 						& $4$ \\ \hline
Upper Bound & $4.81854$~\cite{czyzowicz2020priority123}	&  $3.8327$~\cite{czyzowicz2020priority123}	& $3.3738$~\cite{czyzowicz2020priority123}			&  3.30129~\cite{czyzowicz2020priority4} \\ \hline
Lower Bound & $4.56798$~\cite{GW24-iwoca}	&  $3.6307$~\cite{czyzowicz2020priority123}	& $3.2017$~\cite{czyzowicz2020priority123}	&  $2.91322$ ~\cite{czyzowicz2020priority4} \\ \hline
\end{tabular}
\caption{Summary of best upper and lower bound results known for \de{k}, prior to our work.
Contrast the upper and lower bound values to the improved lower bounds of Theorem~\ref{thm: improved lower bounds summary}.}
\label{tab:previous positive and negative disk evac results}
\end{table}

\paragraph{Polygon and Disk $w$-Weighted Search}
Our third contribution pertains to a search problem where the objective is the arithmetic weighted average of the termination times of $2$ mobile agents, first considered in~\cite{GW24-iwoca}. In this problem, the hidden item lies either in the unit radius disk and the $2$ agents start from the center of the disk, or it lies on a vertex of an $n$-gon, and the initial placement of the $2$ agents is an algorithmic choice. Having the feasible trajectories identified exactly as in \pe{n}{1}~and \de{1}, the cost of the solution is instead defined as the arithmetic weighted average of the times that the $2$ agents reach the hidden item. Indeed, for each $w\in [0,1]$ and feasible trajectories $\tau_0,\tau_1$, and for the exit placement $x$, let $T_i(x)$ denote the time that agent $i=0,1$ reaches $x$. The $w$-weighted search cost for input $x$ is defined as $(T_0(x) +w\cdot T_1(x))/(1+w)$. 

For the \emph{$n$-Gon $w$-Weighted Search Problem \wpe{w}{n}}, the objective is to minimize $\max_x \tfrac{T_0(x) +w\cdot T_1(x)}{1+w}$, where $x$ ranges over all $n$ vertices of the $n$-gon.
Similarly, for the \emph{Disk $w$-Weighted Search Problem \wde{w}}, the objective is to minimize $\sup_x \tfrac{T_0(x) +w\cdot T_1(x)}{1+w}$ and the supremum is considered over all points $x$ on the perimeter of the disk. 
From these definitions, it is immediate that setting $w=0$ results in the priority evacuation objective with $1$ Servant, i.e. that 
\wpe{0}{n}~is equivalent to problem \pe{n}{1}, 
and that \wde{0}~ is equivalent to problem \de{1}.

\ignore{
\pe{n}{k}		n-gon priority evacuation with k servants \\
\de{k}		disk evacuation with k servants \\
\wde{w}		weighted disk evacuation with 1 servant, with weight $w=0.5$ \\
\wpe{w}{n}	n-gon weighted evacuation with 1 servant and weight $w=0.5$ \\
\wde{0}=\de{1}
\wpe{0}{n} = \pe{n}{1}
}

The main contribution pertaining to the $w$-weighted search objective is an improvement to the previously best lower bound known for the disk. The result is obtained numerically, and is quantified in the next statement. 

\begin{theorem}
\label{thm: new w-weighted lower bound}
Each of the purple and brown curves in Figure~\ref{fig: wWeightedDisk} is a lower bound to the \wde{w}, for all $w\in [0,1]$.\footnote{The lower bounds were computed for values of $w$ from $0$ to $1$ with step size $0.01$.}
\end{theorem}

\begin{figure}[h!]
\centering
\includegraphics[width=7.5cm]{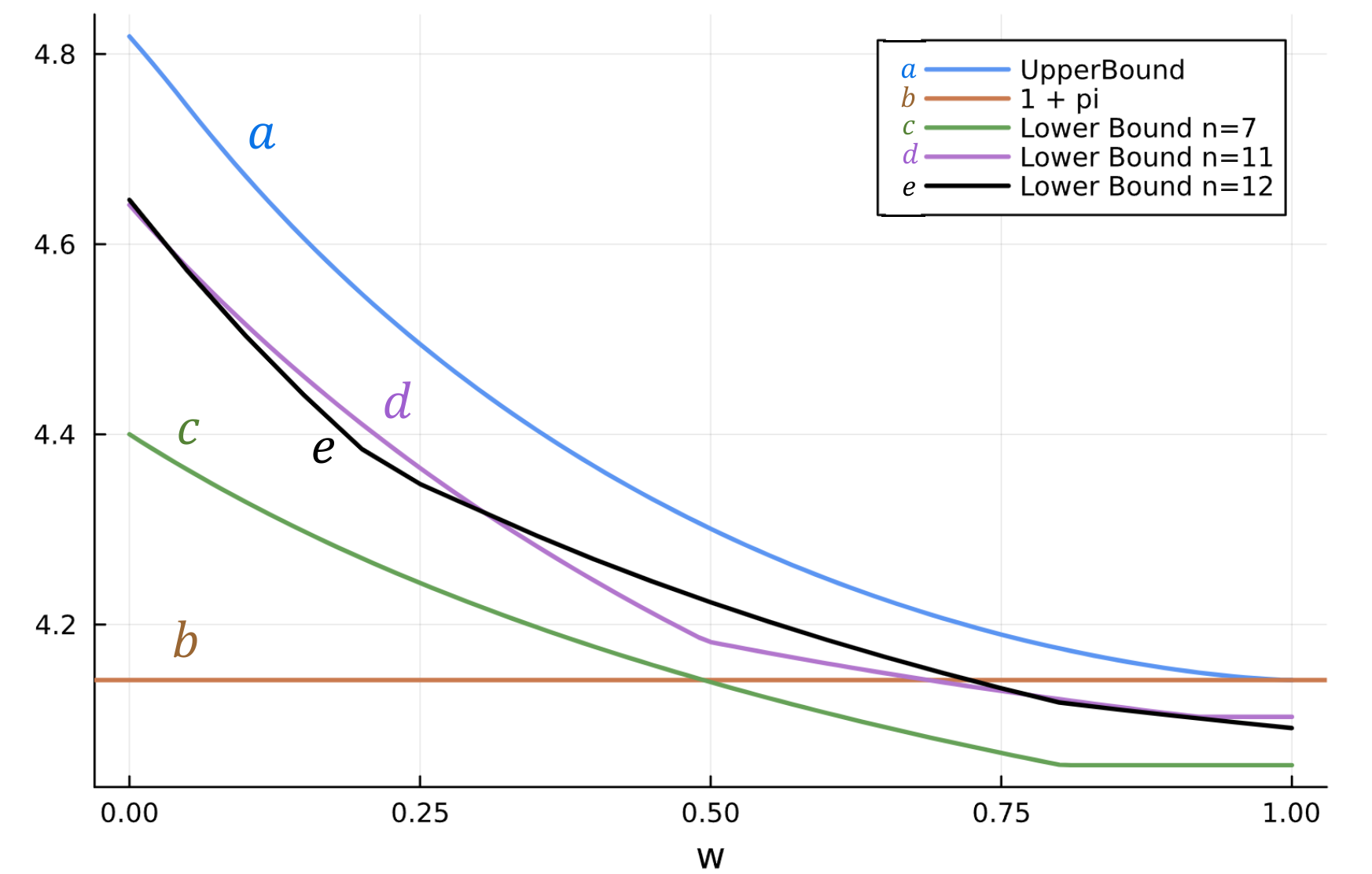}
\caption{Known upper bounds compared to our improved lower bounds for \wde{w}. 
The blue curve (labeled $a$) depicts the best upper bound known~\cite{GW24-iwoca}. The green curve (labeled $c$) depicts the previously best lower bound known, by a reduction to \wpe{w}{7}. The orange line (labeled $b$) depicts a universal lower bound of $1+\pi$ also proved in~\cite{GW24-iwoca}. 
The purple and black curves (labeled $d,e$, respectively) are new lower bounds on \wde{w} by a reductions to \wpe{w}{11}~ and \wpe{w}{12}, respectively (hence the maximum of them applies). 
The lower bounds were calculated for values of $w\in[0,1]$ starting from $0$ and with step size $0.01$. 
}
\label{fig: wWeightedDisk}
\end{figure}

As indicated before, the lower bounds on \wde{w} are obtained via reductions to \wpe{w}{n}. 
The best lower bound achieved as $w$ ranges in $[0,1]$ were obtained for $n=11,12$. We are therefore motivated to report upper and lower bounds for \wpe{w}{11}~and \wpe{w}{12}, demonstrating that our analysis is (nearly) tight. This is quantified in the following statement. 
\begin{theorem}
\label{thm: upper and lower w-weighted ngon}
For values of $w\in [0,1]$, Figure~\ref{fig: wWeightedgon} shows upper and lower bounds 
to the \wpe{w}{n} problem, for $n=11,12$. 
\footnote{The lower bounds were computed for values of $w$ from $0$ to $1$ with step size $0.01$.
The upper bounds were computed for values of $w$ from $0$ to $1$ with step size $0.02$, hence we only depict them as points. 
}
\end{theorem}

\begin{figure}[h!]
\centering
\includegraphics[width=14.5cm]{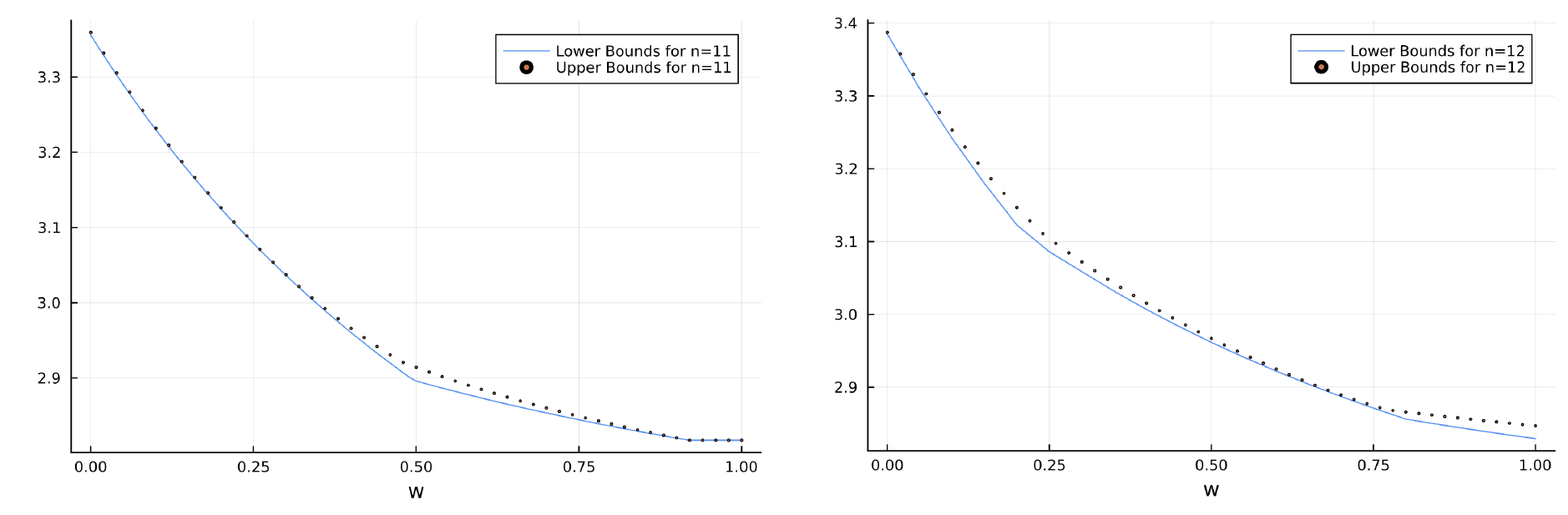}
\caption{
Upper and lower bound obtained for \wpe{w}{11} (left hand-side) and \wpe{w}{12} (right hand-side). Starting from $w=0$, the lower bounds were obtained with $w$ step size of $0.01$ and the upper bounds with $w$ step size of $0.02$. }
\label{fig: wWeightedgon}
\end{figure}


\section{Formulations \& Relaxations to \pe{n}{k}}
\label{sec: formulations}

Our contributions are based on the observation that the optimal solution to the \pe{n}{k} problem can be found using a Non-Linear Program (NLP), assuming we know the order in which the polygon vertices are visited and the identities of the agents who first visit those vertices. This idea was first implemented in~\cite{GW24-iwoca} for $k=1$ Servants, primarily to provide lower bounds for \de{1}, the disk priority evacuation problem with 1 Servant, and for the disk $w$-weighted search problem. Here, we extend this approach to provide upper and lower bounds also for \pe{n}{k} for $k=2,3,4$ and various values of $n$, improving the known lower bounds for \de{1} and introducing the first lower bound improvements for \de{k} for $k=2$.

First, we introduce some terminology to clarify the notation. For any feasible solution $\{\tau_i\}_{i \in \{0,\ldots,k\}}$ to \pe{n}{k}, each vertex $V^n_i$ of the $n$-gon is visited by some agent. Let permutation $\rho \in \mathcal P_n$ be the ordered list of vertices by their visitation times, where vertex $V^n_{\rho_i}$ has the $i$'th smallest visitation time (ties are broken arbitrarily). Define $s \in \mathcal B^n_k$ as the $(k+1)$-ary string where $s_i$ is the label of the agent visiting $\rho_i$. Thus, the corresponding feasible solution to \pe{n}{k} is called an \emph{$(s,\rho)$-algorithm}. We begin with the following observation.

\begin{observation}
The optimal solution to \pe{n}{k} is an $(s,\rho)$-algorithm, for some $s\in \mathcal B^n_k$ and $\rho \in \mathcal P_n$. 
\end{observation}

For fixed $s,\rho$, we show next how to find the optimal $(s,\rho)$-algorithm to \pe{n}{k}. 
To that end, consider the following Non Linear Program (NLP) on variables $t_1, \ldots, t_n$ and $d_{r_1,r_2}$, where $r_1,r_2 \in \mathcal X^n_k $.

\begin{align}
\min &~\max_{j \in \{1,\ldots, n\}} 	
\left\{  
t_j + d_{(0,j),(s_j,j)} 
\right\}   
\tag{$\textsc{NLP}^n_k(s,\rho)$}  \label{equa:nlp}   \\
s.t.: ~& t_{j+1} -t_j  \geq d_{(i,j+1),(i,j)}, ~~j\in \{0,\ldots,n\}, i\in\{0,\ldots,k\}  \label{eq: speed const} \\
& d_{(i,0),(s_l,l)} = 1, ~~l \in \{1,\ldots,n\}, i\in\{0,\ldots, k\} \label{eq: initial placement} \\
& d_{(s_j,j),(s_l,l)} = 2 \sinn{ mod(|\rho_j-\rho_l|,n) \cdot \tfrac{\pi}n  }, ~~j,l \in \{1,\ldots,n\} \label{eq: fixed dist const} \\
& t_0=-1, t_1 = 0 \label{eq: fixed time const} \\
& \left(\mathcal X^n_k,d \right)~\textrm{is a metric space}  \label{eq: metric const} \\
& \left(\mathcal X^n_k,d \right)~\textrm{is isometrically embeddable to}~ \left(\reals^2, \|\cdot\|_2 \right) \label{eq: emb const}
\end{align}
\ignore{
\label{equa:nlp}
{eq: speed const}
\label{eq: fixed dist const}
{eq: initial placement}
\label{eq: fixed time const}
\label{eq: metric const}
\label{eq: emb const}
}

\begin{lemma}[Introduced in \cite{GW24-iwoca}, for $k=1$]
\label{lem: nlp is opt s,rho algo}
Fix $s\in \mathcal B^n_k$ and $\rho \in \mathcal P_n$. 
Then the optimal value to~\ref{equa:nlp} equals the cost of the optimal $(s,\rho)$-algorithm to \pe{n}{k}.\footnote{For our lower bound arguments, we only need that the the optimal value to~\ref{equa:nlp} is a lower bound to the cost of the optimal $(s,\rho)$-algorithm.}
\end{lemma}

\begin{proof}
First we show how to find a feasible solution to~\ref{equa:nlp} starting from a feasible $(s,\rho)$-algorithm to \pe{n}{k}.
For this, consider an $(s,\rho)$-algorithm to \pe{n}{k} along with its corresponding trajectories $\{\tau_i\}_{i \in \{0,\ldots,k\}}$ (i.e. feasible unit speed trajectories). 
We design a feasible solution to~\ref{equa:nlp} with objective cost equal to $c$.
For convenience, we modify the definition of \pe{n}{k} so that all agents start from the centre of the $n$-gon, and then we subtract 1 from the total evacuation time. 
Since the $n$-gon is inscribed in a unit radius disk, and the first agent will need to reach the first vertex as soon as possible, the two problems are equivalent. 
This is done by setting $t_0=-1$, and $0$ in the index refers the $0$'th visited location, which in this case is the origin, hence all agents are equidistant from the $n$-gone vertices, satisfying~\eqref{eq: initial placement}.
More generally, the intended meaning of the variables is that $t_i$ equals the time that vertex $\rho_i$ (i.e. the $i$-th in time order visited vertex) is visited by agent $s_i\in \{0,\ldots,k\}$.
Moreover the execution of an algorithm produces movements for every placement of the exit. Since the exit is not known to the algorithm, two executions of the algorithm
are identical for as long as the exit is not visited yet. 

For this reason, we consider an execution of the algorithm in which only the last visited vertex is the one with the exit. 
At time $t_j$, we consider point $A_{i,j}\in \reals^2$ to be the positioning of agent $i$, when vertex $s_j$ is visited. 
Clearly we have that these points, equipped with the $\ell_2$ metric define a metric space, by definition, embeddable in $(\reals^2,\ell_2)$. 
Moreover, denote the distance of points $A_{i,j}, A_{i',j'}$ by $d_{(i,j),(i',j')}$. 

With these evaluations in mind, we observe that constraints~\eqref{eq: speed const} are satisfied because each agent is moving at speed at most 1. 
Constraint~\eqref{eq: fixed dist const} is satisfied because each vertex is visited by some agent, and the distance between the positions of the corresponding agents when they visit these vertices equals the distances between these vertices (in other words, because $A_{s_j,j} = \rho_j$). And indeed, $2 \sinn{ mod(|\rho_j-\rho_k|,n) \cdot \tfrac{\pi}n  }$ equals the distance of $n$-gon vertices $\rho_j,\rho_l$. 

Next observe that without loss of generality, we may assume that at the beginning of the execution of the search algorithm, some agent lies on some vertex (recall that the initial positioning of the agents is an algorithmic choice). 
If not, we can reset the clock when the first vertex is visited, resulting in lower termination cost. 
Therefore, constraint~\eqref{eq: fixed time const} is satisfied as well. As for constraints~\eqref{eq: emb const}, it is satisfied by construction, while constraint~\eqref{eq: metric const} is redundant. 

We are left to argue that the objective quantifies correctly the worst case cost. For this we observe that if the exit is placed at the $j$-th visited vertex, which by definition is done by agent $s_j$ at time $t_j$, then the Queen at point $A_{0,j}$ needs to travel $\|A_{0,j}-A_{s_j,j}\|_2=d_{(0,j),(s_j,j)}$ to go the exit. 
Hence, $t_j + d_{(0,j),(s_j,j)}$ measures correctly the running time of the search algorithm, and the objective of the NLP correctly considers the maximum, among all placements of the exit, as the overall termination time. 

Second we show that an optimal solution to~\ref{equa:nlp} corresponds to a feasible trajectory for \pe{n}{k}. 
For this, we note that constraints~\eqref{eq: fixed dist const}, \eqref{eq: metric const} and~\eqref{eq: emb const} identify an embedding of the regular $n$-gon on the 2-dimensional Cartesian plane, 
where also by~\eqref{eq: initial placement}, all agents are initially placed at its centre. 

We now design the trajectories for the $k+1$ agents, having them all starting from the centre of the $n$-gon, and then subtracting 1 from the evacuation cost 
(this is equivalent to allowing agents starting from anywhere within the unit disk). 
Using the embedding~\eqref{eq: metric const} implied by \eqref{eq: emb const}, we mark on the plane points that we identify as $(i,j)$, for each agent $i\in \{0,\ldots,k\}$ and each time interval $j\in\{0,\ldots,n\}$. 
Note, have one time interval for each vertex being discovered, and an additional one, for the initial placement of the agents who ``discover'' the origin. 
We let agent $i$ move along line segments with endpoints $(i,j)\rightarrow (i,j+1)$, $j=0,\ldots,n-1$. These trajectories visit indeed all vertices of the $n$-gon, 
where specifically, due to the definition of configuration $s\in \mathcal B^n_k$ and $\rho \in \mathcal P_n$ and constraint~\eqref{eq: fixed time const}. 
It remains to observe that~\eqref{eq: speed const} imposes that these trajectories are indeed at most unit speed, showing this way that the trajectories are feasible to~\pe{n}{k}.
Now observe that in an optimal solution, and for each $j=0,\ldots,n$, at least one of the constraints~\eqref{eq: speed const} has to be tight for some $i$, as otherwise one would be able to reduce the value of the objective. 
Effectively, this shows that each $t_i$ can be thought as the time that vertex $\rho_i$ is visited, hence also showing that the objective~\ref{equa:nlp} quantifies correctly the cost of the derived feasible trajectory. 
\end{proof}

We note that while formulation~\ref{equa:nlp} can be effectively coded as an NLP in any programming language, solving it optimally remains a challenge. 
Due to the non-convex nature of the program, unless a sophisticated method specifically tailored to the NLP is implemented, it is impossible to guarantee a global solution.
Indeed, solvers can only ensure that the returned solution is a local optimizer.
Motivated by that, we relax~\ref{equa:nlp} into a tracktable optimization problem. This idea is borrowed again from~\cite{GW24-iwoca}.

\begin{lemma}
\label{lem: relaxation is LP}
The relaxation of~\ref{equa:nlp} where constraint~\eqref{eq: emb const} is omitted (\emph{denoted as $\textsc{LP}^n_k(s,\rho)$}), is a Linear Program.
\end{lemma}
\begin{proof}
The linearization of the objective of~\ref{equa:nlp} is possible by the standard trick of replacing the objective by that of minimizing $y$, where $y$ is a brand new variable, along with the constraints
$$
y \geq t_j + d_{(0,j),(s_j,j)}, ~j=1,\ldots,n.
$$
Hence, it remains to show that the requirement that $\left(\mathcal X^n_k,d \right)$ is a metric space, is actually a polyhedron in variables $d$. This is indeed possible, since $\mathcal X^n_k$ is finite, and variable $d$ is indexed by $X^n_k \times X^n_k$. That is, we may think of $d: \mathcal X^n_k \times \mathcal X^n_k \mapsto \reals$ as a metric, hence identified by the properties that \\
- $d_{r_1,r_2} \geq 0$, for $r_1,r_2 \in \mathcal X^n_k$, and inequality is tight for  $r_1=r_2$, \\
- $d_{r_1,r_2} = d_{r_2,r_1}$, for $r_1,r_2 \in \mathcal X^n_k$ \\
- $d_{r_1,r_2}+d_{r_2,r_3} \geq d_{r_1,r_3}$, for $r_1,r_2,r_3 \in \mathcal X^n_k$.
\end{proof}

The advantage of considering $\textsc{LP}^n_k(s,\rho)$ is that the optimization problem can be solved efficiently, and moreover, every solution (even numerical and computer-based) comes with a certificate of global optimality, unlike what can be accomplished for the NLP. 
For given $s,\rho$, let $obj^n_k(s,\rho)$ be the optimal solution to $\textsc{LP}^n_k(s,\rho)$. 
Since the LP is a relaxation of~\ref{equa:nlp}, we see that $obj^n_k(s,\rho)$ is a lower bound on the optimal solution to \ref{equa:nlp}, which is also a lower bound on the evacuation cost of the optimal $(s,\rho)$ algorithm for \pe{n}{k}. 
Since also the optimal solution to \pe{n}{k} is some $(s,\rho)$-algorithm, we obtain the following. 

\begin{corollary}
\label{cor: lower bound to penk}
The optimal cost for solving \pe{n}{k} is at least $\min_{s,\rho} obj^n_k(s,\rho)$, where $obj^n_k(s,\rho)$ denotes the optimal solution to $\textsc{LP}^n_k(s,\rho)$. 
\end{corollary}

\section{Lower Bounds on \pe{n}{k}}
\label{sec: lower bounds to pe}

We present our findings on the lower bounds for the polygon priority evacuation problem \pe{n}{k} for $k=1,2,3,4$ and various values of $n$. 
This section serves as the proof of the lower bounds mentioned in Theorem~\ref{thm: n-gon k Servants upper and lower bounds}. 
For each $n$ and $k$, our results for \pe{n}{k} are obtained by applying Corollary~\ref{cor: lower bound to penk}, specifically by solving $\textsc{LP}^n_k(s,\rho)$ for all $s \in \mathcal B^n_k$ and $\rho \in \mathcal P_n$, and reporting the smallest value. Our work extends the results presented by~\cite{GW24-iwoca}, which covered only the case of $k=1$ and $n \leq 9$.

A brute force consideration of Corollary~\ref{cor: lower bound to penk} suggests that one needs to solve $c_{n,k}:=n! (k+1)^n$ linear programs, each with $\Omega\left( n^2k^2 \right)$ variables and $\Omega\left( n^3k^3 \right)$ constraints. The number of configurations presents a challenge, reaching $c_{9,1}=1.8 \cdot 10^8$ LPs for $n=9$ and $k=1$, the limit configuration considered in~\cite{GW24-iwoca}. Part of our contribution was to efficiently handle the $c_{n,k}$ configurations. The strategies implemented to reduce the number of configurations include:

\begin{itemize}
    \item Assuming that the first visited vertex is vertex 1 and that the second visited vertex is $j \leq \lceil n/2\rceil$.
    \item Omitting configurations that provably result in weak lower bounds, based on the necessity for agents to visit the vertices suggested by configuration $(s,\rho)$, and omitting the last evacuation step by the Queen. For this we use the following observation. 
    
    \begin{lemma}
\label{naive_algs}
    The cost of an optimal priority  algorithm on the $n$-gon with $k$ Servants does not exceed $\min\left\{1+\left(\lceil\tfrac{n}{k}\rceil-1\right)e_n,2+\left(\lceil\tfrac{n}{k+1}\rceil-1\right)e_n\right\}$.
\end{lemma}

\begin{proof}
    It suffices to provide two naive evacuation algorithms $\mathcal{A}_1$ and $\mathcal{A}_2$ with costs at most $1+\left(\lceil\tfrac{n}{k}\rceil-1\right)e_n$ and $2+\left(\lceil\tfrac{n}{k+1}\rceil-1\right)e_n$ respectively, where $e_n$ is the length of a regular $n$-gon inscribed in a unit radius circle. 

    For algorithm $\mathcal{A}_1$, the Queen starts at the origin and each of the $k$ Servants start at vertices of the $n$-gon, spaced as ``evenly'' as possible. 
    The Servants move clockwise along a shortest path to the next adjacent vertex in the $n$-gon until every vertex has been searched (and hence the exit is found). 
    Then, the Queen moves to the exit, which takes time $1$. No Servant will need to search more than $\lceil\tfrac{n}{k}\rceil$ distinct vertices of the $n$-gon, so the exit is found in time at most $\left(\lceil\tfrac{n}{k}\rceil-1\right)e_n$ (recall that each Servant searches a vertex at time $0$ and moving to an adjacent vertex can be done in time $e_n$). 
    Hence, the cost of $\mathcal{A}_1$ is at most $1+\left(\lceil\tfrac{n}{k}\rceil-1\right)e_n$.

    For algorithm $\mathcal{A}_2$, each agent (including the Queen) starts on a vertex of the $n$-gon, spaced as ``evenly'' as possible. 
    The agents move clockwise along a shortest path to the next adjacent vertex in the $n$-gon until every vertex has been searched. 
    Then, the Queen moves to the exit, which takes time at most $2$, since the worst case scenario is that the Queen is opposite the exit when it is found. 
    Since there are $k+1$ agents searching for the exit, no agent will need to search more than $\lceil\tfrac{n}{k+1}\rceil$ distinct vertices of the $n$-gon, so the exit is found in time at most $\left(\lceil\tfrac{n}{k+1}\rceil-1\right)e_n$. Hence, the cost of $\mathcal{A}_2$ is at most $2+\left(\lceil\tfrac{n}{k+1}\rceil-1\right)e_n$.
\end{proof}

Given certain sub-optimal permutations $\rho$ and identity strings $s$, Lemma~\ref{naive_algs} can be used to quickly determine their sub-optimality. In particular, if the trajectory of any one agent is of length greater than the cost of $\mathcal{A}_1$ or $\mathcal{A}_2$, then no $(s,\rho)$-algorithm can be optimal, and hence can be omitted when computing lower bounds. 
    
    \item Considering only strings $s \in \mathcal B^n_k$, for which if the Queen visits two vertices at consecutive time stamps, then no other Servant visits any vertex afterwords (since it is always better than Servants should visit their assigned vertices without breaks). 
    \item Considering that Servants are identical, thus making configurations invariant under the renaming of Servants.
\end{itemize}

We combined these strategies with a practical implementation for solving the LPs using Julia's JuMP Clp interface, an efficient open-source linear programming solver~\cite{Clp,dunning2017jump}. 
Each LP solution comes with a proof of optimality, quantified by a dual gap, ensuring accuracy to at least 10 decimal digits. 
Additionally, where applicable, we preset specific variables to known optimal values (e.g., the visiting times of initially visited vertices by distinct Servants are zero, as Servants start at polygon vertices without loss of generality). 
This technique significantly reduced the running time for many instances. Most importantly, we processed the different LPs using parallel threading, always tracking the tentatively lowest value achieved. For subsequent configurations, we solved the corresponding LP only if the cost of visiting the polygon vertices by the agents, as per the configuration (excluding the final evacuation cost), was lower than the tentative lowest value, further expediting the running time of our numerical calculations. 
These strategies allowed us to handle values of $(n,k)$ as high as $(13,1)$ or $(10,4)$, corresponding to $5.1 \cdot 10^{13}$ and $3.5 \cdot 10^{13}$ configurations, respectively. However, these methods were insufficient to address other values of $(n,k)$ such as $(10,3)$ with seemingly fewer $c_{10,3}= 3.8 \cdot 10^{12}$ configurations. We suspect that for certain values of $n$ and $k$, the naive bounds used to avoid solving certain LPs were not stringent enough.

In Tables~\ref{tab: lb ngons k=1}, \ref{tab: lb ngons k=2}, \ref{tab: lb ngons k=3} and \ref{tab: lb ngons k=4}
below, we summarize the findings for computing $l^n_k := \min_{s,\rho} obj^n_k(s,\rho)$, where $obj^n_k(s,\rho)$ denotes the optimal solution to $\textsc{LP}^n_k(s,\rho)$. 
Specifically, we use $l^n_k$ as the reported lower bound for \pe{n}{k}, as per Corollary~\ref{cor: lower bound to penk}. 
For small values of $n$ we occasionally observe values of the LP which are provably lower than the universal lower bound of $1$. We explain this later, but we note that this is compatible with that the LP is a relaxation to the corresponding NLPs. Whenever this happen, we report the correct lower bound of $1$. 
For each $n$ and $k$ for which this lower bound is reported, we also provide the minimizing configuration $(s,\rho)$, which is also used as guidance for our upper bound results later. 
We report the minimum valued LP with 15 decimal digits, thus providing a numerical evaluation for the lower bound $l^n_k$ for \pe{n}{k}. 
For many of these values, and due to some tight results presented next, we have identified concrete representation of many of these numbers. For instance:
$
\min_{s,\rho} obj^5_2(s,\rho)= \sqrt{(5+\sqrt{5})/2} \approx 1.90211303
$
and
$
\min_{s,\rho} obj^5_3(s,\rho)= \sqrt{5-2\sqrt{5}}/2 \approx 1.53884176.
$
However, many of these numerical values have impractically long representations, and for this reason, we choose not to include them in this exposition.

Finally, as evident from the numerical values in Theorem~\ref{thm: n-gon k Servants upper and lower bounds}, many of our lower bounds are matched or nearly matched with upper bounds. 
Hence, provably, for some configurations $(s,\rho)$ (specifically those determining the lower bound for \pe{n}{k}), we have that~\ref{equa:nlp} and $\textsc{LP}^n_k(s,\rho)$ have the same optimal solution, indicating that the relaxations exhibit no gap. Since the LP relaxation to the NLP was obtained as a metric embedding relaxation, by dropping the constraint of embeddability in $\reals^2$, and retaining the $\ell_2$ distances only between polygon vertices, one might be tempted to check whether the metric space of the optimal LP solutions is indeed embeddable in $(\reals^2, \ell_2)$. 
This can be done efficiently by solving a Positive Semidefinite Program that computes the minimum distortion required to embed any given metric into $(\reals^2, \ell_2)$, with a distortion equal to 1 corresponding to affirmative embeddability. Moreover, a distortion equal to 1 would actually provide a matching upper bound without explicitly describing the evacuation algorithm. However the minimum possible distortion in many cases is provably greater than $1$, even when the metric is restricted to the positioning of the Queen alone. Therefore, this phenomenon provides a solid explanation as to why in many cases we were unable to provide matching upper and lower bounds.

\begin{table}[h!]
\centering
\begin{small}
\begin{tabular}{c||cccc}
$n$		& 		$\min_{s,\rho} obj^n_1(s,\rho)$	&		$\rho$	&	$s$		\\
\hline
3			&	1.7320508075688772			&	(1,2,3)	& (1,0,1) \\
4			&	2.121320343559643			&	(1,2,4,3)	& (1,0,1,0) \\
5			&	2.7144122731725724			&	(1,4,5,3,2)	& (1,0,1,0,0) \\
6$^{\star}$		& 2.8660253779249727		&	(1,2,6,3,5,4)	& (1,0,1,0,1,0) \\
7$^{\star}$		& 2.95125017805582		&	(1,2,3,7,4,6,5)	& (0,1,1,0,1,0,1) \\
8$^{\star}$		& 3.003207375377086		&	(1,2,3,8,4,6,5,7)	& (0,1,1,0,1,0,1,0) \\
9$^{\star}$		& 3.218913730099321		&	(1,2,9,3,8,4,7,5,6)	& (1,0,1,0,1,0,1,0,1) \\
10			& 3.1871244937949434		&	(1,2,10,3,9,4,8,6,7,5)	& (1,0,1,0,1,0,1,0,1,0) \\
11			& 3.3557769107573536		&	(1,2,3,11,10,4,9,5,8,6,7)	& (1,0,0,1,1,0,1,0,1,0,0) \\
12			& 3.3848655006886306		&	(1,2,3,12,4,11,10,5,9,8,6,7)	& (1,0,0,1,0,1,1,0,1,1,0,1) \\
13			& 3.3636191025088142		&	(1,2,13,3,4,12,5,11,6,7,9,8,10)	& (0,1,0,1,1,0,1,0,1,1,0,1,0)  
\end{tabular}
\end{small}
\caption{Lower bounds for \pe{n}{1}, i.e., $n$-gon priority evacuation with 1 Servant. Values of $n$ labeled with $^{\star}$ indicate known evaluations reported previously in~\cite{GW24-iwoca}.}
\label{tab: lb ngons k=1}
\end{table}

\begin{table}[h!]
\centering
\begin{small}
\begin{tabular}{c||ccc}
$n$		& 		$\min_{s,\rho} obj^n_2(s,\rho)$	&		$\rho$	 &	$s$	\\
\hline
3	&	0.8660254037844386	&	(1,2,3)	& (1,2,0)	\\
4	&	1.7071067811865475	&	(1,2,3,4)	& (1,2,0,0)	\\
5	&	1.9021130325903068	&	(1,2,3,5,4)	& (1,2,0,1,2)	\\
6	&	1.9999999999999993 	&	(1,3,5,4,2,6)	& (1,2,0,2,1,0)	\\
7	&	2.0834826998207037	&	(1,3,7,2,6,4,5)	& (1,0,2,1,2,0,2)	\\
8	&	2.2378405106469064	&	(1,3,2,8,4,5,7,6)	& (1,2,0,1,2,2,1,0)	\\
9	&	2.3528883263148823	&	(1,2,5,3,9,6,4,8,7)	& (1,2,0,2,1,0,2,1,1)	\\
10	&	2.3781074994199956	&	(1,5,2,3,10,6,9,4,7,8)	& (1,0,2,2,1,0,1,2,0,1)	\\
11	&	2.4629185509183094	&	(1,2,9,3,11,8,10,4,6,5,7)	& (1,2,0,2,1,0,1,2,0,2,0)	
\end{tabular}
\end{small}
\caption{Lower bounds for \pe{n}{2}, i.e., $n$-gon priority evacuation with 2 Servants. All results are novel. }
\label{tab: lb ngons k=2}
\end{table}

\begin{table}[h!]
\centering
\begin{small}
\begin{tabular}{c||ccc}
$n$		& 		$\min_{s,\rho} obj^n_3(s,\rho)$	&		$\rho$ 	&	$s$	\\
\hline
3	&	0.8660254037844386	&	(1,2,3)	& (1,2,0)	\\
4	&	1.0		&	(1,2,3,4)	& (1,2,3,0)	\\
5	&	1.5388417685876266	&	(1,2,5,4,3)	& (1,2,3,0,2)	\\
6	&	1.866025403784437	&	(1,3,5,4,6,2)	& (1,2,3,2,3,1)	\\
7	&	1.8426953904169392	&	(1,3,6,2,4,7,5)	& (1,2,3,1,2,3,2)	\\
8	&	1.9134171618254483	&	(1,5,3,8,4,2,7,6)	& (1,2,3,1,2,3,1,0)	\\
9	&	1.8508331567966465	&	(1,2,4,6,5,3,9,8,7)	& (1,2,3,0,3,2,1,1,0)	
\end{tabular}
\end{small}
\caption{Lower bounds for \pe{n}{3}, i.e., $n$-gon priority evacuation with 3 Servants. All results are novel. }
\label{tab: lb ngons k=3}
\end{table}

\begin{table}[h!]
\centering
\begin{small}
\begin{tabular}{c||ccc}
$n$		& 		$\min_{s,\rho} obj^n_4(s,\rho)$	num	&		$\rho$		&	$s$	\\
\hline
3	&	0.8660254037844386	&	(1,2,3)	& (1,2,0)	\\
4	&	1.0 	&	(1,2,3,4)	& (1,2,3,0)	\\
5	&	0.9510565162951532	&	(1,2,5,4,3)	& (1,2,3,4,0)	\\
6	&	1.4999999999999991	&	(1,4,6,5,2,3)	& (1,2,3,4,1,2)	\\
7	&	1.6495989607031372	&	(1,3,7,2,5,4,6)	& (1,2,3,4,0,2,3)	\\
8	&	1.6892463972414653	&	(1,3,8,6,5,2,7,4)	& (1,2,3,4,4,2,3,4)	\\
9	&	1.6688480396635432	&	(1,3,5,9,4,2,8,6,7)	& (1,2,3,4,2,1,4,3,3)	\\
10	&	1.618033988749892	&	(1,4,5,9,3,6,10,2,7,8)	& (1,2,3,4,2,3,4,1,3,0)	
\end{tabular}
\end{small}
\caption{Lower bounds for \pe{n}{4}, i.e., $n$-gon priority evacuation with 4 Servants. All results are novel. }
\label{tab: lb ngons k=4}
\end{table}

\section{Upper Bounds on \pe{n}{k}}
\label{sec: upper bounds to pe}

For each $n,k$, we prove Theorem~\ref{thm: n-gon k Servants upper and lower bounds} by presenting feasible search trajectories to \pe{n}{k} for the $k+1$ agents and performing a worst case analysis. Due to the large number pairs $(n,k)$ for which we report results, we move all the details to Section~\ref{sec: algo details upper bounds polygon priority evacuation}. In this section we demonstrate our techniques by discussing the solution to~\pe{9}{1}, which is also one of the upper bounds that is provably optimal. 
For this, we also standarize the notation we will use for all other upper bounds that we present in the Appendix. 

When considering a solution to~\pe{n}{k}, we fix some $(s,\rho)\in \mathcal X^n_k $ and then we describe the $(s,\rho)$-algorithm together with its analysis based on a figure (depicting the trajectories) and a table giving trajectories' details and performing the worst case analysis. Next we describe the notation that is used in the figures and the tables, e.g. 
Figure~\ref{9gonk1fig} and Table~\ref{9gonk1tab}.

Recall that permutation element $\rho_i$ denotes the $i^{th}$ discovered vertex, which is visited by agent $s_i\in \{0,\ldots,k\}$. 
Servants move always at full speed $1$ over the prescribed vertices, and hence no further detail on their movements is required. 
Their whereabouts are irrelevant after they have visited all their assigned vertices. 
We denote the length of a regular $n$-gon edge inscribed in the unit radius disk as $e_n=2\sin(\pi/n)$, and we denote the origin by $\cal{O}$. 
Vertices of the $n$-gon where previously denoted as $V^n_j$, $j=1,\ldots,n$. 
In order to ease notation, we drop the superscript, whenever $n$ is clear from the context, and in our figures we simply write $V_j$.

We use the colour red to denote the trajectory of the Servant(s), while the colour blue is used to denote the trajectory of the Queen. 
Moreover, we use square vertices to denote the placements of the exits which induce the maximum evacuation time for any given search trajectory. 
We use a solid line to denote the trajectory the agents follow. 
We use a dashed line to show when the Queen deviates from the prescribed trajectory to evacuate at an announced exit that yields the maximum evacuation time. 

We let $t_i$ denote the first time when vertex $\rho_i$ is discovered by an agent.
We denote by $Q^{(i)}$ the location of the Queen at time $t_i$. 
We use segments with endpoints of the form $Q^{(i)},Q^{(i+s)}$ to indicate that the Queen has stayed put at point $Q^{(i+s)}$ from time $t_{i+1}$ and until $t_{i+s}$
Also, the tables inlcude $d(Q^{(i)},Exit)$ which is the distance between the Queen and exit, when the exit is placed at vertex $\rho_i$.

The trajectories of the $(s,\rho)$-algorithms are formally given in tables. 
For each $i \in \{1,\ldots,n\}$, we specify the location of the Queen at time $t_i$. Together with the Queen's distance to vertex $\rho_i$, they define the evacuation cost for the exit placement $\rho_i$. 
Whenever the positioning of the Queen is not a vertex, or a self-evident point, we define it formally either as an abstract expression, or as the solution to a formal non-linear system (similarly to how upper bounds have been described for all previous positive results to priority evacuation problems). 
In the reported trajectory tables, cost entries indicated by $*$ denote the worst case evacuation time for the provided trajectory, over all exit placements. 
Now we discuss the solution details to~\pe{9}{1}. 

\begin{example}[Upper Bound to \pe{9}{1}]
For \pe{9}{1}, we use $s=(1,0,1,0,1,0,1,0,1)$ and $\rho=(1,2,9,3,8,4,7,5,6)$, see also Figure~\ref{9gonk1fig} and Table~\ref{9gonk1tab}.
\input{UpperBoundsFigs/k1servants/9gonk1fig}
\input{UpperBoundsTabs/k1servants/9gonk1tab}

For the formal description of Queen's positions, we define points $Q^{(5)}$ and $Q^{(7)}$ as convex combinations of the points, $V_4$ and $V_8$, and $V_4$ and $V_7$, respectively. 
Thus, $Q^{(5)}=(1-l_1)V_4+l_1V_8$ and $Q^{(7)}=(1-l_2)V_4+l_2V_7$, for some $l_1,l_2\in[0,1]$. $l_1$ and $l_2$ are obtained as the solution to the following non-linear system
\begin{align}
&        l_1\| V_4 -V_8\|+l_2\| V_4-V_7\|=e_9     \label{k1n9eq1new} \\
&        (1-l_1)\| V_4 -V_8\|=l_1\| V_4 -V_8\|+(1-l_2)\| V_4-V_7\|+l_2\| V_4-V_7\|         \label{k1n9eq2new}
\end{align}
The solution to the above system satisfies $l_1\| V_4 -V_8\|\leq e_9$. 
\eqref{k1n9eq1new} is derived by imposing that the time taken by the Queen to travel from $Q^{(5)}$ to $V_4$ to $Q^{(7)}$ is exactly $e_9$. 
Moreover, \eqref{k1n9eq2new} is derived by equating the evacuation times for when the exit is placed at $V_8$ and $V_7$. It follows that $l_1=0.06031$ and $l_2=0.32635$.\\
\end{example}

\section{Improved Lower Bounds on \de{k} for $k=1,2$}
\label{sec: lower bounds to de}

This section provides the proof of our second main contribution, as stated in Theorem~\ref{thm: improved lower bounds summary}. 
The key element in this proof is the following lemma, first utilized in~\cite{czyzowicz2020priority123} and restricted to $k=1$. 

\begin{lemma}
\label{lem: n-gon lb gives disk lb}
The following is true for all $n \geq 3$. 
Let $l_{n,k}$ be a lower bound on \pe{n}{k}. 
Then, no algorithm for \de{k} has evacuation cost less than 
$1+ 2\pi/((k+1)n) + l_{n,k}$.
\end{lemma}

\begin{proof}
Fix $\epsilon>0$, along with an algorithm for \de{k}.
We let the search algorithm run for time $1+t - \epsilon$, where $t=2\pi/((k+1)n)$
We claim that there is an $n$-gon $P_n$ inscribed in the DISK, such that none of its vertices has been explored. 

Indeed, let $S$ denote the subset of the disk searched by the $k+1$ agents. Because each agent needs time 1 to go the perimeter, 
and no agent has searched more than $t - \epsilon$, we have that $|S| \leq (k+1)(t - \epsilon) $. 
 
Now consider an arbitrary placement (inscription) of the n-gon $P_n$ in the disk. 
For any $\phi\in [0,2\pi/n)$, consider the corresponding rotation of $ P_n$ by $\phi$, and define the set 
$$D=\{ \phi \in [0,2\pi/n): ~\textrm{at least one vertices of}~P_n\textrm{, rotated by}~\phi,~\textrm{has been visited}\}.$$
Note that all rotations $\phi \in [0,2\pi/n)$ correspond to distinct $n$-gons, and the trajectories of the $n$-gon vertices over these rotations cover disjoint parts of the perimeter of the disk. 
Hence, assuming that for every rotation $\phi \in [0,2\pi/n)$ of $P_n$, at least one vertices of that $P_n$ has been visited, we conclude that $|D|=2\pi/n$, and hence 
that for the searched space we have $|S|\geq 2\pi/n$. 
Together with that $|S| \leq (k+1)(t - \epsilon) $, we conclude that $2\pi/n \leq (k+1)(t-\epsilon)$, a contradiction for $t=2\pi/((k+1)n)$. 

Hence, we showed that for every $\epsilon>0$, that can also be chosen to be arbitrarily small, 
after the algorithm for \de{k} runs for time $1+4\pi/((k+1)n) - \epsilon$, there is a $n$-gon with the property that none of its vertices has been explored. 
The $n$-gon can be presented then to the algorithm, and hence, the additional time required to terminate is that of solving the evacuation problem for $P_n$. 
Overall, we showed that \de{k} requires time
$1+ 4\pi/((k+1)n) + l_{n,k}'-\epsilon$, for every $\epsilon>0$. 
\end{proof}

We can now prove Theorem~\ref{thm: improved lower bounds summary} using Lemma~\ref{lem: n-gon lb gives disk lb} and the lower bounds established by Theorem~\ref{thm: n-gon k Servants upper and lower bounds}. 
For this, we utilize the already derived bounds of \pe{n}{k} for $n \geq 6$, and we summarize the induced lower bounds for \de{k} in Table~\ref{tab: summary disk lb based on old method}, which are in agreement with those reported in Theorem~\ref{thm: improved lower bounds summary}. 
The star $*$ next to some lower bounds indicates an improvement upon the previous best results known. Bold text indicates the best result for a fixed $k$, and over all values of $n$ for which a lower bound to \pe{n}{k} was available. 
There is no improvement in the state-of-the-art with this method for $k=3,4$, but we still report the derived lower bounds for completeness. 

\begin{table}[h]
\centering
\begin{scriptsize}
\begin{tabular}{|c|c|c|c|c|c|c|c|c|c|}
\hline
\diagbox{$k$}{$n$} & 6 & 7 & 8 & 9 & 10 & 11 & 12 & 13 & Best Previous LB  \\ \hline
1 & 4.38962 & 4.40005 & 4.3959 & 4.56798 & 4.50128  & 4.64138 & \textbf{4.64666}$^*$ & 4.60528  & 4.56798~\cite{GW24-iwoca}  \\ \hline
2 &3.34907  & 3.38268 & 3.49964 & 3.5856 & 3.58755 & \textbf{3.65332}$^*$ & NA & NA & 3.6307~\cite{czyzowicz2020priority123}  \\ \hline
3 & \textbf{3.12782} & 3.06709  & 3.10977 & 3.02537 & 3.10814 & NA & NA & NA & 3.2017~\cite{czyzowicz2020priority123}  \\ \hline
4 & 2.70944 & 2.82912 & \textbf{2.84633} & 2.80847 & 2.74369 & NA & NA & NA & 2.91322~\cite{czyzowicz2020priority4} \\ \hline
\end{tabular}
\end{scriptsize}
\caption{Lower bounds on \de{k} as derived as immediate corollaries of Lemma~\ref{lem: n-gon lb gives disk lb} and Theorem~\ref{thm: improved lower bounds summary}.
Entries with NA correspond to values of $(n,k)$ for which no lower bound to \pe{n}{k} was derived, and hence these parameters are not applicable to our current argument. 
}
\label{tab: summary disk lb based on old method}
\end{table}

\section{Upper and Lower Bounds on \wpe{w}{n} \& Improved Lower Bounds on \wde{w}}
\label{sec: w-weighted results}

In this section, we describe the derivation of values in Figures~\ref{fig: wWeightedDisk} and~\ref{fig: wWeightedgon}, leading to Theorems~\ref{thm: new w-weighted lower bound} and~\ref{thm: upper and lower w-weighted ngon}. Setting $k=1$, we begin with a modification to \ref{equa:nlp} to model \wpe{w}{n} for all $n \geq 3$ and $w \in [0,1]$. Note that \wpe{0}{n} coincides with problem \pe{n}{1}. The feasible trajectories of the two agents are characterized by the same series of constraints. The objective of \wpe{w}{n}, as defined in Section~\ref{sec: definition and results}, is to minimize
$
\max_{j \in \{1,\ldots, n\}} \left\{ t_j + \tfrac{d_{(0,j),(s_j,j)} + w \cdot d_{(1,j),(s_j,j)}}{1+w} \right\}.
$
Importantly, the objective remains linear in variables $d$. Therefore, similar to our approach in Lemma~\ref{lem: relaxation is LP}, dropping constraint~\ref{eq: emb const} results in a Linear Program, whose optimal solution provides a lower bound to the optimal $(s,\rho)$-algorithm for \wpe{w}{n}. We denote this Linear Program as $\textsc{wLP}^n(s,\rho)$, and its optimal solution as $wobj^n(s,\rho)$. The following lemma was used in~\cite{GW24-iwoca} only for $n \leq 7$.
\begin{lemma}[\cite{GW24-iwoca}]
\label{lem: lb to disk from ngon w-search}
No algorithm for \wde{w} has a cost less than $1 + \frac{\pi}{n} + \min_{(s,\rho) \in \mathcal{X}^n_1} wobj^n(s,\rho)$, for all $n \geq 3$.
\end{lemma}

Figure~\ref{fig: wWeightedDisk} shows the known upper bound for \wde{w} and the known lower bounds obtained by applying Lemma~\ref{lem: lb to disk from ngon w-search} with $n=7$. As described in Section~\ref{sec: lower bounds to pe}, our contribution includes pushing the computational boundaries for computing $\min_{(s,\rho) \in \mathcal{X}^n_1} wobj^n(s,\rho)$ for larger $n$, reaching $n=12$ for the $w$-weighted search problem. We compute lower bounds on \wpe{w}{n} for $w$ starting from $0$ with a step size of $0.01$. These linear programs were solved using Julia's JuMP Clp solver~\cite{Clp,dunning2017jump}. The strongest lower bounds were obtained for $n=11$ and $n=12$, with different $n$-gon lower bounds dominating for different values of $w$. Figure~\ref{fig: wWeightedDisk} depicts the new and stronger lower bounds obtained by applying Lemma~\ref{lem: lb to disk from ngon w-search} for $n=11$ and $n=12$.

The new lower bounds are stronger than the straightforward lower bound of $1 + \pi$ reported in~\cite{GW24-iwoca} for all $w \in [0,0.7]$, and represent an improvement over the previously best lower bound known for the same range of $w$. More importantly, the gap between the previously best upper and lower bounds is now reduced by more than half for all $w \in [0,0.5]$. This discussion concludes the statements in Theorem~\ref{thm: new w-weighted lower bound}.

Next, we justify and motivate Theorem~\ref{thm: upper and lower w-weighted ngon}. The argument used to derive Theorem~\ref{thm: new w-weighted lower bound} relies on the computed lower bounds for \wpe{w}{n}, which were obtained using LP relaxations to the exact formulations for determining optimal $(s,\rho)$-algorithms. If these relaxations introduced a significant gap, the lower bound argument would not be tight, modulo the reduction proposed in Lemma~\ref{lem: lb to disk from ngon w-search}. Therefore, for $w \in [0,1]$ and $n=11,12$, we aim to find upper bounds for \wpe{w}{n}. If these bounds match the derived lower bounds, it would show that the lower bound analysis is tight, and further improvement would require a different technique or higher values of $n$.

Despite the lower bounds on $\min_{(s,\rho) \in \mathcal{X}^n_1} wobj^n(s,\rho)$ being obtained for various $(s,\rho) \in \mathcal{X}^n_1$ as $w$ ranged over $[0,1]$, computing upper bounds can be done using fixed configurations $(s,\rho) \in \mathcal{X}^n_1$. For $n=11$, we use $s=(1,0,0,1,1,0,1,0,1,0,1)$ and $\rho=(1,2,3,11,10,4,9,5,8,6,7)$, while for $n=12$, we use $s=(1,0,1,0,1,0,1,0,1,0,1,0)$ and $\rho=(1,2,12,3,11,4,10,5,9,6,8,7)$. For these configurations, we compute a numerical solution to the Non-Linear Program~\eqref{equa:nlp} using the objective of~\wpe{w}{n}, from $w=0$ with a step size of $0.02$. The numerical solutions were obtained using Julia's JuMP Ipopt~\cite{Ipopt,dunning2017jump}, which uses an interior point method for solving NLPs. Although these solutions are only guaranteed to be locally optimal, they correspond to feasible search trajectories (which may not be optimal). Nevertheless, in Figure~\ref{fig: wWeightedgon}, the derived upper and lower bounds are tight or nearly tight, similar to what we observed for \pe{n}{k}.

\section{Discussion}
\label{sec: discussion}

In this work, we made several lower bound improvements for well-known multi-agent search problems. Our results are obtained by novel upper and lower bounds, often tight or nearly tight, for searching hidden items where target locations form regular $n$-gons.
We extend existing techniques to larger $n$-gons and more agents, addressing significant computational challenges. Our contributions also lie in overcoming these challenges. The tightness of our results for $n$-gons show that further improvements in lower bounds on disk-related search problems require either increasing \(n\) or changing the lower bound techniques.
The tight bounds are based on solutions to LP relaxations inspired by metric embedding relaxations, which relax Non-Linear Programs (NLPs). We observed that these relaxations often impose no gap, even though the solutions may induce metric-related discrepancies. Investigating this phenomenon and strengthening the LPs to model optimal search algorithms more would be valuable.
Finally, our techniques apply to other search problems, especially those with asymmetric objectives. It remains to be seen if these techniques can be adapted for the face-to-face search problems, where good lower bounds still elude us.

\bibliographystyle{abbrv}
\bibliography{ngonevacuation}

\appendix

\section{Algorithmic Details for the Upper Bounds on \pe{n}{k} }
\label{sec: algo details upper bounds polygon priority evacuation}

\subsection{Polygon Priority Evacuation with $1$ Servant}
\label{sec: 1 servant upper bounds}

Before we give the upper bound details for \pe{n}{1}, we present all relevant figures. 
\input{UpperBoundsFigs/k1servants/ALLgonk1fig}

\subsubsection{\pe{3}{1} }
\label{sec: n3k1}
We use $s=(1,0,0)$ and $\rho=(1,2,3)$. 
The agents' trajectories show in Figure~\ref{3gonk1fig} and in Table~\ref{3gonk1tab}. 
\input{UpperBoundsTabs/k1servants/3gonk1tab}

\subsubsection{\pe{4}{1} }
\label{sec: n4k1}
We use $s=(1,0,1,0)$ and $\rho=(1,2,4,3)$. 
The agents' trajectories show in Figure~\ref{4gonk1fig} and in Table~\ref{4gonk1tab}. 
\input{UpperBoundsTabs/k1servants/4gonk1tab}

The evacuation time is maximized when the exit is placed at either $V_3$ or $V_4$. The point $Q^{(3)}$ is found by equating the evacuation times for both above placements of the exit. Moreover, $Q^{(3)}$ is given as the solution to the following non-linear system 
\begin{align}
&       \| V_3 -Q^{(3)}\| =\| V_4-Q^{(3)} \| 	        \label{k1n4eq1} \\
 &       \| V_2 -Q^{(3)}\|=e_4		        \label{k1n4eq2}
\end{align}
\eqref{k1n4eq1} is found by equating the evacuation times for when the exit is placed at $V_3$ and $V_4$. Similarly, \eqref{k1n4eq2} is derived from imposing that the distance the Queen travels from $Q^{(2)}$ to $Q^{(3)}$ is exactly $e_4$.

\subsubsection{\pe{5}{1} }
\label{sec: n5k1}
We use $s=(1,0,1,0,0)$ and $\rho=(1,4,5,3,2)$. 
The agents' trajectories show in Figure~\ref{5gonk1fig} and in Table~\ref{5gonk1tab}. 
Note that the Queen waits a $V_4$ for a time such that the Queen reaches $Q^{(3)}$ at time $e_5$.
\input{UpperBoundsTabs/k1servants/5gonk1tab}

The evacuation time is maximized when the exit is placed at either $V_2$ or $V_5$. The point $Q^{(3)}$ is found by equating the evacuation times for both above placements of the exit. Note that we require that $Q^{(3)}$ is a convex combination of the points $V_3$ and $V_5$, where $Q^{(3)}=(1-l)V_3+lV_5$, and $l\in [0,1]$. The value of $l$ can be found by solving the following equation $(1-2l)\| V_3 -V_5 \| = e_5$, asserting that the evacuation times for when the exit is placed at $V_5$ and $V_2$ have to be equal. It follows that $l=0.19098$\\

\subsubsection{\pe{6}{1} }
\label{sec: n6k1}
We use $s=(1,0,1,0,1,0)$ and $\rho=(1,2,6,3,5,4)$. 
The agents' trajectories show in Figure~\ref{6gonk1fig} and in Table~\ref{6gonk1tab}. 
The Queen waits at $Q^{(3)}$ until time $e_6$ has elapsed.
\input{UpperBoundsTabs/k1servants/6gonk1tab}

The evacuation time is maximized when the exit is placed at either $V_6$ or $V_5$. The point $Q^{(3)}$ and $Q^{(5)}$ are found by equating the evacuation times for both above placements of the exit. Moreover, we require that the points $Q^{(3)}$ and $Q^{(5)}$ are convex combinations of the points, $V_3$ and $V_6$, and $V_3$ and $V_5$, respectively. 
Hence, we require that $Q^{(3)}=(1-l_1)V_3 +l_1V_6$ and $Q^{(5)}=(1-l_2)V_3+l_2V_5$, for $l_1,l_2\in [0,1]$. 
We determine $l_1$ and $l_2$ by solving to the following system.
\begin{eqnarray}
&    (1-l_1) \| V_6-V_3 \|-(1-l_2)\| V_5-V_3 \|=e_6 	\label{k1n6eq1} \\
&    l_1 \| V_6-V_3 \| + l_2 \| V_5-V_3 \| =e_6	\label{k1n6eq2}
\end{eqnarray}
\eqref{k1n6eq1} is derived from equating the evacuation times for when the exit is placed at $V_6$ and $V_5$. Similarly, \eqref{k1n6eq2} is derived by imposing that the time taken by the Queen to travel from $Q^{(3)}$ to $V_3$ to $Q^{(5)}$ is $e_6$. It follows that $l_1=\tfrac{1}{2}(1-\tfrac{\sqrt{3}}{2})$ and $l_2=\tfrac{1}{2}$.

\subsubsection{\pe{7}{1} }
\label{sec: n7k1}
We use $s=(0,1,1,0,1,0,0)$ and $\rho=(1,2,3,7,4,6,5)$. 
The agents' trajectories show in Figure~\ref{7gonk1fig} and in Table~\ref{7gonk1tab}. 
\input{UpperBoundsTabs/k1servants/7gonk1tab}

The evacuation time is maximized when the exit is placed at either $V_4$ or $V_5$. The point $Q^{(5)}$ is found by equating the evacuation times for both above placements of the exit. The point $Q^{(5)}$ is given as the solution to the following non-linear system
    \begin{eqnarray}
&        \| Q^{(5)}-V_4 \| -\| Q^{(5)}-V_6\|=e_7	        \label{k1n7eq1} \\
&        \| Q^{(5)}-V_7 \| = e_7        \label{k1n7eq2}
    \end{eqnarray}
\eqref{k1n7eq1} is derived by equating the evacuation times for when the exit is placed at $V_4$ and $V_5$. 
Similarly, \eqref{k1n7eq2} is derived by imposing that the time taken for the Queen to travel from $V_7$ to $Q^{(5)}$ is exactly $e_7$.

\subsubsection{\pe{8}{1} }
\label{sec: n8k1}
We use $s=(0,1,1,0,1,0,1,0)$ and $\rho=(1,2,3,8,4,6,5,7)$. 
The agents' trajectories show in Figure~\ref{8gonk1fig} and in Table~\ref{8gonk1tab}. 
\input{UpperBoundsTabs/k1servants/8gonk1tab}
The evacuation time is maximized when the exit is placed at either $V_8$, $V_5$, or $V_7$. The points $Q^{(5)}$ and $Q^{(7)}$ are found by equating the evacuation times for all above placements of the exit. The points $Q^{(5)}$ and $Q^{(7)}$ are given as the solution to the following non-linear system
    \begin{eqnarray}
&        \| Q^{(5)}-V_8 \| = e_8		  \label{k1n8eq1}\\
&        \| Q^{(5)}-V_6 \| +\| Q^{(7)}-V_6 \| = e_8		    \label{k1n8eq2}\\
&      \| Q^{(7)}-V_7 \|-\| Q^{(7)}-V_5 \| =0		  \label{k1n8eq3} \\
&        \| Q^{(5)}-V_4 \|-\| Q^{(7)}-V_5 \|=e_8		    \label{k1n8eq4}
    \end{eqnarray}
\eqref{k1n8eq1} is derived by imposing that the time taken by the Queen to travel from $V_8$ to $Q^{(5)}$ is exactly $e_8$. 
\eqref{k1n8eq2} is derived by imposing that the time taken by the Queen to travel from $Q^{(5)}$ to $V_6$ to $Q^{(7)}$ is exactly $e_8$. 
\eqref{k1n8eq3} is derived by equating the evacuation times for when the exit is placed at $V_7$ and $V_5$. 
\eqref{k1n8eq4} is derived by equating the evacuation time for when the exit is placed at $V_4$ and $V_5$. \\

\subsubsection{\pe{9}{1} }
\label{sec: n9k1}
We use $s=(1,0,1,0,1,0,1,0,1)$ and $\rho=(1,2,9,3,8,4,7,5,6)$. 
The agents' trajectories show in Figure~\ref{9gonk1fig} and in Table~\ref{9gonk1tab}. 
The Queen waits at $V_3$ for an amount of time such that the Queen arrives at $Q^{(5)}$ at time $2e_9$.

\input{UpperBoundsTabs/k1servants/9gonk1tab}
The evacuation time is maximized when the exit is placed at either $V_8$ or $V_7$. The points $Q^{(5)}$ and $Q^{(7)}$ are found by equating the evacuation times for both above placements of the exit. The points $Q^{(5)}$ and $Q^{(7)}$ are given as convex combinations of the points, $V_4$ and $V_8$, and $V_4$ and $V_7$, respectively. 
Thus, we require that 
$Q^{(5)}=(1-l_1)V_4+l_1V_8$ and $Q^{(7)}=(1-l_2)V_4+l_2V_7$, for some $l_1,l_2\in [0,1]$. The values for $l_1$ and $l_2$ are given as the solution to the following non-linear system
    \begin{eqnarray}
&      l_1\| V_4 -V_8\|+l_2\| V_4-V_7\|=e_9	    \label{k1n9eq1}\\
&      (1-l_1)\| V_4 -V_8\|=l_1\| V_4 -V_8\|+(1-l_2)\| V_4-V_7\|+l_2\| V_4-V_7\|.		        \label{k1n9eq2}
    \end{eqnarray}
The solution to the system above satisfies $l_1\| V_4 -V_8\|\leq e_9$. 
Moreover, \eqref{k1n9eq1} is derived by imposing that the time taken by the Queen to travel from $Q^{(5)}$ to $V_4$ to $Q^{(7)}$ is exactly $e_9$. 
Similarly, \eqref{k1n9eq2} is derived by equating the evacuation times for when the exit is placed at $V_8$ and $V_7$.
It follows that $l_1=0.06031$ and $l_2=0.32635$.

\subsubsection{\pe{10}{1} }
\label{sec: n10k1}
We use $s=(1,0,1,0,1,0,1,0,1,0)$ and $\rho=(1,2,10,3,9,4,8,6,7,5)$. 
The agents' trajectories show in Figure~\ref{10gonk1fig} and in Table~\ref{10gonk1tab}. 
The Queen waits at $V_3$ for an amount of time such that the Queen arrives at $Q^{(5)}$ at time $2e_{10}$. The position of $Q^{(5)}$ has been adjusted for clarity.
\input{UpperBoundsTabs/k1servants/10gonk1tab}

The evacuation time is maximized when the exit is placed at either $V_9$ or $V_8$. The points $Q^{(5)}$ and $Q^{(7)}$ are found by equating the evacuation times for both above placements of the exit. The points $Q^{(5)}$ and $Q^{(7)}$ are given as the solution to the following non-linear system
    \begin{eqnarray}
&        l\| V_4-V_9\| + \| Q^{(7)}-V_4\|=e_{10}	        \label{k1n10eq1}\\
&      (1-l)\| V_4-V_9\|-\| Q^{(7)}-V_8\|=e_{10}.	      \label{k1n10eq2}\\
&	      \| V_3-((1-l)V_4+lV_9)\| + t =e_{10}		                    \label{k1n10eq3}\\
    \end{eqnarray}

\eqref{k1n10eq1} is derived by imposing that the time taken by the Queen to travel from $Q^{(5)}$ to $V_4$ to $Q^{(7)}$ is exactly $e_{10}$.
\eqref{k1n10eq2} is derived by imposing equating the evacuation times for when the exit is placed at $V_8$ and $V_9$.
\eqref{k1n10eq3} is derived by imposing that the Queen arrives at $Q^{(5)}$ at time $2e_{10}$. Note that $t$ denotes the total time the Queen spends waiting, given as $t=e_{10}-\| Q^{(5)}-V_3\|$.
It follows that $l=0.01029$.

\subsubsection{\pe{11}{1}}
\label{sec: n11k1}
We use $s=(1,0,0,1,1,0,1,0,1,0,0)$ and $\rho=(1,2,3,11,10,4,9,5,8,6,7)$. 
The agents' trajectories show in Figure~\ref{11gonk1fig} and in Table~\ref{11gonk1tab}. 
\input{UpperBoundsTabs/k1servants/11gonk1tab}

The evacuation time is maximized when the exit is placed at either $V_9$ or $V_8$. The points $Q^{(7)}$ and $Q^{(9)}$ are found by equating the evacuation times for both above placements of the exit. The points $Q^{(7)}$ and $Q^{(9)}$ are given by the solution to the following non-linear system. Note that we require that $Q^{(9)}$ is a convex combination of the points $V_8$ and $V_5$, where $Q^{(9)}=(1-l)V_5+lV_8$ and $l\in [0,1]$.
    \begin{eqnarray}
&       \| Q^{(7)}-V_4\|=e_{11}	        \label{k1n11eq1} \\
 &     \| Q^{(7)}-V_9\| - \| Q^{(7)}-V_5 \| = \| V_8 - V_5 \|	      \label{k1n11eq2} \\
&   l\| V_8 -V_5 \|+\| Q^{(7)}-V_5 \|=e_{11}	        \label{k1n11eq3}
    \end{eqnarray}
\eqref{k1n11eq1} is derived by imposing that the time taken by the Queen to travel from $V_4$ to $Q^{(7)}$ is exactly $e_{11}$. 
\eqref{k1n11eq2} is derived by equating the evacuation times for when the exit is placed at $V_9$ and $V_8$. 
\eqref{k1n11eq3} is derived by imposing that the time taken for the Queen to travel from $Q^{(7)}$ to $V_5$ to $Q^{(9)}$ is exactly $e_{11}$. 
It follows that $l=0.26872$.

\subsubsection{\pe{12}{1} }
\label{sec: n12k1}
We use $s=(1,0,0,1,0,1,1,0,1,1,0,1)$ and $\rho=(1,2,3,12,4,11,10,5,9,8,6,7)$. 
The agents' trajectories show in Figure~\ref{12gonk1fig} and in Table~\ref{12gonk1tab}. 
\input{UpperBoundsTabs/k1servants/12gonk1tab}

The evacuation time is maximized when the exit is placed at either $V_{10}$ or $V_9$. The points $Q^{(7)}$ and $Q^{(9)}$ are found by equating the evacuation times for both above placements of the exit. The points $Q^{(7)}$ and $Q^{(9)}$ are given as the solution to the following non-linear system.  We require that $Q^{(9)}$ is a convex combination of the points $V_5$ and $V_9$, hence $Q^{(9)}=(1-l)V_5+lV_9$, for some $l\in [0,1]$. 
    \begin{eqnarray}
&        \| Q^{(7)}-V_4 \| = e_{12}		        \label{k1n12eq1} \\
&      \| Q^{(7)}-V_5 \| + \| V_5-V_9 \| = \| Q^{(7)}-V_{10} \|	        \label{k1n12eq2} \\
&      \| V_5-Q^{(7)}\|+l\|V_5-V_9\| =e_{12}			        \label{k1n12eq3}
    \end{eqnarray}
\eqref{k1n12eq1} is derived by imposing that the time taken for the Queen to travel from $V_4$ to $Q^{(7)}$ is exactly $e_{12}$. 
Similarly, \eqref{k1n12eq2} is derived by equating the evacuation times for when the exit is placed at $V_9$ and $V_8$.
\eqref{k1n12eq3} is derived by imposing that the time take for the Queen to travel from $Q^{(7)}$ to $V_5$ to $Q^{(9)}$ is exactly $e_{12}$.
It follows that $l=0.20903$.

\subsubsection{\pe{13}{1} }
\label{sec: n13k1}
We use $s=(0,1,0,1,1,0,1,0,1,1,0,1,0)$ and $\rho=(1,2,13,3,4,12,5,11,6,7,9,8,10)$. 
The agents' trajectories show in Figure~\ref{12gonk1fig} and in Table~\ref{12gonk1tab}. 
Positions for $Q^{(7)}$ and $Q^{(12)}$ have been adjusted for clarity.
\input{UpperBoundsTabs/k1servants/13gonk1tab}

The evacuation time is maximized when the exit is placed at either $V_6$ or $V_5$. The points $Q^{(7)}$ and $Q^{(9)}$ are found by equating the evacuation times for both above placements of the exit. We require that $Q^{(7)}$ and $Q^{(9)}$ are convex combinations of the points, $V_5$ and $V_{11}$, and $V_6$ and $V_{11}$, respectively. Hence, we have 
$Q^{(7)}=(1-l_1)V_5 + l_1V_{11}$ and $Q^{(7)}=(1-l_2)V_6+l_2V_{11}$, for osme $l_1,l_2 \in [0,1]$. 
These values can be found by solving the following system of equations. 
    \begin{eqnarray}
&        (1-l_1)\| V_{11}-V_5\| + (1-l_2)\| V_{11}-V_6\| =e_{13}	        \label{k1n13eq1} \\
&       (1-l_1)\| V_{11}-V_5\| +\| V_{11}-V_6\|=l_1\| V_{11}-V_5\|		        \label{k1n13eq2}
    \end{eqnarray}
The solution to the system above satisfies $ \| V_{12}-Q^{(7)}\| \leq e_{13}$
\eqref{k1n13eq1} is derived by imposing that the time taken for the Queen to travel from $Q^{(7)}$ to $V_{11}$ to $Q^{(9)}$ is exactly $e_{13}$. 
\eqref{k1n13eq2} is derived by equating the evacuation times for when the exit is placed $V_5$ and $V_6$.
It follows that $l_1=0.97094$ and $l_2=0.77490$.\\

\subsection{Polygon Priority Evacuation with $2$ Servants}
\label{sec: 2 servant upper bounds}

Before we give the upper bound details for \pe{n}{2}, we present all relevant figures. 
\input{UpperBoundsFigs/k2servants/ALLgonk2fig}

\subsubsection{\pe{3}{2} }
\label{sec: n3k2}
We use $s=(1,2,0)$ and $\rho=(1,2,3)$. 
The agents' trajectories show in Figure~\ref{3gonk2fig} and in Table~\ref{3gonk2tab}. 
\input{UpperBoundsTabs/k2servants/3gonk2tab}

\subsubsection{\pe{4}{2} }
\label{sec: n4k2}
We use $s=(1,2,0,0)$ and $\rho=(1,2,3,4)$. 
The agents' trajectories show in Figure~\ref{4gonk2fig} and in Table~\ref{4gonk2tab}. 
\input{UpperBoundsTabs/k2servants/4gonk2tab}

The evacuation time is maximized when the exit is placed at either $V_1$ or $V_4$, and hence the following equation holds.
\begin{equation}
\|V_1-Q^{(1)}\|=\|V_3-Q^{(1)}\|+\|V_3-V_4\|
\end{equation}
Since the Queen immediately moves to either $V_1$ or $V_3$, the Queen's initial position can be expressed as a convex combination of these two points. That is, the Queen starts on the $y$-axis, so we have one variable (the Queen's initial $y$-coordinate) and one equation. The solution can be seen in the table, and the cost is approximately 1.707107.

\subsubsection{\pe{5}{2} }
\label{sec: n5k2}
We use $s=(1,2,0,1,2)$ and $\rho=(1,2,3,5,4)$. 
The agents' trajectories show in Figure~\ref{5gonk2fig} and in Table~\ref{5gonk2tab}. 
\input{UpperBoundsTabs/k2servants/5gonk2tab}
$Q^{(4)}$ is identified as the intersection of the lines $V_1V_4$ and $V_3V_5$. This is because the Queen wishes to be equidistant from both $V_4$ and $V_5$ at time $t_4$. This is when the exit is revealed, and the only options are that is it discovered by $S_1$ to be at $V_5$, or it is determined to be at the unique unsearched vertex $V_4$. Alternatively, $Q^{(4)}$ can be found by solving the following system of equations:
\begin{eqnarray*}
&\|V_3-Q^{(4)}\|=e_5\\
&\|V_3-V_5\|=\|V_3-Q^{(4)}\|+\|V_4-Q^{(4)}\|
\end{eqnarray*}
In the second equation, both sides are equal to the cost of the algorithm, which is approximately 1.902113.

\subsubsection{\pe{6}{2} }
\label{sec: n6k2}
We use $s=(1,2,0,1,2,0)$ and $\rho=(1,3,5,2,4,6)$. 
The agents' trajectories show in Figure~\ref{6gonk2fig} and in Table~\ref{6gonk2tab}. 
\input{UpperBoundsTabs/k2servants/6gonk2tab}
Here, a formal system of equations is unnecessary. The Queen should head from $V_6$ to the origin, at which point the exit is revealed to be either $V_2$, $V_4$, or $V_6$, and the Queen is equidistant from all of them.

\subsubsection{\pe{7}{2} }
\label{sec: n7k2}
We use $s=(1,0,2,1,2,0,0)$ and $\rho=(4,6,3,5,2,7,1)$. 
The agents' trajectories show in Figure~\ref{7gonk2fig} and in Table~\ref{7gonk2tab}. 
\input{UpperBoundsTabs/k2servants/7gonk2tab}
In this example, the Queen moves in such a way as to minimize the cost when the exit is at $V_2$, $V_5$, or $V_1$. Her second position $Q^{(5)}$ is equidistant from both $V_2$ and $V_5$ (i.e.~on the $x$-axis), and slightly to the left of $V_7$. Hence, $Q^{(5)}$ is a convex combination of $V_7$ and the origin.

The distance between $Q^{(5)}$ and $V_7$ is tuned so that the evacuation time is equal to the cost when the exit is at either $V_2$, $V_5$ or $V_1$. Also note that the Queen must be at $Q^{(5)}$ at time $e_7$, since this is when $V_2$ and $V_5$ are visited for the first time, however the distance from $V_6$ to $Q^{(5)}$ is less than $e_7$. We have the equation
\begin{equation}
e_7+\|V_2-Q^{(5)}\|=2e_7+\|V_7-Q^{(5)}\|
\end{equation}
where both sides represent the cost of the algorithm, which is approximately 2.140269. The right hand side is the evacuation time when the exit is at $V_1$.
Solving the equation, we get that the $x$-coordinate of $Q^{(5)}$ as $x_{2,7}=\tfrac{2\sin(\pi/7)(1+\sin(\pi/7))}{1+2\sin(\pi/7)+\sin(\pi/14)} \approx 0.595266$.

\subsubsection{\pe{8}{2} }
\label{sec: n8k2}
We use $s=(1,2,0,1,2,1,2,0)$ and $\rho=(1,3,2,8,4,7,5,6)$. 
The agents' trajectories show in Figure~\ref{8gonk2fig} and in Table~\ref{8gonk2tab}. 
\input{UpperBoundsTabs/k2servants/8gonk2tab}
In this example, the Queen's best strategy is to move to $V_6$ as fast as possible while remaining equidistant from each Servant (i.e.~on the $y$-axis). Hence, both $Q^(5)$ and $Q^(7)$ are convex combinations of $V_2$ and $V_6$. The distances $\|V_4-Q^{(5)}\|$ and $\|V_5-Q^{(7)}\|$ are $w_{2,8}=\sqrt{4-\sqrt{2}-2\sqrt{2-\sqrt{2}}} \approx 1.027158$ and $z_{2,8}=\sqrt{10-3\sqrt{2}-(4+2\sqrt{2})\sqrt{2-\sqrt{2}}} \approx 0.728771$ respectively. The cost of the algorithm is approximately 2.259505.

\subsubsection{\pe{9}{2} }
\label{sec: n9k2}
We use $s=(1,2,0,1,2,0,1,2,0)$ and $\rho=(4,5,8,3,6,9,2,7,1)$. 
The agents' trajectories show in Figure~\ref{9gonk2fig} and in Table~\ref{9gonk2tab}. 
\input{UpperBoundsTabs/k2servants/9gonk2tab}
In this example, the worst case cost occurs when the exit is placed at either $V_3$/$V_6$ or $V_2$/$V_7$. At the times these pairs of vertices are visited, the Queen wishes to be equidistant from each vertex, and hence $Q^{(5)}$ and $Q^{(8)}$ are both on the $x$-axis (i.e.~convex combinations of $V_9$ and the origin). The $x$-coordinates of $Q^{(5)}$ and $Q^{(8)}$ are tuned so that the evacuation times are equal when the exit is at $V_3$/$V_6$ or $V_2$/$V_7$. We therefore have the following system of equations. Note that the first equation ensures that the Queen moves at unit speed from $Q^{(5)}$ to $V_9$ to $Q^{(8)}$, and in the second equation, both sides are equal to the cost. Also note that the Queen is at $Q^{(5)}$ at time $e_9$ even though $\|V_8-Q^{(5)}\|<e_9$, and that the worst case cost does \emph{not} occur when the exit is placed at $V_1$.
\begin{eqnarray}
& \|V_9-Q^{(5)}\|+\|V_9-Q^{(8)}\|=e_9\\
& e_9+\|V_3-Q^{(5)}\|=2e_9+\|V_2-Q^{(8)}\|
\end{eqnarray}
The solution is $Q^{(5)}=(x_{2,9},0)$ and $Q^{(8)}=(y_{2,9},0)$, where $x_{2,9}\approx 0.948587$ and $y_{2,9}\approx 0.367373$. 
The cost of the algorithm is approximately $2.371762$.

\subsubsection{\pe{10}{2} }
\label{sec: n10k2}
We use $s=(1,2,0,1,2,0,1,2,0,0)$ and $\rho=(5,6,9,4,7,10,3,8,1,2)$. 
The agents' trajectories show in Figure~\ref{10gonk2fig} and in Table~\ref{10gonk2tab}. 
\input{UpperBoundsTabs/k2servants/10gonk2tab}
The worst case cost occurs when the exit is at $V_4$, $V_3$, or $V_8$. The positions of the two unknowns $Q^{(5)}$ and $Q^{(8)}$ can be found by solving the following system of equations, where both sides of the final equation are the cost. Note that we have four equations and four unknowns, since $Q^{(5)}$ and $Q^{(8)}$ each have two unknown coordinates.
\begin{eqnarray}
& \|V_9-Q^{(5)}\|=e_{10}\\
& \|V_{10}-Q^{(5)}\|+\|V_{10}-Q^{(8)}\|=e_{10}\\
& \|V_3-Q^{(8)}\|=\|V_8-Q^{(8)}\|\\
& e_{10}+\|V_4-Q^{(5)}\|=2e_{10}+\|V_3-Q^{(8)}\|
\end{eqnarray}

The solution can be seen in the table, and the cost is approximately $2.389560$.

\subsubsection{\pe{11}{2} }
\label{sec: n11k2}
We use $s=(1,2,0,1,2,0,1,2,0,2,0)$ and $\rho=(5,6,2,4,7,1,3,8,10,9,11)$. 
The agents' trajectories show in Figure~\ref{11gonk2fig} and in Table~\ref{11gonk2tab}. They were found using Julia's JuMP Ipopt NLP solver.
\input{UpperBoundsTabs/k2servants/11gonk2tab}

\subsection{Polygon Priority Evacuation with $3$ Servants}
\label{sec: 3 servant upper bounds}

Before we give the upper bound details for \pe{n}{3}, we present all relevant figures. 
Note that we omit the algorithm for \pe{3}{3}, which is straightforward. 
\input{UpperBoundsFigs/k3servants/ALLgonk3fig}


\subsubsection{\pe{4}{3} }
\label{sec: n4k3}
We use $s=(1,2,3,0)$ and $\rho=(1,2,3,4)$. 
The agents' trajectories show in Figure~\ref{4gonk3fig} and in Table~\ref{4gonk3tab}. 
\input{UpperBoundsTabs/k3servants/4gonk3tab}

\subsubsection{\pe{5}{3} }
\label{sec: n5k3}
We use $s=(1,2,3,0,2)$ and $\rho=(1,2,3,5,4)$. 
The agents' trajectories show in Figure~\ref{5gonk3fig} and in Table~\ref{5gonk3tab}. 
\input{UpperBoundsTabs/k3servants/5gonk3tab}

In this example the evacuation time is maximized when the exit is placed at either $V_2$/$V_3$ or $V_4$. The Queen starts on the $x$-axis so as to be equidistant from both $V_2$ and $V_3$, so $Q^{(3)}$ is a convex combination of $V_5$ and the origin. The Queen's initial $x$-coordinate is tuned so that the evacuation times are equal when the exit is placed at $V_2$/$V_3$ or $V_4$. We therefore have the following equation, in which both sides represent the cost of the algorithm.
\begin{equation}
\|V_2-Q^{(3)}\|=\|V_5-Q^{(3)}\|+e_5
\end{equation}
We can solve the equation to get that $Q^{(3)}=(x_{3,5},0)$, where $x_{3,5}=\tfrac{2\sqrt{10-2\sqrt{5}}+5-\sqrt{5}}{2\sqrt{10-2\sqrt{5}}+5+\sqrt{5}}\approx0.625397$. Also write $y_{3,5}=1-x_{3,5}\approx0.374603$. The cost of the algorithm is approximately $1.550174$.

\subsubsection{\pe{6}{3} }
\label{sec: n6k3}
We use $s=(1,2,3,2,3,1)$ and $\rho=(1,3,5,4,6,2)$. 
The agents' trajectories show in Figure~\ref{6gonk3fig} and in Table~\ref{6gonk3tab}. 
\input{UpperBoundsTabs/k3servants/6gonk3tab}

The evacuation time is maximized when the exit is placed at either $V_6$, $V_4$ or $V_2$. The point $Q^{(6)}$ is found by equating the evacuation times for all above placements of the exit.
The point $Q^{(6)}$ is given as the solution to the following non-linear system.
    \begin{eqnarray}
&        \| V_2-Q^{(6)}\| - \| V_4-Q^{(6)}\| = 0	        \label{k3n6eq1}	\\
&        \| V_2-Q^{(6)}\| - \| V_6-Q^{(6)}\| = 0	        \label{k3n6eq2}
    \end{eqnarray}
\eqref{k3n6eq1} is derived by equating the evacuation times for the when the exit is placed at $V_2$ and $V_4$. 
Similarly, \eqref{k3n6eq2} is derived from equating the evacuation times for when the exit is placed at $V_2$ and $V_6$.

\subsubsection{\pe{7}{3} }
\label{sec: n7k3}
We use $s=(1,2,3,1,2,3,2)$ and $\rho=(1,3,6,2,4,7,5)$. 
The agents' trajectories show in Figure~\ref{7gonk3fig} and in Table~\ref{7gonk3tab}. The Queen's strategy is very simple; she waits at the origin until $V_2$, $V_4$, and $V_7$ are searched, at which time the exit is revealed, and she heads there immediately.
\input{UpperBoundsTabs/k3servants/7gonk3tab}

\subsubsection{\pe{8}{3} }
\label{sec: n8k3}
We use $s=(1,2,3,1,2,3,1,0)$ and $\rho=(1,5,3,8,4,2,7,6)$. 
The agents' trajectories show in Figure~\ref{8gonk3fig} and in Table~\ref{8gonk3tab}. 
\input{UpperBoundsTabs/k3servants/8gonk3tab}

The evacuation time is maximized when the exit is placed at either $V_7$ or $V_6$. The point $Q^{(7)}$ is found by equating the evacuation times for both above placements of the exit.
We define $Q^{(7)}$ as a convex convex combination of the points $V_6$ and $V_7$, hence $Q^{(7)}=(1-l)V_6+lV_7$ for some $l\in [0,1]$. 
The value of $l$ is the solution to the non-linear equation.
$(1-l)\| V_7-V_6\|-l\| V_7-V_6\| = 0,$
which is derived by equating the evacuation times for when the exit is placed at $V_6$ and $V_7$. It follows that $l=\frac{1}{2}$.

\subsubsection{\pe{9}{3} }
\label{sec: n9k3}
We use $s=(1,2,3,0,1,2,3,1,0)$ and $\rho=(4,5,7,9,3,6,8,2,1)$. 
The agents' trajectories show in Figure~\ref{9gonk3fig} and in Table~\ref{9gonk3tab}. 
\input{UpperBoundsTabs/k3servants/9gonk3tab}
In this example the evacuation time is maximized when the exit is placed at $V_3$/$V_6$ or $V_4$/$V_5$. The Queen must therefore be equidistant from $V_4$ and $V_5$ at time $0$, and equidistant from $V_3$ and $V_6$ at time $e_9$. Hence, $Q^{(3)}$ and $Q^{(7)}$ both lie on the $x$-axis (i.e.~they are convex combinations of $V_9$ and the origin). The distance from the Queen's initial position $Q^{(3)}$ to $V_9$ is tuned so that the evacuation times are equal when the exit is placed at $V_3$/$V_6$ or $V_4$/$V_5$, assuming she moves at unit speed from $Q^{(3)}$ to $V_9$ to $Q^{(7)}$. 

In the following system of equations, both sides of the last equation are the cost of the algorithm. The first equation ensures the Queen moves at unit speed from $Q^{(3)}$ to $V_9$ to $Q^{(7)}$. 
\begin{eqnarray}
&\|V_{9}-Q^{(3)}\|+\|V_{9}-Q^{(7)}\|=e_{9}\\
&\|V_4-Q^{(3)}\|=e_{9}+\|V_3-Q^{(7)}\|
\end{eqnarray}

Solving this system we get $Q^{(3)}=(x_{3,9},0)$ where $x_{3,9}\approx 0.943113$, and $Q^{(7)}=(y_{3,9},0)$ where $y_{3,9}=2-e_9-x_{3,9}\approx 0.372847$. The exact position of $Q^{(4)}$ is not as important, as there are several placements for which the evacuation time is less than the cost when the exit is placed at $V_1$ or $V_2$. We therefore arbitrarily decide for $Q^{(4)}$ to be distance $e_9$ from $Q^{(3)}$ and equidistant from both $V_1$ and $V_2$. This yields $Q^{(4)}\approx (0.394729,0.683691)$. Since the cost occurs when the exit is placed at $V_4$, the cost is $\|V_4-Q^{(3)}\|=\sqrt{\left(x_{3,9}-\cos\left(\tfrac{8\pi}{9}\right)\right)^2+\sin^2\left(\tfrac{8\pi}{9}\right)}\approx1.913618$.

\subsection{Polygon Priority Evacuation with $4$ Servants}
\label{sec: 4 servant upper bounds}

Before we give the upper bound details for \pe{n}{4}, we present all relevant figures. 
Note that we omit the algorithms for \pe{3}{4} and \pe{4}{4}, which are straightforward. 

\input{UpperBoundsFigs/k4servants/ALLgonk4fig}


\subsubsection{\pe{5}{4} }
We use $s=(1,2,3,4,0)$ and $\rho=(1,2,3,4,5)$. 
The agents' trajectories show in Figure~\ref{5gonk4fig} and in Table~\ref{5gonk4tab}. 
\input{UpperBoundsTabs/k4servants/5gonk4tab}

\subsubsection{\pe{6}{4} }
\label{sec: n6k4}
We use $s=(1,2,3,4,1,2)$ and $\rho=(1,4,6,5,2,3)$. 
The agents' trajectories show in Figure~\ref{6gonk4fig} and in Table~\ref{6gonk4tab}. The Queen moves to the midpoint of $V_2$ and $V_3$ and waits there until time $e_6$, at which point the exit is revealed and she moves there immediately.
\input{UpperBoundsTabs/k4servants/6gonk4tab}

\subsubsection{\pe{7}{4} }
\label{sec: n7k4}
We use $s=(1,2,3,4,0,2,3)$ and $\rho=(3,5,2,4,7,6,1)$. 
The agents' trajectories show in Figure~\ref{6gonk4fig} and in Table~\ref{6gonk4tab}. 
\input{UpperBoundsTabs/k4servants/7gonk4tab}
In this example the evacuation time is maximized when the exit is placed at either $V_1$ or $V_6$. In order to minimize the cost, the Queen should be equidistant and as close as possible to both vertices at the time they are visited, i.e.~at time $e_7$. Hence, $Q^{(7)}=\left(\cos\left(\tfrac{2\pi}{7}\right),0\right)$. There are many placements of the Queen's first position $Q^{(4)}$ so that the cost does not occur if the exit is at either $V_3$ or $V_4$. We therefore arbitrarily decide that $Q^{(4)}$ lies on the $x$-axis, and that the following equation holds, which ensures she moves at unit speed from $Q^{(4)}$ to $V_7$ to $Q^{(7)}$:
 \begin{equation}\|V_7-Q^{(4)}\|+\|V_7-Q^{(7)}\|=e_7
 \end{equation}

Using this equation and the known value of $Q^{(7)}$, we can solve for the Queen's initial position, which is $Q^{(4)}=(x_{4,7},0)$ where $x_{4,7}=1-2\sin\left(\tfrac{\pi}{7}\right)+2\sin^2\left(\tfrac{\pi}{7}\right)\approx 0.508743$. The cost is approximately $1.649599$.

\subsubsection{\pe{8}{4} }
\label{sec: n8k4}
We use $s=(1,2,3,4,2,3,4,0)$ and $\rho=(1,3,8,6,2,7,5,4)$. 
The agents' trajectories show in Figure~\ref{8gonk4fig} and in Table~\ref{8gonk4tab}. In this example, the Queen's strategy is to wait at the origin until time $e_8$, at which point the exit is revealed and she moves there immediately. The cost is approximately 1.765367.
\input{UpperBoundsTabs/k4servants/8gonk4tab}

\subsubsection{\pe{9}{4} }
\label{sec: n9k4}
We use $s=(1,2,3,4,1,2,3,4,0)$ and $\rho=(1,3,5,7,2,4,6,8,9)$. 
The agents' trajectories show in Figure~\ref{9gonk4fig} and in Table~\ref{9gonk4tab}. This is another example where the Queen's strategy is quite simple. She waits at the origin until time $e_9$, at which point the exit is revealed and she moves there immediately. The cost is approximately 1.684040.
\input{UpperBoundsTabs/k4servants/9gonk4tab}

\subsubsection{\pe{10}{4} }
\label{sec: n10k4}

\input{UpperBoundsTabs/k4servants/10gonk4tab}
We use $s=(1,2,3,4,1,2,3,4,3,0)$ and $\rho=(2,3,5,10,1,4,6,9,7,8)$. The agents' trajectories show in Figure~\ref{10gonk4fig} and in Table~\ref{10gonk4tab}. In this example, the evacuation time is maximized when the exit is at either $V_1$/$V_4$ or $V_7$/$V_8$. The Queen waits at some position $(0,y_{4,10})$ until time $e_{10}$ when $V_1$ and $V_4$ are visited. If the exit is found at either of these vertices, the Queen heads there immediately. Otherwise, she moves downward immediately so that she is as close to $V_7$/$V_8$ as possible at time $2e_{10}$. Thus, at time $2e_{10}$ she is at $(0,y_{4,10}-e_{10})$. The value of $y_{4,10}$ is tuned so that the evacuation times are equal when the exit is at either $V_1$/$V_4$ or $V_7$/$V_8$, and hence we have the following equation with a single unknown:

\begin{equation}
e_{10}+\| V_1-(0,y_{4,10})\|=2e_{10}+\| V_7-(0,y_{4,10}-e_{10})\|
\end{equation}

We can solve this equation to get $y_{4,10}\approx -0.0553293$, and the cost is approximately $1.651526$.


\end{document}

%% file: UpperBoundsFigs/k1servants/9gonk1fig.tex
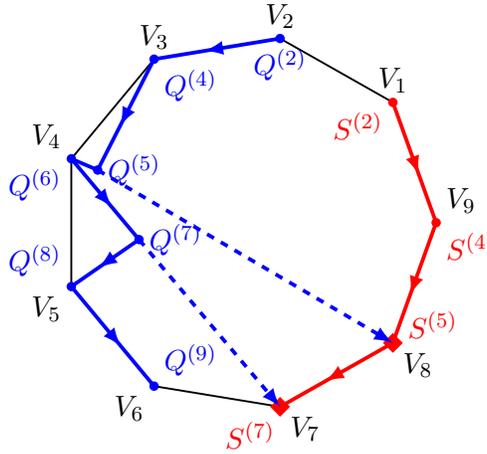
\begin{figure}[h!]
    \centering
    \begin{tikzpicture}[scale=0.25]
    \tikzset{
        mid arrow/.style={
            postaction={decorate, decoration={markings, mark=at position 0.6 with {\arrow{latex}}}}
        }
    }
        \node (a) at (7.7,6.4) {};
        \fill [fill=black] (a) circle (0.105) node [above] {$V_1$};
        \node (b) at (1.7,9.8) {};
        \fill [fill=black] (b) circle (0.105) node [above] {$V_2$};
        \node (c) at (-5,8.7) {};
        \fill [fill=black] (c) circle (0.105) node [above] {$V_3$};
        \node (d) at (-9.4,3.4) {};
        \fill [fill=black] (d) circle (0.105) node [above left] {$V_4$};
        \node (e) at (-9.4,-3.4) {};
        \fill [fill=black] (e) circle (0.105) node [below left] {$V_5$};
        \node (f) at (-5,-8.7) {};
        \fill [fill=black] (f) circle (0.105) node [below left] {$V_6$};
        \node (g) at (1.7,-9.8) {};
        \fill [fill=black] (g) circle (0.105) node [below right] {$V_7$};
        \node (h) at (7.7,-6.4) {};
        \fill [fill=black] (h) circle (0.105) node [below right] {$V_8$};
        \node (i) at (10,0) {};
        \fill [fill=black] (i) circle (0.105) node [above right] {$V_9$};

        \draw [line width=0.25mm, black] (10,0)--(7.7,6.4)--(1.7,9.8)--(-5,8.7)--(-9.4,3.4)--(-9.4,-3.4)--(-5,-8.7)--(1.7,-9.8)--(7.7,-6.4)--(10,0);

        \node at (a) (7.7,6.4) {};
        \fill [fill=red] (a) circle (0.25) node [below left] {\textcolor{red}{$S^{(2)}$}};
        \node at (i) (10,0) {};
        \fill [fill=red] (i) circle (0.25) node [below right] {\textcolor{red}{$S^{(4)}$}};
        \node[diamond, draw=red, fill=red, minimum size=0.25cm, inner sep=0] (h) at (7.7,-6.4) {};
        \node[above right] at (h.south east) {\textcolor{red}{$S^{(5)}$}};
        \node[diamond, draw=red, fill=red, minimum size=0.25cm, inner sep=0] (g) at (1.7,-9.8) {};
        \node[below left] at (g.south east) {\textcolor{red}{$S^{(7)}$}};
        \draw [line width=0.5mm, red, mid arrow] (7.7,6.4) -- (10,0);
        \draw [line width=0.5mm, red, mid arrow] (10,0) -- (7.7,-6.4);
        \draw [line width=0.5mm, red, mid arrow] (7.7,-6.4) -- (1.7,-9.8);

        \node at (b) (1.7,9.8) {};
        \fill [fill=blue] (b) circle (0.25) node [below] {\textcolor{blue}{$Q^{(2)}$}};
        \node at (c) (-5,8.7) {};
        \fill [fill=blue] (c) circle (0.25) node [below right] {\textcolor{blue}{$Q^{(4)}$}};
        \node (j) at (-8,2.8) {};
        \fill [fill=blue] (j) circle (0.25) node [right] {\textcolor{blue}{$Q^{(5)}$}};
        \node at (d) (-9.4,3.4) {};
        \fill [fill=blue] (d) circle (0.25) node [below left] {\textcolor{blue}{$Q^{(6)}$}};
        \node (k) at (-5.8,-0.9) {};
        \fill [fill=blue] (k) circle (0.25) node [right] {\textcolor{blue}{$Q^{(7)}$}};
        \node at (e) (-9.4,-3.4) {};
        \fill [fill=blue] (e) circle (0.25) node [above left] {\textcolor{blue}{$Q^{(8)}$}};
        \node at (f) (-5,-8.7) {};
        \fill [fill=blue] (f) circle (0.25) node [above right] {\textcolor{blue}{$Q^{(9)}$}};
        \draw [line width=0.5mm, blue, mid arrow] (1.7,9.8) -- (-5,8.7);
        \draw [line width=0.5mm, blue, mid arrow] (-5,8.7) -- (-8,2.8);
        \draw [line width=0.5mm, blue] (-8,2.8)  -- (-9.4,3.4) node[currarrow,blue,pos=0.55,xscale=-1,sloped,scale=0.75] {};
        \draw [line width=0.5mm, blue, mid arrow] (-9.4,3.4) -- (-5.8,-0.9);
        \draw [line width=0.5mm, blue, mid arrow] (-5.8,-0.9) -- (-9.4,-3.4);
        \draw [line width=0.5mm, blue, mid arrow] (-9.4,-3.4) -- (-5,-8.7);
        \draw[dashed] [line width=0.5mm, blue] [-latex] (-8,2.8) -- (7.7,-6.4);
        \draw[dashed] [line width=0.5mm, blue] [-latex] (-5.8,-0.9) -- (1.7,-9.8);
    \end{tikzpicture}
    \caption{A search trajectory for \pe{9}{1}. 
Trajectory details can be found in Table~\ref{9gonk1tab}.
The Queen waits at $V_3$ for an amount of time such that the Queen arrives at $Q^{(5)}$ at time $2e_9$. The position of $Q^{(5)}$ has been adjusted for clarity.}
\label{9gonk1fig}
\end{figure}

%% file: UpperBoundsTabs/k1servants/9gonk1tab.tex
\begin{table}[h!]
\centering
\begin{small}
    \begin{tabular}{|c|c|c|c|c|c|c|}
        \hline
        $i$ & Agent & Vertex & $t_i$ & $Q^{(i)}$ & $d(Q^{(i)} , Exit)$ & Cost \\
        \hline
        1 & $S_1$ & $V_1$ & 0 & $V_2$ & $e_9$ & $e_9$\\
        \hline
        2 & $Q$ & $V_2$ & 0 & $V_2$ & 0 & $2\sin(\tfrac{\pi}{9})$ \\
        \hline
        3 & $S_1$ & $V_9$ & $e_9$ & $V_3$ & $\sqrt{3}$ & $e_9+\sqrt{3}$ \\
        \hline
        4 & $Q$ & $V_3$ & $e_9$ & $V_3$ & 0 & $e_9$ \\
        \hline
        5 & $S_1$ & $V_8$ & $2e_9$ & $(-0.83682,0.28263)$ & 1.85083 & ${3.21891}^*$ \\
        \hline
        6 & $Q$ & $V_4$ & 1.48686 & $V_4$ & 0 & 1.48686 \\
        \hline
        7 & $S_1$ & $V_7$ & $3e_9$ & $(-0.57635,-0.09099)$ & 1.16679 & ${3.21891}^*$ \\
        \hline
        8 & $Q$ & $V_5$ & 2.49374 & $V_5$ & 0 & 2.49374 \\
        \hline
        9 & $Q$ & $V_6$ & 3.17778 & $V_6$ & 0 & 3.17778 \\
        \hline
    \end{tabular}
\end{small}
\caption{Trajectories' details for the upper bound of \pe{9}{1}; 
see Figure~\ref{9gonk1fig}.}
\label{9gonk1tab}
\end{table}

%% file: UpperBoundsFigs/k1servants/ALLgonk1fig.tex

\begin{figure}[H]
    \centering
    \begin{minipage}[t]{0.23\textwidth}
        \centering
        \begin{tikzpicture}[scale=0.15]
        \tikzset{
        mid arrow/.style={
            postaction={decorate, decoration={markings, mark=at position 0.6 with {\arrow{latex}}}}
        }
    }
\node (a) at (-5,8.7) {};
\fill [fill=black] (a) circle (0.105) node [above left] {$V_1$};
\node (b) at (-5,-8.7) {};
\fill [fill=black] (b) circle (0.105) node [below left] {$V_2$};
\node (c) at (10,0) {};
\fill [fill=black] (c) circle (0.105) node [above] {$V_3$};

\draw [line width=0.25mm, black] (10,0)--(-5,8.7)--(-5,-8.7)--(10,0);


\node[diamond, draw=red, fill=red, minimum size=0.15cm, inner sep=0] (a) at (-5,8.7) {};

\node (b) at (-5,-8.7) {};
\fill [fill=blue] (b) circle (0.45);
\node[diamond, draw=blue, fill=blue, minimum size=0.15cm, inner sep=0] (c) at (10,0) {};
\draw [line width=0.5mm, blue, mid arrow] (-5,-8.7) -- (10,0);
\draw[dashed] [line width=0.5mm, blue] [-latex] (-5.1,-8.7) -- (-5.1,8.7);
            
        \end{tikzpicture}
\caption{A search trajectory for \pe{3}{1}. 
Trajectory details can be found in Table~\ref{3gonk1tab}.}
\label{3gonk1fig}
    \end{minipage}
    \hfill
    \begin{minipage}[t]{0.23\textwidth}
        \centering
        \begin{tikzpicture}[scale=0.15]
        \tikzset{
        mid arrow/.style={
            postaction={decorate, decoration={markings, mark=at position 0.6 with {\arrow{latex}}}}
        }
    }
    \node (a) at (0,10) {};
    \fill [fill=black] (a) circle (0.105) node [above right] {$V_1$};
    \node (b) at (-10,0) {};
    \fill [fill=black] (b) circle (0.105) node [above left] {$V_2$};
    \node (c) at (0,-10) {};
    \fill [fill=black] (c) circle (0.105) node [below left] {$V_3$};
    \node (d) at (10,0) {};
    \fill [fill=black] (d) circle (0.105) node [above right] {$V_4$};
    \draw [line width=0.25mm, black] (10,0)--(0,10)--(-10,0)--(0,-10)--(10,0);

    \node at (a) (0,10) {};
    \fill [fill=red] (a) circle (0.45);
    \node[diamond, draw=red, fill=red, minimum size=0.15cm, inner sep=0] (d) at (10,0) {};
    \draw [line width=0.5mm, red, mid arrow] (0,10) -- (10,0);

    \node at (b) (-10,0) {};
    \fill [fill=blue] (b) circle (0.45);
    \node (e) at (3.7,-3.7) {};
    \fill [fill=blue] (e) circle (0.45);
    \node[diamond, draw=blue, fill=blue, minimum size=0.15cm, inner sep=0] (c) at (0,-10) {};
    \draw [line width=0.5mm, blue, mid arrow] (-10,0) -- (3.7,-3.7);
    \draw [line width=0.5mm, blue, mid arrow] (3.7,-3.7) -- (0,-10);
    \draw[dashed] [line width=0.5mm, blue] [-latex] (3.7,-3.7) -- (10,0);
        \end{tikzpicture}
\caption{A search trajectory for \pe{4}{1}. 
Trajectory details can be found in Table~\ref{4gonk1tab}.}
\label{4gonk1fig}
    \end{minipage}
    \hfill
    \begin{minipage}[t]{0.23\textwidth}
        \centering
        \begin{tikzpicture}[scale=0.15]
        \tikzset{
        mid arrow/.style={
            postaction={decorate, decoration={markings, mark=at position 0.6 with {\arrow{latex}}}}
        }
    }
        \node (a) at (3.1,9.5) {};
        \fill [fill=black] (a) circle (0.105) node [above right] {$V_1$};
        \node (b) at (-8.1,5.9) {};
        \fill [fill=black] (b) circle (0.105) node [left] {$V_2$};
        \node (c) at (-8.1,-5.9) {};
        \fill [fill=black] (c) circle (0.105) node [left] {$V_3$};
        \node (d) at (3.1,-9.5) {};
        \fill [fill=black] (d) circle (0.105) node [below left] {$V_4$};
        \node (e) at (10,0) {};
        \fill [fill=black] (e) circle (0.105) node [above right] {$V_5$};

        \draw [line width=0.25mm, black] (10,0)--(3.1,9.5)--(-8.1,5.9)--(-8.1,-5.9)--(3.1,-9.5)--(10,0);

        \node at (a) (3.1,9.5) {};
        \fill [fill=red] (a) circle (0.45);
        \node[diamond, draw=red, fill=red, minimum size=0.15cm, inner sep=0] (e) at (10,0) {};
        \draw [line width=0.5mm, red, mid arrow] (3.1,9.5) -- (10,0);

        \node at (d) (3.1,-9.5) {};
        \fill [fill=blue] (d) circle (0.45);
        \node (f) at (-4.6,-4.8) {};
        \fill [fill=blue] (f) circle (0.45);
        \node at (c) (-8.1,-5.9) {};
        \fill [fill=blue] (c) circle (0.45);
        \node[diamond, draw=blue, fill=blue, minimum size=0.15cm, inner sep=0] (b) at (-8.1,5.9) {};
        \draw [line width=0.5mm, blue, mid arrow] (3.1,-9.5) -- (-4.6,-4.8);
        \draw [line width=0.5mm, blue, mid arrow] (-4.6,-4.8) -- (-8.1,-5.9);
        \draw [line width=0.5mm, blue, mid arrow] (-8.1,-5.9) -- (-8.1,5.9);
        \draw[dashed] [line width=0.5mm, blue] [-latex] (-4.6,-4.8) -- (10,0);
        \end{tikzpicture}
\caption{A search trajectory for \pe{5}{1}. 
Trajectory details can be found in Table~\ref{5gonk1tab}.}
\label{5gonk1fig}
    \end{minipage}
    \hfill
    \begin{minipage}[t]{0.23\textwidth}
        \centering
        \begin{tikzpicture}[scale=0.15]
        \tikzset{
        mid arrow/.style={
            postaction={decorate, decoration={markings, mark=at position 0.6 with {\arrow{latex}}}}
        }
    }
\node (a) at (5,8.7) {};
\fill [fill=black] (a) circle (0.105) node [above] {$V_1$};
\node (b) at (-5,8.7) {};
\fill [fill=black] (b) circle (0.105) node [above] {$V_2$};
\node (c) at (-10,0) {};
\fill [fill=black] (c) circle (0.105) node [above left] {$V_3$};
\node (d) at (-5,-8.7) {};
\fill [fill=black] (d) circle (0.105) node [below left] {$V_4$};
\node (e) at (5,-8.7) {};
\fill [fill=black] (e) circle (0.105) node [below right] {$V_5$};
\node (f) at (10,0) {};
\fill [fill=black] (f) circle (0.105) node [above right] {$V_6$};

\draw [line width=0.25mm, black] (10,0)--(5,8.7)--(-5,8.7)--(-10,0)--(-5,-8.7)--(5,-8.7)--(10,0);

\node at (a) (5,8.7) {};
\fill [fill=red] (a) circle (0.45);
\node[diamond, draw=red, fill=red, minimum size=0.15cm, inner sep=0] (f) at (10,0) {};
\node[diamond, draw=red, fill=red, minimum size=0.15cm, inner sep=0] (e) at (5,-8.7) {};
\draw [line width=0.5mm, red, mid arrow] (5,8.7) -- (10,0);
\draw [line width=0.5mm, red, mid arrow] (10,0) -- (5,-8.7);

\node at (b) (-5,8.7) {};
\fill [fill=blue] (b) circle (0.45);
\node (g) at (-7,0) {};
\fill [fill=blue] (g) circle (0.45);
\node at (c) (-10,0) {};
\fill [fill=blue] (c) circle (0.45);
\node (h) at (-2.5,-4.3) {};
\fill [fill=blue] (h) circle (0.45);
\node at (d) (-5,-8.7) {};
\fill [fill=blue] (d) circle (0.45);
\draw [line width=0.5mm, blue, mid arrow] (-5,8.7) -- (-7,0);
\draw [line width=0.5mm, blue] (-7,0)  -- (-10,0) node[currarrow,blue,pos=0.5,xscale=-1,sloped,scale=0.75] {};
\draw [line width=0.5mm, blue, mid arrow] (-10,0) -- (-2.5,-4.3);
\draw [line width=0.5mm, blue, mid arrow] (-2.5,-4.3) -- (-5,-8.7);
\draw[dashed] [line width=0.5mm, blue] [-latex] (-7,0) -- (10,0);
\draw[dashed] [line width=0.5mm, blue] [-latex] (-2.5,-4.3) -- (5,-8.7);
        \end{tikzpicture}
\caption{A search trajectory for \pe{6}{1}. 
Trajectory details can be found in Table~\ref{6gonk1tab}.}
\label{6gonk1fig}
    \end{minipage}
\end{figure}


\begin{figure}[H]
    \centering
    \begin{minipage}[t]{0.23\textwidth}
        \centering
        \begin{tikzpicture}[scale=0.15]
        \tikzset{
        mid arrow/.style={
            postaction={decorate, decoration={markings, mark=at position 0.6 with {\arrow{latex}}}}
        }
    }
\node (a) at (6.2,7.8) {};
        \fill [fill=black] (a) circle (0.105) node [above] {$V_1$};
        \node (b) at (-2.2,9.7) {};
        \fill [fill=black] (b) circle (0.105) node [above] {$V_2$};
        \node (c) at (-9,4.3) {};
        \fill [fill=black] (c) circle (0.105) node [above left] {$V_3$};
        \node (d) at (-9,-4.3) {};
        \fill [fill=black] (d) circle (0.105) node [above left] {$V_4$};
        \node (e) at (-2.2,-9.7) {};
        \fill [fill=black] (e) circle (0.105) node [below left] {$V_5$};
        \node (f) at (6.2,-7.8) {};
        \fill [fill=black] (f) circle (0.105) node [below right] {$V_6$};
        \node (g) at (10,0) {};
        \fill [fill=black] (g) circle (0.105) node [above right] {$V_7$};

        \draw [line width=0.25mm, black] (10,0)--(6.2,7.8)--(-2.2,9.7)--(-9,4.3)--(-9,-4.3)--(-2.2,-9.7)--(6.2,-7.8)--(10,0);

        \node (b) at (-2.2,9.7) {};
        \fill [fill=red] (b) circle (0.45);
        \node (c) at (-9,4.3) {};
        \fill [fill=red] (c) circle (0.45);
        \node[diamond, draw=red, fill=red, minimum size=0.15cm, inner sep=0] (d) at (-9,-4.3) {};
        \draw [line width=0.5mm, red, mid arrow] (-2.2,9.7) -- (-9,4.3);
        \draw [line width=0.5mm, red, mid arrow] (-9,4.3) -- (-9,-4.3);

        \node (a) at (6.2,7.8) {};
        \fill [fill=blue] (a) circle (0.45);
        \node (g) at (10,0) {};
        \fill [fill=blue] (g) circle (0.45);
        \node (h) at (3.3,-5.5) {};
        \fill [fill=blue] (h) circle (0.45);
        \node (f) at (6.2,-7.8) {};
        \fill [fill=blue] (f) circle (0.45);
        \node[diamond, draw=blue, fill=blue, minimum size=0.15cm, inner sep=0] (e) at (-2.2,-9.7) {};
        \draw [line width=0.5mm, blue, mid arrow] (6.2,7.8) -- (10,0);
        \draw [line width=0.5mm, blue, mid arrow] (10,0) -- (3.3,-5.5);
        \draw [line width=0.5mm, blue, mid arrow] (3.3,-5.5) -- (6.2,-7.8);
        \draw [line width=0.5mm, blue, mid arrow] (6.2,-7.8) -- (-2.2,-9.7);
        \draw[dashed] [line width=0.5mm, blue] [-latex] (3.3,-5.5) -- (-9,-4.3);
            
        \end{tikzpicture}
\caption{A search trajectory for \pe{7}{1}. 
Trajectory details can be found in Table~\ref{7gonk1tab}.}
\label{7gonk1fig}
    \end{minipage}
    \hfill
    \begin{minipage}[t]{0.23\textwidth}
        \centering
        \begin{tikzpicture}[scale=0.15]
        \tikzset{
        mid arrow/.style={
            postaction={decorate, decoration={markings, mark=at position 0.6 with {\arrow{latex}}}}
        }
    }
    \node (a) at (7.1,7.1) {};
        \fill [fill=black] (a) circle (0.105) node [above] {$V_1$};
        \node (b) at (0,10) {};
        \fill [fill=black] (b) circle (0.105) node [above right] {$V_2$};
        \node (c) at (-7.1,7.1) {};
        \fill [fill=black] (c) circle (0.105) node [above left] {$V_3$};
        \node (d) at (-10,0) {};
        \fill [fill=black] (d) circle (0.105) node [above left] {$V_4$};
        \node (e) at (-7.1,-7.1) {};
        \fill [fill=black] (e) circle (0.105) node [below left] {$V_5$};
        \node (f) at (0, -10) {};
        \fill [fill=black] (f) circle (0.105) node [below right] {$V_6$};
        \node (g) at (7.1,-7.1) {};
        \fill [fill=black] (g) circle (0.105) node [below right] {$V_7$};
        \node (h) at (10,0) {};
        \fill [fill=black] (h) circle (0.105) node [above right] {$V_8$};

        \draw [line width=0.25mm, black] (10,0)--(7.1,7.1)--(0,10)--(-7.1,7.1)--(-10,0)--(-7.1,-7.1)--(0, -10)--(7.1,-7.1)--(10,0);

        \node at (b) (0,10) {};
        \fill [fill=red] (b) circle (0.45);
        \node at (c) (-7.1,7.1) {};
        \fill [fill=red] (c) circle (0.45);
        \node[diamond, draw=red, fill=red, minimum size=0.15cm, inner sep=0] (d) at (-10,0) {};
        \node[diamond, draw=red, fill=red, minimum size=0.15cm, inner sep=0] (e) at (-7.1,-7.1) {};
        \draw [line width=0.5mm, red, mid arrow] (0,10) -- (-7.1,7.1);
        \draw [line width=0.5mm, red, mid arrow] (-7.1,7.1) -- (-10,0);
        \draw [line width=0.5mm, red, mid arrow] (-10,0) -- (-7.1,-7.1);

        \node at (a) (7.1,7.1) {};
        \fill [fill=blue] (a) circle (0.45);
        \node at (h) (10,0) {};
        \fill [fill=blue] (h) circle (0.45);
        \node (i) at (4.2,-4.9) {};
        \fill [fill=blue] (i) circle (0.45);
        \node at (f) (0, -10) {};
        \fill [fill=blue] (f) circle (0.45);
        \node (j) at (0,-7.5) {};
        \fill [fill=blue] (j) circle (0.45);
        \node[diamond, draw=blue, fill=blue, minimum size=0.15cm, inner sep=0] (g) at (7.1,-7.1) {};
        \draw [line width=0.5mm, blue, mid arrow] (7.1,7.1) -- (10,0);
        \draw [line width=0.5mm, blue, mid arrow] (10,0) -- (4.2,-4.9);
        \draw [line width=0.5mm, blue, mid arrow] (4.2,-4.9) -- (0, -10);
        \draw [line width=0.5mm, blue] (0, -10)  -- (0,-7.5) node[currarrow,blue,pos=0.5,xscale=1,sloped,scale=0.75] {};
        \draw [line width=0.5mm, blue, mid arrow] (0,-7.5) -- (7.1,-7.1);
        \draw[dashed] [line width=0.5mm, blue] [-latex] (4.2,-4.9) -- (-10,0);
        \draw[dashed] [line width=0.5mm, blue] [-latex] (0,-7.5) -- (-7.1,-7.1);
        \end{tikzpicture}
\caption{A search trajectory for \pe{8}{1}. 
Trajectory details can be found in Table~\ref{8gonk1tab}.}
\label{8gonk1fig}
    \end{minipage}
    \hfill
    \begin{minipage}[t]{0.23\textwidth}
        \centering
        \begin{tikzpicture}[scale=0.15]
        \tikzset{
        mid arrow/.style={
            postaction={decorate, decoration={markings, mark=at position 0.6 with {\arrow{latex}}}}
        }
    }
        \node (a) at (7.7,6.4) {};
        \fill [fill=black] (a) circle (0.105) node [above] {$V_1$};
        \node (b) at (1.7,9.8) {};
        \fill [fill=black] (b) circle (0.105) node [above] {$V_2$};
        \node (c) at (-5,8.7) {};
        \fill [fill=black] (c) circle (0.105) node [above] {$V_3$};
        \node (d) at (-9.4,3.4) {};
        \fill [fill=black] (d) circle (0.105) node [above left] {$V_4$};
        \node (e) at (-9.4,-3.4) {};
        \fill [fill=black] (e) circle (0.105) node [below left] {$V_5$};
        \node (f) at (-5,-8.7) {};
        \fill [fill=black] (f) circle (0.105) node [below left] {$V_6$};
        \node (g) at (1.7,-9.8) {};
        \fill [fill=black] (g) circle (0.105) node [below right] {$V_7$};
        \node (h) at (7.7,-6.4) {};
        \fill [fill=black] (h) circle (0.105) node [below right] {$V_8$};
        \node (i) at (10,0) {};
        \fill [fill=black] (i) circle (0.105) node [above right] {$V_9$};

        \draw [line width=0.25mm, black] (10,0)--(7.7,6.4)--(1.7,9.8)--(-5,8.7)--(-9.4,3.4)--(-9.4,-3.4)--(-5,-8.7)--(1.7,-9.8)--(7.7,-6.4)--(10,0);

        \node at (a) (7.7,6.4) {};
        \fill [fill=red] (a) circle (0.45);
        \node at (i) (10,0) {};
        \fill [fill=red] (i) circle (0.45);
        \node[diamond, draw=red, fill=red, minimum size=0.15cm, inner sep=0] (h) at (7.7,-6.4) {};
        \node[diamond, draw=red, fill=red, minimum size=0.15cm, inner sep=0] (g) at (1.7,-9.8) {};
        \draw [line width=0.5mm, red, mid arrow] (7.7,6.4) -- (10,0);
        \draw [line width=0.5mm, red, mid arrow] (10,0) -- (7.7,-6.4);
        \draw [line width=0.5mm, red, mid arrow] (7.7,-6.4) -- (1.7,-9.8);

        \node at (b) (1.7,9.8) {};
        \fill [fill=blue] (b) circle (0.45);
        \node at (c) (-5,8.7) {};
        \fill [fill=blue] (c) circle (0.45);
        \node (j) at (-7,3.5) {};
        \fill [fill=blue] (j) circle (0.45);
        \node at (d) (-9.4,3.4) {};
        \fill [fill=blue] (d) circle (0.45);
        \node (k) at (-5.8,-0.9) {};
        \fill [fill=blue] (k) circle (0.45);
        \node at (e) (-9.4,-3.4) {};
        \fill [fill=blue] (e) circle (0.45);
        \node at (f) (-5,-8.7) {};
        \fill [fill=blue] (f) circle (0.45);
        \draw [line width=0.5mm, blue, mid arrow] (1.7,9.8) -- (-5,8.7);
        \draw [line width=0.5mm, blue, mid arrow] (-5,8.7) -- (-7,3.5);
        \draw [line width=0.5mm, blue] (-7,3.5)  -- (-9.4,3.4) node[currarrow,blue,pos=0.55,xscale=-1,sloped,scale=0.75] {};
        \draw [line width=0.5mm, blue, mid arrow] (-9.4,3.4) -- (-5.8,-0.9);
        \draw [line width=0.5mm, blue, mid arrow] (-5.8,-0.9) -- (-9.4,-3.4);
        \draw [line width=0.5mm, blue, mid arrow] (-9.4,-3.4) -- (-5,-8.7);
        \draw[dashed] [line width=0.5mm, blue] [-latex] (-7,3.5) -- (7.7,-6.4);
        \draw[dashed] [line width=0.5mm, blue] [-latex] (-5.8,-0.9) -- (1.7,-9.8);
        \end{tikzpicture}
    \caption{A search trajectory for \pe{9}{1}. 
Trajectory details can be found in Table~\ref{9gonk1tab}.}
\label{9gonk1fig}
    \end{minipage}
    \hfill
    \begin{minipage}[t]{0.23\textwidth}
        \centering
        \begin{tikzpicture}[scale=0.15]
        \tikzset{
        mid arrow/.style={
            postaction={decorate, decoration={markings, mark=at position 0.6 with {\arrow{latex}}}}
        }
    }
\node (a) at (8.1,5.9) {};
        \fill [fill=black] (a) circle (0.105) node [above] {$V_1$};
        \node (b) at (3.1,9.5) {};
        \fill [fill=black] (b) circle (0.105) node [above] {$V_2$};
        \node (c) at (-3.1,9.5) {};
        \fill [fill=black] (c) circle (0.105) node [above] {$V_3$};
        \node (d) at (-8.1,5.9) {};
        \fill [fill=black] (d) circle (0.105) node [above left] {$V_4$};
        \node (e) at (-10,0) {};
        \fill [fill=black] (e) circle (0.105) node [above left] {$V_5$};
        \node (f) at (-8.1,-5.9) {};
        \fill [fill=black] (f) circle (0.105) node [below left] {$V_6$};
        \node (g) at (-3.1,-9.5) {};
        \fill [fill=black] (g) circle (0.105) node [below] {$V_7$};
        \node (h) at (3.1,-9.5) {};
        \fill [fill=black] (h) circle (0.105) node [below right] {$V_8$};
        \node (i) at (8.1,-5.9) {};
        \fill [fill=black] (i) circle (0.105) node [below right] {$V_9$};
        \node (j) at (10,0) {};
        \fill [fill=black] (j) circle (0.105) node [above right] {$V_{10}$};

        \draw [line width=0.25mm, black] (10,0)--(8.1,5.9)--(3.1,9.5)--(-3.1,9.5)--(-8.1,5.9)--(-10,0)--(-8.1,-5.9)--(-3.1,-9.5)--(3.1,-9.5)--(8.1,-5.9)--(10,0);

        \node at (a) (8.1,5.9) {};
        \fill [fill=red] (a) circle (0.45);
        \node at (j) (10,0) {};
        \fill [fill=red] (j) circle (0.45);
        \node[diamond, draw=red, fill=red, minimum size=0.15cm, inner sep=0] (i) at (8.1,-5.9) {};
        \node[diamond, draw=red, fill=red, minimum size=0.15cm, inner sep=0] (h) at (3.1,-9.5) {};
        \node at (g) (-3.1,-9.5) {};
        \fill [fill=red] (g) circle (0.45);
        \draw [line width=0.5mm, red, mid arrow] (8.1,5.9) -- (10,0);
        \draw [line width=0.5mm, red, mid arrow] (10,0) -- (8.1,-5.9);
        \draw [line width=0.5mm, red, mid arrow] (8.1,-5.9) -- (3.1,-9.5);
        \draw [line width=0.5mm, red, mid arrow] (3.1,-9.5) -- (-3.1,-9.5);

        \node at (b) (3.1,9.5) {};
        \fill [fill=blue] (b) circle (0.45);
        \node at (c) (-3.1,9.5) {};
        \fill [fill=blue] (c) circle (0.45);
        \node (k) at (-6,4.5) {};
        \fill [fill=blue] (k) circle (0.45);
        \node at (d) (-8.1,5.9) {};
        \fill [fill=blue] (d) circle (0.45);
        \node (l) at (-6.5,0.1) {};
        \fill [fill=blue] (l) circle (0.45);
        \node at (f) (-8.1,-5.9) {};
        \fill [fill=blue] (f) circle (0.45);
        \node at (e) (-10,0) {};
        \fill [fill=blue] (e) circle (0.45);
        \draw [line width=0.5mm, blue, mid arrow] (3.1,9.5) -- (-3.1,9.5);
        \draw [line width=0.5mm, blue, mid arrow] (-3.1,9.5) -- (-6,4.5);
        \draw [line width=0.5mm, blue] (-6,4.5)  -- (-8.1,5.9) node[currarrow,blue,pos=0.55,xscale=-1,sloped,scale=0.75] {};
        \draw [line width=0.5mm, blue, mid arrow] (-8.1,5.9) -- (-6.5,0.1);
        \draw [line width=0.5mm, blue, mid arrow] (-6.5,0.1) -- (-8.1,-5.9);
        \draw [line width=0.5mm, blue, mid arrow] (-8.1,-5.9) -- (-10,0);
        \draw[dashed] [line width=0.5mm, blue] [-latex] (-6,4.5) -- (8.1,-5.9);
        \draw[dashed] [line width=0.5mm, blue] [-latex] (-6.5,0.1) -- (3.1,-9.5);
        \end{tikzpicture}
\caption{A search trajectory for \pe{10}{1}. 
Trajectory details can be found in Table~\ref{10gonk1tab}.}
\label{10gonk1fig}
    \end{minipage}
\end{figure}


\begin{figure}[H]
    \centering
    \begin{minipage}[t]{0.23\textwidth}
        \centering
        \begin{tikzpicture}[scale=0.15]
        \tikzset{
        mid arrow/.style={
            postaction={decorate, decoration={markings, mark=at position 0.6 with {\arrow{latex}}}}
        }
    }
            \node (a) at (8.4,5.4) {};
            \fill [fill=black] (a) circle (0.105) node [above] {$V_1$};
            \node (b) at (4.2,9.1) {};
            \fill [fill=black] (b) circle (0.105) node [above] {$V_2$};
            \node (c) at (-1.4,9.9) {};
            \fill [fill=black] (c) circle (0.105) node [above] {$V_3$};
            \node (d) at (-6.5,7.6) {};
            \fill [fill=black] (d) circle (0.105) node [above left] {$V_4$};
            \node (e) at (-9.6,2.8) {};
            \fill [fill=black] (e) circle (0.105) node [above left] {$V_5$};
            \node (f) at (-9.6,-2.8) {};
            \fill [fill=black] (f) circle (0.105) node [below left] {$V_6$};
            \node (g) at (-6.5,-7.6) {};
            \fill [fill=black] (g) circle (0.105) node [below left] {$V_7$};
            \node (h) at (-1.4,-9.9) {};
            \fill [fill=black] (h) circle (0.105) node [below right] {$V_8$};
            \node (i) at (4.2,-9.1) {};
            \fill [fill=black] (i) circle (0.105) node [below right] {$V_9$};
            \node (j) at (8.4,-5.4) {};
            \fill [fill=black] (j) circle (0.105) node [below right] {$V_{10}$};
            \node (k) at (10,0) {};
            \fill [fill=black] (k) circle (0.105) node [above right] {$V_{11}$};

            \draw [line width=0.25mm, black] (10,0)--(8.4,5.4)--(4.2,9.1)--(-1.4,9.9)--(-6.5,7.6)--(-9.6,2.8)--(-9.6,-2.8)--(-6.5,-7.6)--(-1.4,-9.9)--(4.2,-9.1)--(8.4,-5.4)--(10,0);

            \node at (a) (8.4,5.4) {};
            \fill [fill=red] (a) circle (0.45);
            \node at (k) (10,0) {};
            \fill [fill=red] (k) circle (0.45);
            \node at (j) (8.4,-5.4) {};
            \fill [fill=red] (j) circle (0.45);
            \node[diamond, draw=red, fill=red, minimum size=0.15cm, inner sep=0] (i) at (4.2,-9.1) {};
            \node[diamond, draw=red, fill=red, minimum size=0.15cm, inner sep=0] (h) at (-1.4,-9.9) {};
            \draw [line width=0.5mm, red, mid arrow] (8.4,5.4) -- (10,0);
            \draw [line width=0.5mm, red, mid arrow] (10,0) -- (8.4,-5.4);
            \draw [line width=0.5mm, red, mid arrow] (8.4,-5.4) -- (4.2,-9.1);
            \draw [line width=0.5mm, red, mid arrow] (4.2,-9.1) -- (-1.4,-9.9);

            \node at (b) (4.2,9.1) {};
            \fill [fill=blue] (b) circle (0.45);
            \node at (c) (-1.4,9.9) {};
            \fill [fill=blue] (c) circle (0.45);
            \node at (d) (-6.5,7.6) {};
            \fill [fill=blue] (d) circle (0.45);
            \node (l) at (-7,3.5) {};
            \fill [fill=blue] (l) circle (0.45);
            \node at (e) (-9.6,2.8) {};
            \fill [fill=blue] (e) circle (0.45);
            \node (m) at (-7.4,-0.6) {};
            \fill [fill=blue] (m) circle (0.45);
            \node at (f) (-9.6,-2.8) {};
            \fill [fill=blue] (f) circle (0.45);
            \node at (g) (-6.5,-7.6) {};
            \fill [fill=blue] (g) circle (0.45);
            \draw [line width=0.5mm, blue, mid arrow] (4.2,9.1) -- (-1.4,9.9);
            \draw [line width=0.5mm, blue, mid arrow] (-1.4,9.9) -- (-6.5,7.6);
            \draw [line width=0.5mm, blue, mid arrow] (-6.5,7.6) -- (-7,3.5);
            \draw [line width=0.5mm, blue] (-7,3.5)  -- (-9.6,2.8) node[currarrow,blue,pos=0.55,xscale=-1,sloped,scale=0.75] {};
            \draw [line width=0.5mm, blue] (-9.6,2.8)-- (-7.4,-0.6) node[currarrow,blue,pos=0.55,xscale=1,sloped,scale=0.75] {};
            \draw [line width=0.5mm, blue] (-7.4,-0.6)-- (-9.6,-2.8) node[currarrow,blue,pos=0.55,xscale=-1,sloped,scale=0.75] {};
            \draw [line width=0.5mm, blue, mid arrow] (-9.6,-2.8)-- (-6.5,-7.6);
            \draw[dashed] [line width=0.5mm, blue] [-latex] (-7,3.5) -- (4.2,-9.1);
            \draw[dashed] [line width=0.5mm, blue] [-latex] (-7.4,-0.6) -- (-1.4,-9.9);
            
        \end{tikzpicture}
\caption{A search trajectory for \pe{11}{1}. 
Trajectory details can be found in Table~\ref{11gonk1tab}.}
\label{11gonk1fig}
    \end{minipage}
    \hfill
    \begin{minipage}[t]{0.23\textwidth}
        \centering
        \begin{tikzpicture}[scale=0.15]
        \tikzset{
        mid arrow/.style={
            postaction={decorate, decoration={markings, mark=at position 0.6 with {\arrow{latex}}}}
        }
    }
        \node (a) at (8.7,5) {};
        \fill [fill=black] (a) circle (0.105) node [above] {$V_1$};
        \node (b) at (5,8.7) {};
        \fill [fill=black] (b) circle (0.105) node [above] {$V_2$};
        \node (c) at (0,10) {};
        \fill [fill=black] (c) circle (0.105) node [above left] {$V_3$};
        \node (d) at (-5,8.7) {};
        \fill [fill=black] (d) circle (0.105) node [above] {$V_4$};
        \node (e) at (-8.7,5) {};
        \fill [fill=black] (e) circle (0.105) node [above left] {$V_5$};
        \node (f) at (-10,0) {};
        \fill [fill=black] (f) circle (0.105) node [above left] {$V_6$};
        \node (g) at (-8.7,-5) {};
        \fill [fill=black] (g) circle (0.105) node [left] {$V_7$};
        \node (h) at (-5,-8.7) {};
        \fill [fill=black] (h) circle (0.105) node [below left] {$V_8$};
        \node (i) at (0,-10) {};
        \fill [fill=black] (i) circle (0.105) node [below right] {$V_9$};
        \node (j) at (5,-8.7) {};
        \fill [fill=black] (j) circle (0.105) node [below right] {$V_{10}$};
        \node (k) at (8.7,-5) {};
        \fill [fill=black] (k) circle (0.105) node [below right] {$V_{11}$};
        \node (l) at (10,0) {};
        \fill [fill=black] (l) circle (0.105) node [above right] {$V_{12}$};

        \draw [line width=0.25mm, black] (10,0)--(8.7,5)--(5,8.7)--(0,10)--(-5,8.7)--(-8.7,5)--(-10,0)--(-8.7,-5)--(-5,-8.7)--(0,-10)--(5,-8.7)--(8.7,-5)--(10,0);

        \node at (a) (8.7,5) {};
        \fill [fill=red] (a) circle (0.45);
        \node at (l) (10,0) {};
        \fill [fill=red] (l) circle (0.45);
        \node at (k) (8.7,-5) {};
        \fill [fill=red] (k) circle (0.45);
        \node[diamond, draw=red, fill=red, minimum size=0.15cm, inner sep=0] (j) at (5,-8.7) {};
        \node[diamond, draw=red, fill=red, minimum size=0.15cm, inner sep=0] (j) at (0,-10) {};
        \node at (h) (-5,-8.7) {};
        \fill [fill=red] (h) circle (0.45);
        \draw [line width=0.5mm, red, mid arrow]  (8.7,5) -- (10,0);
        \draw [line width=0.5mm, red, mid arrow] (10,0) -- (8.7,-5);
        \draw [line width=0.5mm, red, mid arrow] (8.7,-5) -- (5,-8.7);
        \draw [line width=0.5mm, red, mid arrow] (5,-8.7) -- (0,-10);
        \draw [line width=0.5mm, red, mid arrow] (0,-10) -- (-5,-8.7);

        \node at (b) (5,8.7) {};
        \fill [fill=blue] (b) circle (0.45);
        \node at (c) (0,10) {};
        \fill [fill=blue] (c) circle (0.45);
        \node at (d) (-5,8.7) {};
        \fill [fill=blue] (d) circle (0.45);
        \node (m) at (-6.5,4.4) {};
        \fill [fill=blue] (m) circle (0.45);
        \node at (e) (-8.7,5) {};
        \fill [fill=blue] (e) circle (0.45);
        \node (n) at (-6.6,1.4) {};
        \fill [fill=blue] (n) circle (0.45);
        \node (o) at (-7.6,-1.7) {};
        \fill [fill=blue] (o) circle (0.45);
        \node at (f) (-10,0) {};
        \fill [fill=blue] (f) circle (0.45);
        \node at (g) (-8.7,-5) {};
        \fill [fill=blue] (g) circle (0.45);
        \draw [line width=0.5mm, blue, mid arrow] (5,8.7) -- (0,10);
        \draw [line width=0.5mm, blue, mid arrow] (0,10) -- (-5,8.7);
        \draw [line width=0.5mm, blue, mid arrow] (-5,8.7) -- (-6.5,4.4);
        \draw [line width=0.5mm, blue] (-6.5,4.4)  -- (-8.7,5) node[currarrow,blue,pos=0.55,xscale=-1,sloped,scale=0.75] {};
        \draw [line width=0.5mm, blue] (-8.7,5) -- (-6.6,1.4) node[currarrow,blue,pos=0.55,xscale=1,sloped,scale=0.75] {};
        \draw [line width=0.5mm, blue] (-6.6,1.4) -- (-7.6,-1.7) node[currarrow,blue,pos=0.55,xscale=-1,sloped,scale=0.75] {};
        \draw [line width=0.5mm, blue] (-7.6,-1.7) -- (-10,0) node[currarrow,blue,pos=0.55,xscale=-1,sloped,scale=0.75] {};
        \draw [line width=0.5mm, blue, mid arrow] (-10,0) -- (-8.7,-5);
        \draw[dashed] [line width=0.5mm, blue] [-latex] (-6.5,4.4) -- (5,-8.7);
        \draw[dashed] [line width=0.5mm, blue] [-latex] (-6.6,1.4) -- (0,-10);
        \end{tikzpicture}
\caption{A search trajectory for \pe{12}{1}. 
Trajectory details can be found in Table~\ref{12gonk1tab}.}
\label{12gonk1fig}
    \end{minipage}
    \hfill
    \begin{minipage}[t]{0.23\textwidth}
        \centering
        \begin{tikzpicture}[scale=0.15]
        \tikzset{
        mid arrow/.style={
            postaction={decorate, decoration={markings, mark=at position 0.6 with {\arrow{latex}}}}
        }
    }
        \node (a) at (8.9,4.6) {};
        \fill [fill=black] (a) circle (0.105) node [above] {$V_1$};
        \node (b) at (5.7,8.2) {};
        \fill [fill=black] (b) circle (0.105) node [above] {$V_2$};
        \node (c) at (1.2,9.9) {};
        \fill [fill=black] (c) circle (0.105) node [above] {$V_3$};
        \node (d) at (-3.5,9.4) {};
        \fill [fill=black] (d) circle (0.105) node [above] {$V_4$};
        \node (e) at (-7.5,6.6) {};
        \fill [fill=black] (e) circle (0.105) node [above] {$V_5$};
        \node (f) at (-9.7,2.4) {};
        \fill [fill=black] (f) circle (0.105) node [above left] {$V_6$};
        \node (g) at (-9.7,-2.4) {};
        \fill [fill=black] (g) circle (0.105) node [below left] {$V_7$};
        \node (h) at (-7.5,-6.6) {};
        \fill [fill=black] (h) circle (0.105) node [below left] {$V_8$};
        \node (i) at (-3.5,-9.4) {};
        \fill [fill=black] (i) circle (0.105) node [below left] {$V_9$};
        \node (j) at (1.2,-9.9) {};
        \fill [fill=black] (j) circle (0.105) node [below] {$V_{10}$};
        \node (k) at (5.7,-8.2) {};
        \fill [fill=black] (k) circle (0.105) node [below right] {$V_{11}$};
        \node (l) at (8.9,-4.6) {};
        \fill [fill=black] (l) circle (0.105) node [below right] {$V_{12}$};
        \node (m) at (10,0) {};
        \fill [fill=black] (m) circle (0.105) node [above right] {$V_{13}$};

        \draw [line width=0.25mm, black] (10,0)--(8.9,4.6)--(5.7,8.2)--(1.2,9.9)--(-3.5,9.4)--(-7.5,6.6)--(-9.7,2.4)--(-9.7,-2.4)--(-7.5,-6.6)--(-3.5,-9.4)--(1.2,-9.9)--(5.7,-8.2)--(8.9,-4.6)--(10,0);

        \node at (b) (5.7,8.2) {};
        \fill [fill=red] (b) circle (0.45);
        \node at (c) (1.2,9.9) {};
        \fill [fill=red] (c) circle (0.45);
        \node at (d) (-3.5,9.4) {};
        \fill [fill=red] (d) circle (0.45);
        \node[diamond, draw=red, fill=red, minimum size=0.15cm, inner sep=0] (e) at (-7.5,6.6) {};
        \node[diamond, draw=red, fill=red, minimum size=0.15cm, inner sep=0] (f) at (-9.7,2.4) {};
        \node at (g) (-9.7,-2.4) {};
        \fill [fill=red] (g) circle (0.45);
        \node at (h) (-7.5,-6.6) {};
        \fill [fill=red] (h) circle (0.45);
        \draw [line width=0.5mm, red, mid arrow] (5.7,8.2) -- (1.2,9.9);
        \draw [line width=0.5mm, red, mid arrow] (1.2,9.9) -- (-3.5,9.4);
        \draw [line width=0.5mm, red, mid arrow] (-3.5,9.4) -- (-7.5,6.6);
        \draw [line width=0.5mm, red, mid arrow] (-7.5,6.6) -- (-9.7,2.4);
        \draw [line width=0.5mm, red, mid arrow] (-9.7,2.4) -- (-9.7,-2.4);
        \draw [line width=0.5mm, red, mid arrow] (-9.7,-2.4) -- (-7.5,-6.6);

        \node at (a) (8.9,4.6) {};
        \fill [fill=blue] (a) circle (0.45);
        \node at (m) (10,0) {};
        \fill [fill=blue] (m) circle (0.45);
        \node at (l) (8.9,-4.6) {};
        \fill [fill=blue] (l) circle (0.45);
        \node (n) at (6,-6) {};
        \fill [fill=blue] (n) circle (0.45);
        \node at (k) (5.7,-8.2) {};
        \fill [fill=blue] (k) circle (0.45);
        \node (o) at (2.2,-5.8) {};
        \fill [fill=blue] (o) circle (0.45);
        \node (p) at (-2,-7.1) {};
        \fill [fill=blue] (p) circle (0.45);
        \node at (i) (-3.5,-9.4) {};
        \fill [fill=blue] (i) circle (0.45);
        \node (q) at (-4.5,-7) {};
        \fill [fill=blue] (q) circle (0.45);
        \node at (j) (1.2,-9.9) {};
        \fill [fill=blue] (j) circle (0.45);
        \draw [line width=0.5mm, blue, mid arrow] (8.9,4.6) -- (10,0);
        \draw [line width=0.5mm, blue, mid arrow] (10,0) -- (8.9,-4.6);
        \draw [line width=0.5mm, blue] (8.9,-4.6) -- (6,-6) node[currarrow,blue,pos=0.55,xscale=-1,sloped,scale=0.75] {};
        \draw [line width=0.5mm, blue] (6,-6)  -- (5.7,-8.2) node[currarrow,blue,pos=0.55,xscale=-1,sloped,scale=0.75] {};
        \draw [line width=0.5mm, blue, mid arrow] (5.7,-8.2) -- (2.2,-5.8);
        \draw [line width=0.5mm, blue, mid arrow] (2.2,-5.8) -- (-2,-7.1);
        \draw [line width=0.5mm, blue, mid arrow] (-2,-7.1) -- (-3.5,-9.4);
        \draw [line width=0.5mm, blue] (-3.5,-9.4)  -- (-4.5,-7) node[currarrow,blue,pos=0.55,xscale=-1,sloped,scale=0.75] {};
        \draw [line width=0.5mm, blue, mid arrow] (-4.5,-7) -- (1.2,-9.9);
        \draw[dashed] [line width=0.5mm, blue] [-latex] (6,-6) -- (-7.5,6.6);
        \draw[dashed] [line width=0.5mm, blue] [-latex] (2.2,-5.8) -- (-9.7,2.4);
        \end{tikzpicture}
    \caption{A search trajectory for \pe{13}{1}. 
Trajectory details can be found in Table~\ref{13gonk1tab}.}
\label{13gonk1fig}
    \end{minipage}
\end{figure}

%% file: UpperBoundsTabs/k1servants/3gonk1tab.tex
\begin{table}[h!]
\centering
\begin{small}
    \begin{tabular}{|c|c|c|c|c|c|c|}
        \hline
        $i$ & Agent & Vertex & $t_i$ & $Q^{(i)}$ & $d(Q^{(i)} , Exit)$ & Cost \\
        \hline
         1 & $S_1$  &$V_1$ & 0 & $V_2$ & $e_3$ & ${\sqrt{3}}^*$ \\
         \hline
         2 & $Q$ & $V_2$ & 0 &$V_2$ & 0 & 0 \\
         \hline
         3 & $Q$ & $V_3$ & $e_3$ & $V_3$ & 0 & ${\sqrt{3}}^*$ \\
         \hline
    \end{tabular}
    \end{small}
\caption{Trajectories' details for the upper bound of \pe{3}{1}; see Figure~\ref{3gonk1fig}.}
\label{3gonk1tab}
\end{table}

%% file: UpperBoundsTabs/k1servants/4gonk1tab.tex
\begin{table}[h!]
\centering
    \begin{small}
    \begin{tabular}{|c|c|c|c|c|c|c|}
        \hline
        $i$ & Agent & Vertex & $t_i$ & $Q^{(i)}$ & $d(Q^{(i)} , Exit)$ & Cost \\
        \hline
        1 & $S_1$ & $V_1$ & 0 & $V_2$ & $e_4$ & $e_4$ \\
        \hline
        2 & $Q$ & $V_2$ & 0 & $V_2$ & 0 & 0 \\
        \hline
        3 & $S_1$ & $V_4$ & $e_4$ & $(\tfrac{1}{2}(-1 + \sqrt{3}),\tfrac{1}{2}(1-\sqrt{3}))$ & $-1+\sqrt{3}$ & ${-1+\sqrt{2}+\sqrt{3}}^*$ \\
        \hline
        4 & $Q$ & $V_3$ & $-1+\sqrt{2}+\sqrt{3}$ & $V_3$ & 0 & ${-1+\sqrt{2}+\sqrt{3}}^*$  \\
        \hline
        \end{tabular}
            \end{small}
\caption{Trajectories' details for the upper bound of \pe{4}{1}; see Figure~\ref{4gonk1fig}.}
\label{4gonk1tab}
\end{table}

%% file: UpperBoundsTabs/k1servants/5gonk1tab.tex
\begin{table}[h!]
\centering
    \begin{small}
        \begin{tabular}{|c|c|c|c|c|c|c|}
        \hline
        $i$ & Agent & Vertex & $t_i$ & $Q^{(i)}$ & $d(Q^{(i)} , Exit)$ & Cost \\
        \hline
        1 & $S_1$ & $V_1$ & 0 & $V_4$ & $2\sin(\tfrac{2\pi}{5})$ &$2\sin(\tfrac{2\pi}{5})$ \\
        \hline
        2 & $Q$ & $V_4$ & 0 & $V_4$ & 0 & 0 \\
        \hline
        3 & $S_1$ & $V_5$ & $e_5$ & $(-0.46353,-0.47553)$ & 1.53884 & ${2.71441}^*$ \\
        \hline
        4 & $Q$ & $V_3$ & 1.53884 & $V_3$ & 0 & 1.53884 \\
        \hline
        5 & $Q$ & $V_2$ & 2.71441 & $V_2$ & 0 & ${2.71441}^*$  \\
        \hline
        \end{tabular}
      \end{small}
\caption{Trajectories' details for the upper bound of \pe{5}{1}; see Figure~\ref{5gonk1fig}.}
\label{5gonk1tab}
\end{table}

%% file: UpperBoundsTabs/k1servants/6gonk1tab.tex
\begin{table}[h!]
\centering
    \begin{small}
    \begin{tabular}{|c|c|c|c|c|c|c|}
        \hline
        $i$ & Agent & Vertex & $t_i$ & $Q^{(i)}$ & $d(Q^{(i)} , Exit)$ & Cost \\
        \hline
        1 & $S_1$ & $V_1$ & 0 & $V_2$ & $e_6$ & $e_6$ \\
        \hline
        2 & $Q$ & $V_2$ & 0 & $V_2$ & 0 & 0 \\
        \hline
        3 & $S_1$ & $V_6$ & $e_6$ & $(-\tfrac{\sqrt{3}}{2}, 0)$ & $1+\tfrac{\sqrt{3}}{2}$ & ${2+\tfrac{\sqrt{3}}{2}}^*$\\
        \hline
        4 & $Q$ & $V_3$ & $2-\tfrac{\sqrt{3}}{2}$ & $V_3$ & 0 & $2-\tfrac{\sqrt{3}}{2}$ \\
        \hline
        5 & $S_1$ & $V_5$ & $2e_6$ & $(-\tfrac{1}{4},-\tfrac{\sqrt{3}}{4}) $ & $\tfrac{\sqrt{3}}{2}$ & ${2+\tfrac{\sqrt{3}}{2}}^*$\\
        \hline
        6 & $Q$ & $V_4$ & $\tfrac{5}{2}$ & $V_4$ & 0 & $\tfrac{5}{2}$\\
        \hline
    \end{tabular}
  \end{small}
\caption{Trajectories' details for the upper bound of \pe{6}{1}; 
see Figure~\ref{6gonk1fig}.}
\label{6gonk1tab}
\end{table}

%% file: UpperBoundsTabs/k1servants/7gonk1tab.tex
\begin{table}[h!]
\centering
    \begin{small}
    \begin{tabular}{|c|c|c|c|c|c|c|}
        \hline
        $i$ & Agent & Vertex & $t_i$ & $Q^{(i)}$ & $d(Q^{(i)} , Exit)$ & Cost \\
        \hline
        1 & $Q$ & $V_1$ & 0 & $V_1$ & 0 & 0\\
        \hline
        2 & $S_1$ & $V_2$ & 0 & $V_1$ & $e_7$ & $e_7$ \\
        \hline
        3 & $S_1$ & $V_3$ & $e_7$ & $V_7$ & $2\sin(\tfrac{3\pi}{7})$ & $e_7+2\sin(\tfrac{3\pi}{7})$ \\
        \hline
        4 & $Q$ & $V_7$ & $e_7$ & $V_7$ & 0 & $e_7$  \\
        \hline
        5 & $S_1$ & $V_4$ & $2e_7$ & $(0.33162,-0.55344)$ & 1.23838 & ${2.97391}^*$ \\
        \hline
        6 & $Q$ & $V_6$ & 2.10614 & $V_6$ & 0 & 2.10614 \\
        \hline
        7 & $Q$ & $V_5$ & 2.97391 & $V_5$ & 0 & ${2.97391}^*$ \\
        \hline
    \end{tabular}
\end{small}
\caption{Trajectories' details for the upper bound of \pe{7}{1}; 
see Figure~\ref{7gonk1fig}.}
\label{7gonk1tab}
\end{table}

%% file: UpperBoundsTabs/k1servants/8gonk1tab.tex
\begin{table}[h!]
\centering
\begin{small}
    \begin{tabular}{|c|c|c|c|c|c|c|}
        \hline
        $i$ & Agent & Vertex & $t_i$ & $Q^{(i)}$ & $d(Q^{(i)} , Exit)$ & Cost \\
        \hline
        1 & $Q$ & $V_1$ & 0 & $V_1$ & 0 & 0 \\
        \hline
        2 & $S_1$ & $V_2$ & 0 & $V_1$ & $e_8$ & $e_8$ \\
        \hline
        3 & $S_1$ & $V_3$ & $e_8$ & $V_8$ & $2\sin(\tfrac{3\pi}{8})$ & $\sqrt{2(2+\sqrt{2})}$ \\
        \hline
        4 & $Q$ & $V_8$ & $e_8$ & $V_8$ & 0 & $e_8$ \\
        \hline
        5 & $S_1$ & $V_4$ & $2e_8$ & $(0.41287,-0.49099)$ & 1.49575 & ${3.02649}^*$ \\
        \hline
        6 & $Q$ & $V_6$ & 2.18614 & $V_6$ & 0 & 2.18614  \\
        \hline
        7 &  $S_1$ & $V_5$ & $3e_8$ & $(0,-0.89004$) & 0.76039 & ${3.02649}^*$ \\
        \hline
        8 & $Q$ & $V_7$ & 3.02649 & $V_7$ & 0 & ${3.02649}^*$ \\
        \hline
    \end{tabular}
\end{small}
\caption{Trajectories' details for the upper bound of \pe{8}{1}; 
see Figure~\ref{8gonk1fig}.}
\label{8gonk1tab}
\end{table}

%% file: UpperBoundsTabs/k1servants/10gonk1tab.tex
\begin{table}[h!]
\centering
\begin{small}
    \begin{tabular}{|c|c|c|c|c|c|c|}
        \hline
        $i$ & Agent & Vertex & $t_i$ & $Q^{(i)}$ & $d(Q^{(i)} , Exit)$ & Cost \\
        \hline
        1 & $S_1$ & $V_1$ & 0 & $V_2$ & $\tfrac{1}{2}(-1+\sqrt{5})$ & $\tfrac{1}{2}(-1+\sqrt{5})$  \\
        \hline
        2 & $Q$ & $V_2$ & 0 & $V_2$ & 0 & 0 \\
        \hline
        3 & $S_1$ & $V_{10}$ & $e_{10}$ & $V_3$ & $2\sin(\tfrac{3\pi}{10})$ & $\tfrac{1}{2}(-1+\sqrt{5})+2\sin(\tfrac{3\pi}{10})$ \\
        \hline
        4 & $Q$ & $V_3$ & $e_{10}$ & $V_3$ & 0 & $\tfrac{1}{2}(-1+\sqrt{5})$ \\
        \hline
        5 & $S_1$ & $V_9$ & $2e_{10}$ & $(-0.79237,0.57569)$ & 1.97942 & ${3.21549}^{*}$ \\
        \hline
        6 & $Q$ & $V_4$ & 1.2506 & $V_4$ & 0 & 1.2506 \\
        \hline
        7 & $S_1$ & $V_8$ & $3e_{10}$ & $(-0.65457,0.01064)$ & 1.36138 & ${3.21549}^{*}$ \\
        \hline
        8 & $Q$ & $V_6$ & $4e_{10}$ & $V_6$ & 0 & $4e_{10}$ \\
        \hline
        9 & $S_1$ & $V_7$ & $4e_{10}$ & $V_6$ & $e_{10}$ & $5e_{10}$ \\
        \hline
        10 & $Q$ & $V_5$ & $5e_{10}$ & $V_5$ & 0 & $5e_{10}$ \\
        \hline
    \end{tabular}
\end{small}
\caption{Trajectories' details for the upper bound of \pe{10}{1}; 
see Figure~\ref{10gonk1fig}.}
\label{10gonk1tab}
\end{table}

%% file: UpperBoundsTabs/k1servants/11gonk1tab.tex
\begin{table}[h!]
\centering
\begin{small}
    \begin{tabular}{|c|c|c|c|c|c|c|}
        \hline
        $i$ & Agent & Vertex & $t_i$ & $Q^{(i)}$ & $d(Q^{(i)} , Exit)$ & Cost \\
        \hline
        1 & $S_1$ & $V_1$ & 0 & $V_2$ & $e_{11}$ & $e_{11}$ \\
        \hline
        2 & $Q$ & $V_2$ & 0 & $V_2$ & 0 & 0 \\
        \hline
        3 & $Q$ & $V_3$ & $e_{11}$ & $V_3$ & 0 & $e_{11}$ \\
        \hline
        4 & $S_1$ & $V_{11}$ & $e_{11}$ & $V_3$ & $2\sin(\tfrac{3\pi}{11})$ & $e_{11}+\sin(\tfrac{3\pi}{11}))$ \\
        \hline
        5 & $S_1$ & $V_{10}$ & $2e_{11}$ & $V_4$ & 1.97964 & 3.10657 \\
        \hline
        6 & $Q$ & $V_4$ & $2e_{11}$ & $V_4$ & 0 & $4\sin(\tfrac{\pi}{11})$ \\
        \hline
        7 & $S_1$ & $V_9$ & $3e_{11}$ & $(-0.81659,0.21599)$ & 1.66880 & ${3.35919}^{*}$ \\
        \hline
        8 & $Q$ & $V_5$ & 1.84769 & $V_5$ & 0 & 1.84769 \\
        \hline
        9 & $S_1$ & $V_8$ & $4e_{11}$ & $(-0.73990,-0.05996)$ & 1.10533 & ${3.35919}^{*}$ \\
        \hline
        10 & $Q$ & $V_6$ & 2.56596 & $V_6$ & 0 & 2.56596 \\
        \hline
        11 & $Q$ & $V_7$ & 3.12942 & $V_7$ & 0 & 3.12942 \\
        \hline
    \end{tabular}
\end{small}
\caption{Trajectories' details for the upper bound of \pe{11}{1}; 
see Figure~\ref{11gonk1fig}.}
\label{11gonk1tab}
\end{table}

%% file: UpperBoundsTabs/k1servants/12gonk1tab.tex
\begin{table}[h!]
\centering
\begin{small}
    \begin{tabular}{|c|c|c|c|c|c|c|}
        \hline
        $i$ & Agent & Vertex & $t_i$ & $Q^{(i)}$ & $d(Q^{(i)} , Exit)$ & Cost \\
        \hline
        1 & $S_1$ & $V_1$ & 0 & $V_2$ & $e_{12}$ & $e_{12}$ \\
        \hline
        2 & $Q$ & $V_2$ & 0 & $V_2$ & 0 & 0 \\
        \hline
        3 & $Q$ & $V_3$ & $e_{12}$ & $V_3$ & 0 & $e_{12}$ \\
        \hline
        4 & $S_1$ & $V_12$ & $e_{12}$ & $V_3$ & $\sqrt{2}$ & $e_{12}+\sqrt{2}$ \\
        \hline
        5 & $Q$ & $V_4$ & $2e_{12}$ & $V_4$ & 0 & $2e_{12}$ \\
        \hline
        6 & $S_1$ & $V_{11}$ & $2e_{12}$ & $V_4$ & $\tfrac{\sqrt{6}+\sqrt{2}}{2}$ & $\tfrac{3\sqrt{6}-\sqrt{2}}{2}$ \\
        \hline
        7 & $S_1$ & $V_{10}$ & $3e_{12}$ & $(-0.78869,0.43637)$ & 1.83220 & ${3.38511}^*$ \\
        \hline
        8 & $Q$ & $V_5$ & 1.65320 & $V_5$ & 0 & 1.65306 \\
        \hline
        9 & $S_1$ & $V_9$ & $4e_{12}$ & $(-0.68500,0.18645)$ & 1.37000 & ${3.38511}^*$\\
        \hline
        10 & $S_1$ & $V_8$ & $5e_{12}$ & $(-0.86048,-0.17482)$ & 0.77955 & 3.19631\\
        \hline
        11 & $Q$ & $V_6$ & 2.64042 & $V_6$ & 0 & 2.64042 \\
        \hline
        12 & $Q$ & $V_7$ & 3.15806 & $V_7$ & 0 & 3.15806 \\
        \hline
    \end{tabular}
\end{small}
\caption{Trajectories' details for the upper bound of \pe{12}{1}; 
see Figure~\ref{12gonk1fig}.}
\label{12gonk1tab}
\end{table}

%% file: UpperBoundsTabs/k1servants/13gonk1tab.tex
\begin{table}[h!]
\centering
\begin{small}
    \begin{tabular}{|c|c|c|c|c|c|c|}
        \hline
        $i$ & Agent & Vertex & $t_i$ & $Q^{(i)}$ & $d(Q^{(i)} , Exit)$ & Cost \\
        \hline
        1 & $Q$ & $V_1$ & 0 & $V_1$ & 0 & 0 \\
        \hline
        2 & $S_1$ & $V_2$ & 0 & $V_1$ & $e_{13}$ & $e_{13}$\\
        \hline
        3 & $Q$ & $V_{13}$ & $e_{13}$ & $V_{13}$ & 0 & $e_{13}$ \\
        \hline
        4 & $S_1$ & $V_3$ & $e_{13}$ & $V_{13}$ & $2\sin(\tfrac{3\pi}{13})$ & $e_{13}+\sin(\tfrac{3\pi}{13}))$ \\
        \hline
        5 & $S_1$ & $V_4$ & $2e_{13}$ & $V_{12}$ & $2\sin(\tfrac{5\pi}{13})$ & $2(e_{13}+\sin(\tfrac{5\pi}{13}))$ \\
        \hline
        6 & $Q$ & $V_{12}$ & $2e_{13}$ & $V_{12}$ & 0 & $2e_{13}$ \\
        \hline
        7 & $S_1$ & $V_5$ & $3e_{13}$ & $(-0.52981,-0.77980)$ & 1.92773 & ${3.36362}^*$ \\
        \hline
        8 & $Q$ & $V_{11}$ & 1.49359 & $V_{11}$ & 0 & 1.49359 \\
        \hline
        9 & $S_1$ & $V_6$ & $4e_{13}$ & $(0.22163,-0.58386)$ & 1.44909 & ${3.36362}^*$ \\
        \hline
        10 & $S_1$ & $V_7$ & $5e_{13}$ & $(-0.20052,-0.71402)$ & 0.90493 & 3.29809 \\
        \hline
        11 & $Q$ & $V_9$ & 2.66256 & $V_9$ & 0 & 2.66256 \\
        \hline
        12 & $S_1$ & $V_8$ & $6e_{13}$ & $(-0.30447,-0.84673)$ & 0.48050 & 3.35229 \\
        \hline
        13 & $Q$ & $V_{10}$ & 3.32117 & $V_{10}$ & 0 & 3.32117 \\
        \hline
    \end{tabular}
\end{small}
\caption{Trajectories' details for the upper bound of \pe{13}{1}; 
see Figure~\ref{13gonk1fig}.}
\label{13gonk1tab}
\end{table}

%% file: UpperBoundsFigs/k2servants/ALLgonk2fig.tex

\begin{figure}[H]
    \centering
    \begin{minipage}[t]{0.23\textwidth}
        \centering
        \begin{tikzpicture}[scale=0.15]
        \tikzset{
        mid arrow/.style={
            postaction={decorate, decoration={markings, mark=at position 0.6 with {\arrow{latex}}}}
        }
    }
\node (a) at (-5,8.7) {};
\fill [fill=black] (a) circle (0.105) node [above left] {$V_1$};
\node (b) at (-5,-8.7) {};
\fill [fill=black] (b) circle (0.105) node [below left] {$V_2$};
\node (c) at (10,0) {};
\fill [fill=black] (c) circle (0.105) node [right] {$V_3$};

\draw [line width=0.25mm, black] (10,0)--(-5,8.7)--(-5,-8.7)--(10,0);

\node[diamond, draw=red, fill=red, minimum size=0.15cm, inner sep=0] (f) at (-5,8.7) {};
 \node[above right] at (f.south east) {};

\node[diamond, draw=red, fill=red, minimum size=0.15cm, inner sep=0] (f) at (-5,-8.7) {};
 \node[below right] at (f.south east) {};

\node (a) at (0,0) {};
\fill [fill=blue] (a) circle (0.45) node [above right] {};

\node[diamond, draw=blue, fill=blue, minimum size=0.15cm, inner sep=0] (f) at (10,0) {};
 \node[below] at (f.south east) {};

\draw [line width=0.5mm, blue, mid arrow] (0,0) -- (10,0) ;

\draw[dashed] [line width=0.5mm, blue] [-latex] (0,0) -- (-5,8.7);
\draw[dashed] [line width=0.5mm, blue] [-latex] (0,0) -- (-5,-8.7);
        \end{tikzpicture}
\caption{A search trajectory for \pe{3}{2}. 
Trajectory details can be found in Table~\ref{3gonk2tab}.}
\label{3gonk2fig}
    \end{minipage}
    \hfill
    \begin{minipage}[t]{0.23\textwidth}
        \centering
        \begin{tikzpicture}[scale=0.15]
        \tikzset{
        mid arrow/.style={
            postaction={decorate, decoration={markings, mark=at position 0.6 with {\arrow{latex}}}}
        }
    }
   \node (a) at (0,10) {};
    \fill [fill=black] (a) circle (0.105) node [above] {$V_1$};
    \node (b) at (-10,0) {};
    \fill [fill=black] (b) circle (0.105) node [left] {$V_2$};
    \node (c) at (0,-10) {};
    \fill [fill=black] (c) circle (0.105) node [below] {$V_3$};
    \node (d) at (10,0) {};
    \fill [fill=black] (d) circle (0.105) node [right] {$V_4$};
    \draw [line width=0.25mm, black] (10,0)--(0,10)--(-10,0)--(0,-10)--(10,0);

\node[diamond, draw=red, fill=red, minimum size=0.15cm, inner sep=0] (f) at (0,10) {};
 \node[above left] at (f.south east) {};

\node (a) at (-10,0) {};
\fill [fill=red] (a) circle (0.45) node [below left] {};

\node (a) at (0,-7.07) {};
\fill [fill=blue] (a) circle (0.45) node [above right] {};

\node (a) at (0,-10) {};
\fill [fill=blue] (a) circle (0.45) node [below right] {};

\node[diamond, draw=blue, fill=blue, minimum size=0.15cm, inner sep=0] (f) at (10,0) {};
 \node[below right] at (f.south east) {};

\draw [line width=0.5mm, blue] (0,-7.07) -- (0,-10);
\draw [line width=0.5mm, blue, mid arrow] (0,-10) -- (10,0);
\draw[dashed] [line width=0.5mm, blue] [-latex] (0,-7.07) -- (0,10);

        \end{tikzpicture}
\caption{A search trajectory for \pe{4}{2}. 
Trajectory details can be found in Table~\ref{4gonk2tab}.}
\label{4gonk2fig}
    \end{minipage}
    \hfill
    \begin{minipage}[t]{0.23\textwidth}
        \centering
        \begin{tikzpicture}[scale=0.15]
        \tikzset{
        mid arrow/.style={
            postaction={decorate, decoration={markings, mark=at position 0.6 with {\arrow{latex}}}}
        }
    }
        \node (a) at (3.1,9.5) {};
        \fill [fill=black] (a) circle (0.105) node [above right] {$V_1$};
        \node (b) at (-8.1,5.9) {};
        \fill [fill=black] (b) circle (0.105) node [left] {$V_2$};
        \node (c) at (-8.1,-5.9) {};
        \fill [fill=black] (c) circle (0.105) node [left] {$V_3$};
        \node (d) at (3.1,-9.5) {};
        \fill [fill=black] (d) circle (0.105) node [below left] {$V_4$};
        \node (e) at (10,0) {};
        \fill [fill=black] (e) circle (0.105) node [right] {$V_5$};

        \draw [line width=0.25mm, black] (10,0)--(3.1,9.5)--(-8.1,5.9)--(-8.1,-5.9)--(3.1,-9.5)--(10,0);

\node[diamond, draw=red, fill=red, minimum size=0.15cm, inner sep=0] (f) at (3.1,9.5) {};
 \node[above left] at (f.south east) {};

\node[diamond, draw=red, fill=red, minimum size=0.15cm, inner sep=0] (f) at (10,0) {};
 \node[right] at (f.south east) {};

\draw [line width=0.5mm, red, mid arrow] (3.1,9.5)  -- (10,0);
\node (a) at (-8.1,5.9) {};
\fill [fill=red] (a) circle (0.45) node [below right] {};

\node (a) at (-8.1,-5.9) {};
\fill [fill=blue] (a) circle (0.45) node [below] {};

\node (a) at (3.1,-2.25) {};
\fill [fill=blue] (a) circle (0.45) node [above] {};

\node[diamond, draw=blue, fill=blue, minimum size=0.15cm, inner sep=0] (f) at (3.1,-9.5) {};
 \node[below right] at (f.south east) {};

\draw [line width=0.5mm, blue, mid arrow] (-8.1,-5.9) -- (3.1,-2.25);

\draw [line width=0.5mm, blue, mid arrow] (3.1,-2.25)  -- (3.1,-9.5);
\draw[dashed] [line width=0.5mm, blue] [-latex] (-8.1,-5.9) -- (3.1,9.5);

\draw[dashed] [line width=0.5mm, blue] [-latex] (3.1,-2.25) -- (10,0);

        \end{tikzpicture}
\caption{A search trajectory for \pe{5}{2}. 
Trajectory details can be found in Table~\ref{5gonk2tab}.}
\label{5gonk2fig}
    \end{minipage}
   \end{figure}


\begin{figure}[H]
    \centering
    \begin{minipage}[t]{0.23\textwidth}
        \centering
        \begin{tikzpicture}[scale=0.15]
        \tikzset{
        mid arrow/.style={
            postaction={decorate, decoration={markings, mark=at position 0.6 with {\arrow{latex}}}}
        }
    }
\node (a) at (5,8.7) {};
\fill [fill=black] (a) circle (0.105) node [above right] {$V_1$};
\node (b) at (-5,8.7) {};
\fill [fill=black] (b) circle (0.105) node [above left] {$V_2$};
\node (c) at (-10,0) {};
\fill [fill=black] (c) circle (0.105) node [left] {$V_3$};
\node (d) at (-5,-8.7) {};
\fill [fill=black] (d) circle (0.105) node [below left] {$V_4$};
\node (e) at (5,-8.7) {};
\fill [fill=black] (e) circle (0.105) node [below right] {$V_5$};
\node (f) at (10,0) {};
\fill [fill=black] (f) circle (0.105) node [right] {$V_6$};

\draw [line width=0.25mm, black] (10,0)--(5,8.7)--(-5,8.7)--(-10,0)--(-5,-8.7)--(5,-8.7)--(10,0);

\node (a) at (5,8.7) {};
\fill [fill=red] (a) circle (0.45) node [below left] {};

\node[diamond, draw=red, fill=red, minimum size=0.15cm, inner sep=0] (f) at (-5,8.7) {};
 \node[above right] at (f.south east) {};

\draw [line width=0.5mm, red, mid arrow] (5,8.7)  -- (-5,8.7);
\node (a) at (-10,0) {};
\fill [fill=red] (a) circle (0.45) node [below left] {};

\node[diamond, draw=red, fill=red, minimum size=0.15cm, inner sep=0] (f) at (-5,-8.7) {};
 \node[below right] at (f.south east) {};

\draw [line width=0.5mm, red, mid arrow] (-10,0)  -- (-5,-8.7);
\node (a) at (5,-8.7) {};
\fill [fill=blue] (a) circle (0.45) node [below left] {};

\node (a) at (0,0) {};
\fill [fill=blue] (a) circle (0.45) node [above right] {};

\node[diamond, draw=blue, fill=blue, minimum size=0.15cm, inner sep=0] (f) at (10,0) {};
 \node[below right] at (f.south east) {};

\draw [line width=0.5mm, blue, mid arrow] (5,-8.7) -- (0,0);

\draw [line width=0.5mm, blue, mid arrow] (0,0)  -- (10,0);

\draw[dashed] [line width=0.5mm, blue] [-latex] (0,0) -- (-5,8.7);

\draw[dashed] [line width=0.5mm, blue] [-latex] (0,0) -- (-5,-8.7);
        \end{tikzpicture}
\caption{A search trajectory for \pe{6}{2}. 
Trajectory details can be found in Table~\ref{6gonk2tab}.}
\label{6gonk2fig}
    \end{minipage}
    \hfill
    \begin{minipage}[t]{0.23\textwidth}
        \centering
        \begin{tikzpicture}[scale=0.15]
        \tikzset{
        mid arrow/.style={
            postaction={decorate, decoration={markings, mark=at position 0.6 with {\arrow{latex}}}}
        }
    }
\node (a) at (6.2,7.8) {};
        \fill [fill=black] (a) circle (0.105) node [above right] {$V_1$};
        \node (b) at (-2.2,9.7) {};
        \fill [fill=black] (b) circle (0.105) node [above] {$V_2$};
        \node (c) at (-9,4.3) {};
        \fill [fill=black] (c) circle (0.105) node [above left] {$V_3$};
        \node (d) at (-9,-4.3) {};
        \fill [fill=black] (d) circle (0.105) node [below left] {$V_4$};
        \node (e) at (-2.2,-9.7) {};
        \fill [fill=black] (e) circle (0.105) node [below] {$V_5$};
        \node (f) at (6.2,-7.8) {};
        \fill [fill=black] (f) circle (0.105) node [below right] {$V_6$};
        \node (g) at (10,0) {};
        \fill [fill=black] (g) circle (0.105) node [right] {$V_7$};

        \draw [line width=0.25mm, black] (10,0)--(6.2,7.8)--(-2.2,9.7)--(-9,4.3)--(-9,-4.3)--(-2.2,-9.7)--(6.2,-7.8)--(10,0);

\node (a) at (-9,4.3) {};
\fill [fill=red] (a) circle (0.45) node [below right] {};

\node[diamond, draw=red, fill=red, minimum size=0.15cm, inner sep=0] (f) at (-2.2,9.7) {};
 \node[above right] at (f.south east) {};

\draw [line width=0.5mm, red, mid arrow] (-9,4.3)  -- (-2.2,9.7);
\node (a) at (-9,-4.3) {};
\fill [fill=red] (a) circle (0.45) node [above right] {};

\node[diamond, draw=red, fill=red, minimum size=0.15cm, inner sep=0] (f) at (-2.2,-9.7) {};
 \node[below right] at (f.south east) {};

\draw [line width=0.5mm, red, mid arrow] (-9,-4.3)  -- (-2.2,-9.7);
\node (a) at (6.2,-7.8) {};
\fill [fill=blue] (a) circle (0.45) node [above left] {};

\node (a) at (5.9,0) {};
\fill [fill=blue] (a) circle (0.45) node [above right] {};

\node (a) at (10,0) {};
\fill [fill=blue] (a) circle (0.45) node [below right] {};

\node[diamond, draw=blue, fill=blue, minimum size=0.15cm, inner sep=0] (f) at (6.2,7.8) {};
 \node[right] at (f.south east) {};

\draw [line width=0.5mm, blue, mid arrow] (6.2,-7.8) -- (5.9,0);

\draw [line width=0.5mm, blue, mid arrow] (5.9,0)  -- (10,0);
\draw [line width=0.5mm, blue, mid arrow] (10,0)  -- (6.2,7.8);
\draw[dashed] [line width=0.5mm, blue] [-latex] (5.9,0) -- (-2.2,9.7);

\draw[dashed] [line width=0.5mm, blue] [-latex] (5.9,0) -- (-2.2, -9.7);
       \end{tikzpicture}
\caption{A search trajectory for \pe{7}{2}. 
Trajectory details can be found in Table~\ref{7gonk2tab}.}
\label{7gonk2fig}
    \end{minipage}
    \hfill
    \begin{minipage}[t]{0.23\textwidth}
        \centering
        \begin{tikzpicture}[scale=0.15]
        \tikzset{
        mid arrow/.style={
            postaction={decorate, decoration={markings, mark=at position 0.6 with {\arrow{latex}}}}
        }
    }
   \node (a) at (7.1,7.1) {};
        \fill [fill=black] (a) circle (0.105) node [above right] {$V_1$};
        \node (b) at (0,10) {};
        \fill [fill=black] (b) circle (0.105) node [above] {$V_2$};
        \node (c) at (-7.1,7.1) {};
        \fill [fill=black] (c) circle (0.105) node [above left] {$V_3$};
        \node (d) at (-10,0) {};
        \fill [fill=black] (d) circle (0.105) node [left] {$V_4$};
        \node (e) at (-7.1,-7.1) {};
        \fill [fill=black] (e) circle (0.105) node [below left] {$V_5$};
        \node (f) at (0, -10) {};
        \fill [fill=black] (f) circle (0.105) node [below] {$V_6$};
        \node (g) at (7.1,-7.1) {};
        \fill [fill=black] (g) circle (0.105) node [below right] {$V_7$};
        \node (h) at (10,0) {};
        \fill [fill=black] (h) circle (0.105) node [right] {$V_8$};

        \draw [line width=0.25mm, black] (10,0)--(7.1,7.1)--(0,10)--(-7.1,7.1)--(-10,0)--(-7.1,-7.1)--(0, -10)--(7.1,-7.1)--(10,0);

\node (a) at (7.1,7.1) {};
\fill [fill=red] (a) circle (0.45) node [below left] {};

\node (a) at (10,0) {};
\fill [fill=red] (a) circle (0.45) node [below right] {};

\node[diamond, draw=red, fill=red, minimum size=0.15cm, inner sep=0] (f) at (7.1,-7.1) {};
 \node[below] at (f.south east) {};
 
\draw [line width=0.5mm, red, mid arrow] (7.1,7.1)  -- (10,0);
\draw [line width=0.5mm, red, mid arrow] (10,0) -- (7.1,-7.1);
\node (a) at (-7.1,7.1) {};
\fill [fill=red] (a) circle (0.45) node [below right] {};

\node (a) at (-10,0) {};
\fill [fill=red] (a) circle (0.45) node [below left] {};

\node[diamond, draw=red, fill=red, minimum size=0.15cm, inner sep=0] (f) at (-7.1,-7.1) {};
 \node[below] at (f.south east) {};

\draw [line width=0.5mm, red, mid arrow] (-7.1,7.1)  -- (-10,0);
\draw [line width=0.5mm, red, mid arrow] (-10,0) -- (-7.1,-7.1);
\node (a) at (0,10) {};
\fill [fill=blue] (a) circle (0.45) node [above left] {};

\node (a) at (0,2.35) {};
\fill [fill=blue] (a) circle (0.45) node [right] {};

\node (a) at (0,-5.31) {};
\fill [fill=blue] (a) circle (0.45) node [above right] {};

\node (a) at (0,-10) {};
\fill [fill=blue] (a) circle (0.45) node [below left] {};

\draw [line width=0.5mm, blue, mid arrow] (0,10) -- (0,2.35);

\draw [line width=0.5mm, blue, mid arrow] (0,2.35)  -- (0,-5.31);

\draw [line width=0.5mm, blue, mid arrow] (0,-5.31)  -- (0,-10);

\draw[dashed] [line width=0.5mm, blue] [-latex] (0,-5.31) -- (7.1,-7.1);

\draw[dashed] [line width=0.5mm, blue] [-latex] (0,-5.31) -- (-7.1,-7.1);
        \end{tikzpicture}
\caption{A search trajectory for \pe{8}{2}. 
Trajectory details can be found in Table~\ref{8gonk2tab}.}
\label{8gonk2fig}
    \end{minipage}
       \end{figure}


\begin{figure}[H]
    \centering
    \begin{minipage}[t]{0.23\textwidth}
        \centering
        \begin{tikzpicture}[scale=0.15]
        \tikzset{
        mid arrow/.style={
            postaction={decorate, decoration={markings, mark=at position 0.6 with {\arrow{latex}}}}
        }
    }
       \node (a) at (7.7,6.4) {};
        \fill [fill=black] (a) circle (0.105) node [above right] {$V_1$};
        \node (b) at (1.7,9.8) {};
        \fill [fill=black] (b) circle (0.105) node [above] {$V_2$};
        \node (c) at (-5,8.7) {};
        \fill [fill=black] (c) circle (0.105) node [above left] {$V_3$};
        \node (d) at (-9.4,3.4) {};
        \fill [fill=black] (d) circle (0.105) node [above left] {$V_4$};
        \node (e) at (-9.4,-3.4) {};
        \fill [fill=black] (e) circle (0.105) node [below left] {$V_5$};
        \node (f) at (-5,-8.7) {};
        \fill [fill=black] (f) circle (0.105) node [below left] {$V_6$};
        \node (g) at (1.7,-9.8) {};
        \fill [fill=black] (g) circle (0.105) node [below] {$V_7$};
        \node (h) at (7.7,-6.4) {};
        \fill [fill=black] (h) circle (0.105) node [below right] {$V_8$};
        \node (i) at (10,0) {};
        \fill [fill=black] (i) circle (0.105) node [right] {$V_9$};

        \draw [line width=0.25mm, black] (10,0)--(7.7,6.4)--(1.7,9.8)--(-5,8.7)--(-9.4,3.4)--(-9.4,-3.4)--(-5,-8.7)--(1.7,-9.8)--(7.7,-6.4)--(10,0);

\node (a) at (-9.4,3.4) {};
\fill [fill=red] (a) circle (0.45) node [below right] {};

\node[diamond, draw=red, fill=red, minimum size=0.15cm, inner sep=0] (f) at (-5,8.7) {};
 \node[above left] at (f.south east) {};

\node[diamond, draw=red, fill=red, minimum size=0.15cm, inner sep=0] (f) at (1.7,9.8) {};
 \node[above right] at (f.south east) {};

\draw [line width=0.5mm, red, mid arrow] (-9.4,3.4)  -- (-5,8.7);

\draw [line width=0.5mm, red, mid arrow] (-5,8.7) -- (1.7,9.8);

\node (a) at (-9.4,-3.4) {};
\fill [fill=red] (a) circle (0.45) node [above right] {};

\node[diamond, draw=red, fill=red, minimum size=0.15cm, inner sep=0] (f) at (-5,-8.7) {};
 \node[below right] at (f.south east) {};

\node[diamond, draw=red, fill=red, minimum size=0.15cm, inner sep=0] (f) at (1.7,-9.8) {};
 \node[below right] at (f.south east) {};

\draw [line width=0.5mm, red, mid arrow] (-9.4,-3.4)  -- (-5,-8.7);

\draw [line width=0.5mm, red, mid arrow] (-5,-8.7) -- (1.7,-9.8);

\node (a) at (7.7,-6.4) {};
\fill [fill=blue] (a) circle (0.45) node [above left] {};

\node (a) at (8.2,0) {};
\fill [fill=blue] (a) circle (0.45) node [above] {};

\node (a) at (10,0) {};
\fill [fill=blue] (a) circle (0.45) node [below right] {};

\node (a) at (3.7,0) {};
\fill [fill=blue] (a) circle (0.45) node [below left] {};

\node (a) at (7.7,6.4) {};
\fill [fill=blue] (a) circle (0.45) node [right] {};

\draw [line width=0.5mm, blue, mid arrow] (7.7,-6.4) -- (8.2,0.1);

\draw [line width=0.5mm, blue] (8.2,0.2)  -- (10,0.2);

\draw [line width=0.5mm, blue] (10,-0.2)  -- (3.7,-0.2);

\draw [line width=0.5mm, blue, mid arrow] (3.7,0)  -- (7.7,6.4);

\draw[dashed] [line width=0.5mm, blue] [-latex] (8.2,0) -- (-5,8.7);

\draw[dashed] [line width=0.5mm, blue] [-latex] (8.2,0) -- (-5,-8.7);

\draw[dashed] [line width=0.5mm, blue] [-latex] (3.7,0) -- (1.7,9.8);

\draw[dashed] [line width=0.5mm, blue] [-latex] (3.7,0) -- (1.7,-9.8);
        \end{tikzpicture}
    \caption{A search trajectory for \pe{9}{2}. 
Trajectory details can be found in Table~\ref{9gonk2tab}.}
\label{9gonk2fig}
    \end{minipage}
    \hfill
    \begin{minipage}[t]{0.23\textwidth}
        \centering
        \begin{tikzpicture}[scale=0.15]
        \tikzset{
        mid arrow/.style={
            postaction={decorate, decoration={markings, mark=at position 0.6 with {\arrow{latex}}}}
        }
    }
\node (a) at (8.1,5.9) {};
        \fill [fill=black] (a) circle (0.105) node [above right] {$V_1$};
        \node (b) at (3.1,9.5) {};
        \fill [fill=black] (b) circle (0.105) node [above] {$V_2$};
        \node (c) at (-3.1,9.5) {};
        \fill [fill=black] (c) circle (0.105) node [above] {$V_3$};
        \node (d) at (-8.1,5.9) {};
        \fill [fill=black] (d) circle (0.105) node [above left] {$V_4$};
        \node (e) at (-10,0) {};
        \fill [fill=black] (e) circle (0.105) node [left] {$V_5$};
        \node (f) at (-8.1,-5.9) {};
        \fill [fill=black] (f) circle (0.105) node [below left] {$V_6$};
        \node (g) at (-3.1,-9.5) {};
        \fill [fill=black] (g) circle (0.105) node [below] {$V_7$};
        \node (h) at (3.1,-9.5) {};
        \fill [fill=black] (h) circle (0.105) node [below] {$V_8$};
        \node (i) at (8.1,-5.9) {};
        \fill [fill=black] (i) circle (0.105) node [below right] {$V_9$};
        \node (j) at (10,0) {};
        \fill [fill=black] (j) circle (0.105) node [right] {$V_{10}$};

        \draw [line width=0.25mm, black] (10,0)--(8.1,5.9)--(3.1,9.5)--(-3.1,9.5)--(-8.1,5.9)--(-10,0)--(-8.1,-5.9)--(-3.1,-9.5)--(3.1,-9.5)--(8.1,-5.9)--(10,0);

\node (a) at (-10,0) {};
\fill [fill=red] (a) circle (0.45) node [right] {};

\node[diamond, draw=red, fill=red, minimum size=0.15cm, inner sep=0] (f) at (-8.1,5.9) {};
 \node[left] at (f.south east) {};

\node[diamond, draw=red, fill=red, minimum size=0.15cm, inner sep=0] (f) at (-3.1,9.5) {};
 \node[above right] at (f.south east) {};

\draw [line width=0.5mm, red, mid arrow] (-10,0)  -- (-8.1,5.9);

\draw [line width=0.5mm, red, mid arrow] (-8.1,5.9) -- (-3.1,9.5);

\node (a) at (-8.1,-5.9) {};
\fill [fill=red] (a) circle (0.45) node [above right] {};

\node (a) at (-3.1,-9.5) {};
\fill [fill=red] (a) circle (0.45) node [above right] {};

\node[diamond, draw=red, fill=red, minimum size=0.15cm, inner sep=0] (f) at (3.1,-9.5) {};
 \node[below right] at (f.south east) {};

\draw [line width=0.5mm, red, mid arrow] (-8.1,-5.9)  -- (-3.1,-9.5);

\draw [line width=0.5mm, red, mid arrow] (-3.1,-9.5) -- (3.1,-9.5);

\node (a) at (8.1,-5.9) {};
\fill [fill=blue] (a) circle (0.45) node [above left] {};

\node (a) at (8.7,0.28) {};
\fill [fill=blue] (a) circle (0.45) node [below left] {};

\node (a) at (10,0) {};
\fill [fill=blue] (a) circle (0.45) node [below right] {};

\node (a) at (5.5,1.95) {};
\fill [fill=blue] (a) circle (0.45) node [right] {};

\node (a) at (8.1,5.9) {};
\fill [fill=blue] (a) circle (0.45) node [right] {};

\node (a) at (3.1,9.5) {};
\fill [fill=blue] (a) circle (0.45) node [below left] {};

\draw [line width=0.5mm, blue, mid arrow] (8.1,-5.9) -- (8.7,0.3);

\draw [line width=0.5mm, blue] (8.7,0.2)  -- (10,-0.1);

\draw [line width=0.5mm, blue] (10,0.13)  -- (5.5,1.93);

\draw [line width=0.5mm, blue, mid arrow] (5.5,1.8)  -- (8.1,5.9);

\draw [line width=0.5mm, blue, mid arrow] (8.1,5.9) -- (3.1,9.5);

\draw[dashed] [line width=0.5mm, blue] [-latex] (8.7,0.3) -- (-8.1,5.9);

\draw[dashed] [line width=0.5mm, blue] [-latex] (5.5,1.85) -- (-3.1,9.5);

\draw[dashed] [line width=0.5mm, blue] [-latex] (5.5,1.85) -- (3.1,-9.5);
        \end{tikzpicture}
\caption{A search trajectory for \pe{10}{2}. 
Trajectory details can be found in Table~\ref{10gonk2tab}.}
\label{10gonk2fig}
    \end{minipage}    
    \hfill
    \begin{minipage}[t]{0.23\textwidth}
        \centering     
        \begin{tikzpicture}[scale=0.15]
        \tikzset{
        mid arrow/.style={
            postaction={decorate, decoration={markings, mark=at position 0.6 with {\arrow{latex}}}}
        }
    }
           \node (a) at (8.4,5.4) {};
            \fill [fill=black] (a) circle (0.105) node [above right] {$V_1$};
            \node (b) at (4.2,9.1) {};
            \fill [fill=black] (b) circle (0.105) node [above right] {$V_2$};
            \node (c) at (-1.4,9.9) {};
            \fill [fill=black] (c) circle (0.105) node [above] {$V_3$};
            \node (d) at (-6.5,7.6) {};
            \fill [fill=black] (d) circle (0.105) node [above left] {$V_4$};
            \node (e) at (-9.6,2.8) {};
            \fill [fill=black] (e) circle (0.105) node [left] {$V_5$};
            \node (f) at (-9.6,-2.8) {};
            \fill [fill=black] (f) circle (0.105) node [left] {$V_6$};
            \node (g) at (-6.5,-7.6) {};
            \fill [fill=black] (g) circle (0.105) node [below left] {$V_7$};
            \node (h) at (-1.4,-9.9) {};
            \fill [fill=black] (h) circle (0.105) node [below] {$V_8$};
            \node (i) at (4.2,-9.1) {};
            \fill [fill=black] (i) circle (0.105) node [below right] {$V_9$};
            \node (j) at (8.4,-5.4) {};
            \fill [fill=black] (j) circle (0.105) node [below right] {$V_{10}$};
            \node (k) at (10,0) {};
            \fill [fill=black] (k) circle (0.105) node [right] {$V_{11}$};

            \draw [line width=0.25mm, black] (10,0)--(8.4,5.4)--(4.2,9.1)--(-1.4,9.9)--(-6.5,7.6)--(-9.6,2.8)--(-9.6,-2.8)--(-6.5,-7.6)--(-1.4,-9.9)--(4.2,-9.1)--(8.4,-5.4)--(10,0);

\node (d) at (-9.59,2.82) {};
\fill [fill=red] (d) circle (0.45) node [right] {};

\node (c) at (-6.55,7.56) {};
\fill [fill=red] (c) circle (0.45) node [below right] {};    
\node (b) at (-1.42,9.9) {};
\fill [fill=red] (b) circle (0.45) node [above left] {};   

\draw [line width=0.5mm, red, mid arrow] (-9.59,2.82) -- (-6.55,7.56);
\draw [line width=0.5mm, red, mid arrow] (-6.55,7.56) -- (-1.42,9.9);

\node (d) at (-9.59,-2.82) {};
\fill [fill=red] (d) circle (0.45) node [right] {};

\node[diamond, draw=red, fill=red, minimum size=0.15cm, inner sep=0] (f) at (-6.55,-7.56) {};
 \node[below] at (f.south east) {};

\node[diamond, draw=red, fill=red, minimum size=0.15cm, inner sep=0] (f) at (-1.42,-9.9) {};
 \node[below right] at (f.south east) {};
 
 \node (b) at (4.15,-9.1) {};
\fill [fill=red] (b) circle (0.45) node [above] {};  

\draw [line width=0.5mm, red, mid arrow] (-9.59,-2.82) -- (-6.55,-7.56);
\draw [line width=0.5mm, red, mid arrow] (-6.55,-7.56) -- (-1.42,-9.9);
\draw [line width=0.5mm, red, mid arrow] (-1.42,-9.9) -- (4.15,-9.1);

\node (d) at (4.15,9.1) {};
\fill [fill=blue] (d) circle (0.45) node [below left] {};

\node (d) at (7.5,4.5) {};
\fill [fill=blue] (d) circle (0.45) node [left] {};

\node (d) at (8.41,5.41) {};
\fill [fill=blue] (d) circle (0.45) node [right] {};

\node (d) at (7.9,0.2) {};
\fill [fill=blue] (d) circle (0.45) node [above left] {};

\node (d) at (8.41,-5.41)  {};
\fill [fill=blue] (d) circle (0.45) node [left] {};

\node (d) at (10,0)  {};
\fill [fill=blue] (d) circle (0.45) node [below right] {};

\draw [line width=0.5mm, blue, mid arrow] (4.15,9.1) -- (7.5,4.5);
\draw [line width=0.5mm, blue] (7.5,4.5) -- (8.41,5.41) node[currarrow,blue,pos=0.5,xscale=1,sloped,scale=0.75] {};
\draw [line width=0.5mm, blue, mid arrow] (8.41,5.41) -- (7.9,0.2);
\draw [line width=0.5mm, blue, mid arrow] (7.9,0.2) -- (8.41,-5.41);
\draw [line width=0.5mm, blue, mid arrow] (8.41,-5.41) -- (10,0);

\draw[dashed] [line width=0.5mm, blue] [-latex] (7.5,4.5) -- (-6.5,-7.6);

\draw[dashed] [line width=0.5mm, blue] [-latex] (7.9,0.2) -- (-1.4,-9.9);
         
        
                \end{tikzpicture}
\caption{A search trajectory for \pe{11}{2}. 
Trajectory details can be found in Table~\ref{11gonk2tab}.}
\label{11gonk2fig}
    \end{minipage}
    \end{figure}    

%% file: UpperBoundsTabs/k2servants/3gonk2tab.tex
\begin{table}[h!]
\centering
\begin{small}
    \begin{tabular}{|c|c|c|c|c|c|c|}
        \hline
        $i$ & Agent & Vertex & $t_i$ & $Q^{(i)}$ & $d(Q^{(i)} , Exit)$ & Cost \\
        \hline
        1 & $S_1$ & $V_1$ & 0 & $\cal{O}$ & 1 & 1 $^*$\\
        \hline
        2 & $S_2$ & $V_2$ & 0 & $\cal{O}$ & 1 & 1 $^*$\\
        \hline
        3 & $Q$ & $V_3$ & 1 & $V_3$ & 0 & 1 $^*$\\
        \hline
    \end{tabular}
\end{small}
\caption{Trajectories' details for the upper bound of \pe{3}{2}; 
see Figure~\ref{3gonk2fig}.}
\label{3gonk2tab}
\end{table}

%% file: UpperBoundsTabs/k2servants/4gonk2tab.tex
\begin{table}[h!]
\centering
\begin{small}
    \begin{tabular}{|c|c|c|c|c|c|c|}
        \hline
        $i$ & Agent & Vertex & $t_i$ & $Q^{(i)}$ & $d(Q^{(i)} , Exit)$ & Cost \\
        \hline
        1 & $S_1$ & $V_1$ & 0 & $(0,-\sqrt{2}/2)$ & $1+\sqrt{2}/2$ & $1+\sqrt{2}/2$ $^*$\\
        \hline
        2 & $S_2$ & $V_2$ & 0 & $(0,-\sqrt{2}/2)$ & $\sqrt{3/2}$ & $\sqrt{3/2}$\\
        \hline
        3 & $Q$ & $V_3$ & $1-\sqrt{2}/2$ & $V_3$ & 0 & $1-\sqrt{2}/2$ \\
        \hline
        4 & $Q$ & $V_4$ & $1+\sqrt{2}/2$ & $V_4$ & 0 & $1+\sqrt{2}/2$ $^*$\\
        \hline
    \end{tabular}
\end{small}    
\caption{Trajectories' details for the upper bound of \pe{4}{2}; 
see Figure~\ref{4gonk2fig}.}
\label{4gonk2tab}
\end{table}

%% file: UpperBoundsTabs/k2servants/5gonk2tab.tex
\begin{table}[h!]
\centering
\begin{small}
    \begin{tabular}{|c|c|c|c|c|c|c|}
        \hline
        $i$ & Agent & Vertex & $t_i$ & $Q^{(i)}$ & $d(Q^{(i)} , Exit)$ & Cost \\
        \hline
        1 & $S_1$ & $V_1$ & 0 & $V_3$ & $\sqrt{(5+\sqrt{5})/2}$ & $\sqrt{(5+\sqrt{5})/2}$ $^*$\\
        \hline
        2 & $S_2$ & $V_2$ & 0 & $V_3$ & $e_5$ & $e_5$\\
        \hline
        3 & $Q$ & $V_3$ & 0 & $V_3$ & 0 & 0 \\
        \hline
        4 & $S_1$ & $V_5$ & $e_5$ & $\left(\tfrac{\sqrt{5}-1}{4},\tfrac{\sqrt{5-2\sqrt{5}}}{2}-\sqrt{\tfrac{5-\sqrt{5}}{8}}\right)$ & $\sqrt{(5+\sqrt{5})/2}-e_5$ & $\sqrt{(5+\sqrt{5})/2}$ $^*$\\
        \hline
        5 & $Q$ & $V_4$ & $\sqrt{(5+\sqrt{5})/2}$ & $V_4$  & 0 & $\sqrt{(5+\sqrt{5})/2}$ $^*$\\
        \hline
    \end{tabular}
\end{small}
\caption{Trajectories' details for the upper bound of \pe{5}{2}; 
see Figure~\ref{5gonk2fig}.}
\label{5gonk2tab}
\end{table}

%% file: UpperBoundsTabs/k2servants/6gonk2tab.tex
\begin{table}[h!]
\centering
\begin{small}
    \begin{tabular}{|c|c|c|c|c|c|c|}
        \hline
        $i$ & Agent & Vertex & $t_i$ & $Q^{(i)}$ & $d(Q^{(i)} , Exit)$ & Cost \\
        \hline
        1 & $S_1$ & $V_1$ & 0 & $V_5$ & $\sqrt{3}$ & $\sqrt{3}$ \\
        \hline
        2 & $S_2$ & $V_3$ & 0 & $V_5$ & $\sqrt{3}$ & $\sqrt{3}$\\
        \hline
        3 & $Q$ & $V_5$ & 0 & $V_5$ & 0 & 0 \\
        \hline
        4 & $S_1$ & $V_2$ & 1 &  $\cal{O}$  & 1 & 2 $^*$\\
        \hline
        5 & $S_2$ & $V_4$ & 1 & $\cal{O}$ & 1 & 2 $^*$\\
        \hline
        6 & $Q$ & $V_6$ & 2 &  $V_6$  & 0 & 2 $^*$\\
        \hline
    \end{tabular}
\end{small}
\caption{Trajectories' details for the upper bound of \pe{6}{2}; 
see Figure~\ref{6gonk2fig}.}
\label{6gonk2tab}
\end{table}

%% file: UpperBoundsTabs/k2servants/7gonk2tab.tex
\begin{table}[h!]
\centering
\begin{small}
    \begin{tabular}{|c|c|c|c|c|c|c|}
        \hline
        $i$ & Agent & Vertex & $t_i$ & $Q^{(i)}$ & $d(Q^{(i)} , Exit)$ & Cost \\
        \hline
        1 & $S_1$ & $V_4$ & 0 & $V_6$ & $\sqrt{2-2\cos(4\pi/7)}$ & $\sqrt{2-2\cos(4\pi/7)}$ \\
        \hline
        2 & $Q$ & $V_6$ & 0 & $V_6$ & 0 & 0 \\
        \hline
        3 & $S_2$ & $V_3$ & 0 & $V_6$ & $\sqrt{2-2\cos(6\pi/7)}$ & $\sqrt{2-2\cos(6\pi/7)}$ \\
        \hline
        4 & $S_1$ & $V_5$ & $e_7$ & $(x_{2,7},0)$ & $e_7+1-x_{2,7}$ & $2e_7+1-x_{2,7}$ $^*$\\
        \hline
        5 & $S_2$ & $V_2$ & $e_7$ &  $(x_{2,7},0)$  &  $e_7+1-x_{2,7}$ & $2e_7+1-x_{2,7}$ $^*$\\
        \hline
        6 & $Q$ & $V_7$ & $e_7+1-x_{2,7}$ &  $V_7$  & 0 & $e_7+1-x_{2,7}$ \\
        \hline
        7 & $Q$ & $V_1$ & $2e_7$ &  $V_1$  & 0 & $2e_7+1-x_{2,7}$ $^*$\\
        \hline
    \end{tabular}
\end{small}
\caption{Trajectories' details for the upper bound of \pe{7}{2}; 
see Figure~\ref{7gonk2fig}.}
\label{7gonk2tab}
\end{table}

%% file: UpperBoundsTabs/k2servants/8gonk2tab.tex
\begin{table}[h!]
\centering
\begin{small}
    \begin{tabular}{|c|c|c|c|c|c|c|}
        \hline
        $i$ & Agent & Vertex & $t_i$ & $Q^{(i)}$ & $d(Q^{(i)} , Exit)$ & Cost \\
        \hline
        1 & $S_1$ & $V_1$ & 0 & $V_2$ & $e_8$ & $e_8$ \\
        \hline
        2 & $S_2$ & $V_3$ & 0 & $V_2$ & $e_8$ & $e_8$ \\
        \hline
        3 & $Q$ & $V_2$ & 0 & $V_2$ & 0 & 0 \\
        \hline
        4 & $S_1$ & $V_8$ & $e_8$ & $(0,1-e_8)$ & $w_{2,8}$ & $e_8+w_{2,8}$ \\
        \hline
        5 & $S_2$ & $V_4$ & $e_8$ & $(0,1-e_8)$ & $w_{2,8}$ & $e_8+w_{2,8}$ \\
        \hline
        6 & $S_1$ & $V_7$ & $2e_8$ & $(0,1-2e_8)$ & $z_{2,8}$ & $2e_8+z_{2,8}$ $^*$\\
        \hline
        7 & $S_2$ & $V_5$ & $2e_8$ & $(0,1-2e_8)$ & $z_{2,8}$ & $2e_8+z_{2,8}$ $^*$\\
        \hline
        8 & $Q$ & $V_6$ & 2 & $V_6$ & 0 & 2 \\
        \hline
    \end{tabular}
\end{small}
\caption{Trajectories' details for the upper bound of \pe{8}{2}; 
see Figure~\ref{8gonk2fig}.}
\label{8gonk2tab}
\end{table}

%% file: UpperBoundsTabs/k2servants/9gonk2tab.tex
\begin{table}[h!]
\centering
\begin{small}
    \begin{tabular}{|c|c|c|c|c|c|c|}
        \hline
        $i$ & Agent & Vertex & $t_i$ & $Q^{(i)}$ & $d(Q^{(i)} , Exit)$ & Cost \\
        \hline
        1 & $S_1$ & $V_4$ & 0 & $V_8$ & $\sqrt{2-2\cos\left(\tfrac{8\pi}{9}\right)}$ & $\approx1.969616$\\
        \hline
        2 & $S_2$ & $V_5$ & 0 & $V_8$ & $\sqrt{3}$ & $\approx1.73205$ \\
        \hline
        3 & $Q$ & $V_8$ & 0 & $V_8$ & 0 & 0 \\
        \hline
        4 & $S_1$ & $V_3$ & $e_9$ & $(x_{2,9},0)$ & $\sqrt{x_{2,9}^2+x_{2,9}+1}$ &  $\approx2.371762$ $^*$\\
        \hline
        5 & $S_2$ & $V_6$ & $e_9$ & $(x_{2,9},0)$ & $\sqrt{x_{2,9}^2+x_{2,9}+1}$ &  $\approx2.371762$ $^*$\\
        \hline
        6 & $Q$ & $V_9$ & $e_9+1-x_{2,9}$ & $V_9$ & 0 & $\approx0.73545$ \\
        \hline
        7 & $S_1$ & $V_2$ & $2e_9$ & $(y_{2,9},0)$ & $\sqrt{y_{2,9}^2-2y_{2,9}\cos\left(\tfrac{4\pi}{9}\right)+1}$ & $\approx2.37176$ $^*$ \\
        \hline
        8 & $S_2$ & $V_7$ & $2e_9$ & $(y_{2,9},0)$ & $\sqrt{y_{2,9}^2-2y_{2,9}\cos\left(\tfrac{4\pi}{9}\right)+1}$ & $\approx2.37176$ $^*$\\
        \hline
        9 & $Q$ & $V_1$ & $2e_9+\sqrt{y_{2,9}^2-2y_{2,9}\cos\left(\tfrac{2\pi}{9}\right)+1}$ & $V_1$ & 0 & $\approx2.12446$ \\
        \hline 
    \end{tabular}
\end{small}
\caption{Trajectories' details for the upper bound of \pe{9}{2}; 
see Figure~\ref{9gonk2fig}.}
\label{9gonk2tab}
\end{table}

%% file: UpperBoundsTabs/k2servants/10gonk2tab.tex
\begin{table}[h!]
\centering
\begin{small}
    \begin{tabular}{|c|c|c|c|c|c|c|}
        \hline
        $i$ & Agent & Vertex & $t_i$ & $Q^{(i)}$ & $d(Q^{(i)} , Exit)$ & Cost \\
        \hline
        1 & $S_1$ & $V_5$ & 0 & $V_9$ & $\sqrt{2-2\cos\left(\tfrac{4\pi}{5}\right)}$ & $\sqrt{2-2\cos\left(\tfrac{4\pi}{5}\right)}$ \\
        \hline
        2 & $S_2$ & $V_6$ & 0 & $V_9$ & $\sqrt{2-2\cos\left(\tfrac{3\pi}{5}\right)}$ & $\sqrt{2-2\cos\left(\tfrac{3\pi}{5}\right)}$ \\
        \hline
        3 & $Q$ & $V_9$ & 0 & $V_9$ & 0 & 0 \\
        \hline
        4 & $S_1$ & $V_4$ & $e_{10}$ & $(0.87153,0.02706)$ & $1.7716$ & $2.38956$ $^*$ \\
        \hline
        5 & $S_2$ & $V_7$ & $e_{10}$ & $(0.87153,0.02706)$ & $1.5331$ &  $2.15113$ \\
        \hline
        6 & $Q$ & $V_{10}$ & $0.74932$ & $V_{10}$ & 0 &  $0.74932$ \\
        \hline
        7 & $S_1$ & $V_3$ & $2e_{10}$ & $(0.54684,0.17765)$ & $1.1535$ & $2.38956$ $^*$ \\
        \hline
        8 & $S_2$ & $V_8$ & $2e_{10}$ & $(0.54684,0.17765)$ & $1.1535$ & $2.38956$ $^*$ \\
        \hline
        9 & $Q$ & $V_1$ & $1.7325$ & $V_1$ & 0 & $\approx1.7325$ \\
        \hline 
        10 & $Q$ & $V_2$ & $3e_{10}$ & $V_2$ & 0  & $\approx1.8541$ \\
        \hline 
    \end{tabular}
\end{small}
\caption{Trajectories' details for the upper bound of \pe{10}{2}; 
see Figure~\ref{10gonk2fig}.}
\label{10gonk2tab}
\end{table}

%% file: UpperBoundsTabs/k2servants/11gonk2tab.tex
\begin{table}[h!]
\centering
\begin{small}
    \begin{tabular}{|c|c|c|c|c|c|c|}
        \hline
        $i$ & Agent & Vertex & $t_i$ & $Q^{(i)}$ & $d(Q^{(i)} , Exit)$ & Cost \\
        \hline
        1 & $S_1$ & $V_5$ & 0 & $V_2$ & $\sqrt{2-2\cos\left(\tfrac{6\pi}{11}\right)}$ & $1.511499$  \\
        \hline
        2 & $S_2$ & $V_6$ & 0 & $V_2$ & $\sqrt{2-2\cos\left(\tfrac{8\pi}{11}\right)}$ & $1.511499$  \\
        \hline
        3 & $Q$ & $V_2$ & 0 & $V_2$ & 0 & 0 \\
        \hline
        4 & $S_1$ & $V_4$ & $e_{11}$ & $(0.81026,0.51379)$ & $1.48497$ & $ 2.04843$ \\
        \hline
        5 & $S_2$ & $V_7$ & $e_{11}$ & $(0.81026,0.51379)$  &  $1.93864$& $2.50210$ $^*$ \\
        \hline
        6 & $Q$ & $V_1$ & $ e_{11}+0.04100$ & $V_1$ & 0 & $ 0.60446$ \\
        \hline
        7 & $S_1$ & $V_3$ & $2e_{11}$ & $(0.79057,0.02054)$ & $1.34528$ & $2.47221$ \\
        \hline
        8 & $S_2$ & $V_8$ & $2e_{11}$ & $(0.79057,0.02054)$ & $1.37517$ & $2.50210$ $^*$ \\
        \hline
        9 & $Q$ & $V_{10}$ & $3e_{11}$ & $V_{10}$ & 0 &  $1.69039$\\
        \hline 
        10 & $S_2$ & $V_9$ & $3e_{11}$ & $V_{10}$ & $e_{11}$ & $2.25386$ \\
        \hline
        11 & $Q$ & $V_{11}$ & $4e_{11}$ & $V_{11}$ & 0 & $2.25386$ \\
        \hline 
    \end{tabular}
\end{small}
\caption{Trajectories' details for the upper bound of \pe{11}{2}; 
see Figure~\ref{11gonk2fig}.}
\label{11gonk2tab}
\end{table}

%% file: UpperBoundsFigs/k3servants/ALLgonk3fig.tex

\begin{figure}[H]
    \centering
    \begin{minipage}[t]{0.23\textwidth}
        \centering
        \begin{tikzpicture}[scale=0.15]
        \tikzset{
        mid arrow/.style={
            postaction={decorate, decoration={markings, mark=at position 0.6 with {\arrow{latex}}}}
        }
    }
\node (a) at (0,10) {};
    \fill [fill=black] (a) circle (0.105) node [above right] {$V_1$};
    \node (b) at (-10,0) {};
    \fill [fill=black] (b) circle (0.105) node [above left] {$V_2$};
    \node (c) at (0,-10) {};
    \fill [fill=black] (c) circle (0.105) node [below left] {$V_3$};
    \node (d) at (10,0) {};
    \fill [fill=black] (d) circle (0.105) node [above right] {$V_4$};
    \draw [line width=0.25mm, black] (10,0)--(0,10)--(-10,0)--(0,-10)--(10,0);

    \node at (a) (10,0) {};
    \fill [fill=red] (a) circle (0.45);
    \node at (b) (-10,0) {};
    \fill [fill=red] (b) circle (0.45);
    \node at (c) (0,-10) {};
    \fill [fill=red] (c) circle (0.45);

    \node (e) at (0,0) {};
    \fill [fill=blue] (e) circle (0.45);
    \node at (d) (10,0) {};
    \node[diamond, draw=blue, fill=blue, minimum size=0.15cm, inner sep=0] (d) at (10,0) {};
    \draw [line width=0.5mm, blue, mid arrow] (0,0) -- (10,0);            
        \end{tikzpicture}
        \caption{A search trajectory for \pe{4}{3}. Trajectory details can be found in Table~\ref{4gonk3tab}.}
        \label{4gonk3fig}
    \end{minipage}
    \hfill
    \begin{minipage}[t]{0.23\textwidth}
        \centering
        \begin{tikzpicture}[scale=0.15]
        \tikzset{
        mid arrow/.style={
            postaction={decorate, decoration={markings, mark=at position 0.6 with {\arrow{latex}}}}
        }
    }
        \node (a) at (3.1,9.5) {};
        \fill [fill=black] (a) circle (0.105) node [above right] {$V_1$};
        \node (b) at (-8.1,5.9) {};
        \fill [fill=black] (b) circle (0.105) node [above] {$V_2$};
        \node (c) at (-8.1,-5.9) {};
        \fill [fill=black] (c) circle (0.105) node [left] {$V_3$};
        \node (d) at (3.1,-9.5) {};
        \fill [fill=black] (d) circle (0.105) node [below left] {$V_4$};
        \node (e) at (10,0) {};
        \fill [fill=black] (e) circle (0.105) node [above right] {$V_5$};

        \draw [line width=0.25mm, black] (10,0)--(3.1,9.5)--(-8.1,5.9)--(-8.1,-5.9)--(3.1,-9.5)--(10,0);

\node[diamond, draw=red, fill=red, minimum size=0.15cm, inner sep=0] (f) at (-8.1,5.9) {};
        

\node[diamond, draw=red, fill=red, minimum size=0.15cm, inner sep=0] (f) at (-8.1,-5.9) {};
        

\node (a) at (3.1,9.5) {};
\fill [fill=red] (a) circle (0.45);

\node (a) at (6.3,0) {};
\fill [fill=blue] (a) circle (0.45);

\node (a) at (10,0) {};
\fill [fill=blue] (a) circle (0.45);

\node[diamond, draw=blue, fill=blue, minimum size=0.15cm, inner sep=0] (f) at (3.1,-9.5) {};
        

\draw [line width=0.5mm, blue] (6.3,0)  -- (10,0) node[currarrow,blue,pos=0.5,xscale=1,sloped,scale=0.75] {};

\draw [line width=0.5mm, blue, mid arrow] (10,0) -- (3.1,-9.5);

\draw[dashed] [line width=0.5mm, blue] [-latex] (6.3,0) -- (-8.1,5.9);

\draw[dashed] [line width=0.5mm, blue] [-latex] (6.3,0) -- (-8.1,-5.9);
        \end{tikzpicture}
        \caption{A search trajectory for \pe{5}{3}. Trajectory details can be found in Table~\ref{5gonk3tab}.}
        \label{5gonk3fig}
    \end{minipage}
    \hfill
    \begin{minipage}[t]{0.23\textwidth}
        \centering
        \begin{tikzpicture}[scale=0.15]
        \tikzset{
        mid arrow/.style={
            postaction={decorate, decoration={markings, mark=at position 0.6 with {\arrow{latex}}}}
        }
    }
        \node (a) at (5,8.7) {};
        \fill [fill=black] (a) circle (0.105) node [above] {$V_1$};
        \node (b) at (-5,8.7) {};
        \fill [fill=black] (b) circle (0.105) node [above] {$V_2$};
        \node (c) at (-10,0) {};
        \fill [fill=black] (c) circle (0.105) node [above left] {$V_3$};
        \node (d) at (-5,-8.7) {};
        \fill [fill=black] (d) circle (0.105) node [below left] {$V_4$};
        \node (e) at (5,-8.7) {};
        \fill [fill=black] (e) circle (0.105) node [below right] {$V_5$};
        \node (f) at (10,0) {};
        \fill [fill=black] (f) circle (0.105) node [above right] {$V_6$};
        
        \draw [line width=0.25mm, black] (10,0)--(5,8.7)--(-5,8.7)--(-10,0)--(-5,-8.7)--(5,-8.7)--(10,0);

        \node at (a) (5,8.7) {};
        \fill [fill=red] (a) circle (0.45);
        \node at (b) (-5,8.7) {};
        \node[diamond, draw=red, fill=red, minimum size=0.15cm, inner sep=0] (b) at (-5,8.7) {};
        \node at (c) (-10,0) {};
        \fill [fill=red] (c) circle (0.45);
        \node at (d) (-5,-8.7) {};
        \node[diamond, draw=red, fill=red, minimum size=0.15cm, inner sep=0] (d) at (-5,-8.7) {};
        \node at (e) (5,-8.7) {};
        \fill [fill=red] (e) circle (0.45);
        \node at (f) (10,0) {};
        \node[diamond, draw=red, fill=red, minimum size=0.15cm, inner sep=0] (f) at (10,0) {};
        \draw [line width=0.5mm, red, mid arrow] (5,8.7) -- (-5,8.7);
        \draw [line width=0.5mm, red, mid arrow] (-10,0) -- (-5,-8.7);
        \draw [line width=0.5mm, red, mid arrow] (5,-8.7) -- (10,0);

        \node (g) at (0,0) {};
        \fill [fill=blue] (g) circle (0.45);
        \draw[dashed] [line width=0.5mm, blue] [-latex] (0,0) -- (-5,8.7);
        \draw[dashed] [line width=0.5mm, blue] [-latex] (0,0) -- (-5,-8.7);
        \draw[dashed] [line width=0.5mm, blue] [-latex] (0,0) -- (10,0);
        \end{tikzpicture}
        \caption{A search trajectory for \pe{6}{3}. Trajectory details can be found in Table~\ref{6gonk3tab}.}
        \label{6gonk3fig}
    \end{minipage} 
\end{figure}


\begin{figure}[H]
\begin{minipage}[t]{0.23\textwidth}
        \centering
        \begin{tikzpicture}[scale=0.15]
        \tikzset{
        mid arrow/.style={
            postaction={decorate, decoration={markings, mark=at position 0.6 with {\arrow{latex}}}}
        }
    }
\node (a) at (6.2,7.8) {};
        \fill [fill=black] (a) circle (0.105) node [above] {$V_1$};
        \node (b) at (-2.2,9.7) {};
        \fill [fill=black] (b) circle (0.105) node [above left] {$V_2$};
        \node (c) at (-9,4.3) {};
        \fill [fill=black] (c) circle (0.105) node [above left] {$V_3$};
        \node (d) at (-9,-4.3) {};
        \fill [fill=black] (d) circle (0.105) node [below] {$V_4$};
        \node (e) at (-2.2,-9.7) {};
        \fill [fill=black] (e) circle (0.105) node [below left] {$V_5$};
        \node (f) at (6.2,-7.8) {};
        \fill [fill=black] (f) circle (0.105) node [below right] {$V_6$};
        \node (g) at (10,0) {};
        \fill [fill=black] (g) circle (0.105) node [above right] {$V_7$};

        \draw [line width=0.25mm, black] (10,0)--(6.2,7.8)--(-2.2,9.7)--(-9,4.3)--(-9,-4.3)--(-2.2,-9.7)--(6.2,-7.8)--(10,0);
        
\node (a) at (6.2,7.8) {};
\fill [fill=red] (a) circle (0.45);

\node[diamond, draw=red, fill=red, minimum size=0.15cm, inner sep=0] (f) at (-2.2,9.7) {};


\draw [line width=0.5mm, red, mid arrow] (6.2,7.8) -- (-2.2,9.7);

\node (a) at (-9,4.3) {};
\fill [fill=red] (a) circle (0.45);

\node[diamond, draw=red, fill=red, minimum size=0.15cm, inner sep=0] (f) at (-9,-4.3) {};


\draw [line width=0.5mm, red, mid arrow] (-9,4.3)  -- (-9,-4.3);

\node (a) at (6.2,-7.8) {};
\fill [fill=red] (a) circle (0.45);

\node[diamond, draw=red, fill=red, minimum size=0.15cm, inner sep=0] (f) at (10,0) {};


\draw [line width=0.5mm, red, mid arrow] (6.2,-7.8)  -- (10,0);

\node (a) at (0,0) {};
\fill [fill=blue] (a) circle (0.45);

\node[diamond, draw=blue, fill=blue, minimum size=0.15cm, inner sep=0] (f) at (-2.2,-9.7) {};


\draw [line width=0.5mm, blue, mid arrow] (0,0) --( -2.2,-9.7);

\draw[dashed] [line width=0.5mm, blue] [-latex] (0,0) -- (10,0);

\draw[dashed] [line width=0.5mm, blue] [-latex] (0,0) -- (-2.2,9.7);

\draw[dashed] [line width=0.5mm, blue] [-latex] (0,0) -- (-9,-4.3);
        \end{tikzpicture}
        \caption{A search trajectory for \pe{7}{3}. Trajectory details can be found in Table~\ref{7gonk3tab}.}
        \label{7gonk3fig}
    \end{minipage}
    \hfill
    \centering
    \begin{minipage}[t]{0.23\textwidth}
        \centering
        \begin{tikzpicture}[scale=0.15]
        \tikzset{
        mid arrow/.style={
            postaction={decorate, decoration={markings, mark=at position 0.6 with {\arrow{latex}}}}
        }
    }
\node (a) at (7.1,7.1) {};
        \fill [fill=black] (a) circle (0.105) node [above] {$V_1$};
        \node (b) at (0,10) {};
        \fill [fill=black] (b) circle (0.105) node [above right] {$V_2$};
        \node (c) at (-7.1,7.1) {};
        \fill [fill=black] (c) circle (0.105) node [above left] {$V_3$};
        \node (d) at (-10,0) {};
        \fill [fill=black] (d) circle (0.105) node [above left] {$V_4$};
        \node (e) at (-7.1,-7.1) {};
        \fill [fill=black] (e) circle (0.105) node [below left] {$V_5$};
        \node (f) at (0, -10) {};
        \fill [fill=black] (f) circle (0.105) node [below right] {$V_6$};
        \node (g) at (7.1,-7.1) {};
        \fill [fill=black] (g) circle (0.105) node [below right] {$V_7$};
        \node (h) at (10,0) {};
        \fill [fill=black] (h) circle (0.105) node [above right] {$V_8$};

        \draw [line width=0.25mm, black] (10,0)--(7.1,7.1)--(0,10)--(-7.1,7.1)--(-10,0)--(-7.1,-7.1)--(0, -10)--(7.1,-7.1)--(10,0);

        \node at (a) (7.1,7.1) {};
        \fill [fill=red] (a) circle (0.45);
        \node at (h) (10,0) {};
        \fill [fill=red] (h) circle (0.45);
        \node at (g) (7.1,-7.1) {};
        \node[diamond, draw=red, fill=red, minimum size=0.15cm, inner sep=0] (g) at (7.1,-7.1) {};
        \node at (c) (-7.1,7.1) {};
        \fill [fill=red] (c) circle (0.45);
        \node at (b) (0,10) {};
        \fill [fill=red] (b) circle (0.45);
        \node at (e) (-7.1,-7.1) {};
        \fill [fill=red] (e) circle (0.45);
        \node at (d) (-10,0) {};
        \fill [fill=red] (d) circle (0.45);
        \draw [line width=0.5mm, red, mid arrow] (7.1,7.1) -- (10,0);
        \draw [line width=0.5mm, red, mid arrow] (10,0) -- (7.1,-7.1);
        \draw [line width=0.5mm, red, mid arrow] (-7.1,7.1) -- (0,10);
        \draw [line width=0.5mm, red, mid arrow] (-7.1,-7.1) -- (-10,0);

        \node (i) at (0,0) {};
        \fill [fill=blue] (i) circle (0.45);
        \node (j) at (1.3,-1.3) {};
        \fill [fill=blue] (j) circle (0.45);
        \node (k) at (3.5,-8.5) {};
        \fill [fill=blue] (k) circle (0.45);
        \node[diamond, draw=blue, fill=blue, minimum size=0.15cm, inner sep=0] (f) at (0,-10) {};
        \draw [line width=0.5mm, blue] (0,0)  -- (1.3,-1.3) node[currarrow,blue,pos=0.5,xscale=1,sloped,scale=0.75] {};
        \draw [line width=0.5mm, blue, mid arrow] (1.3,-1.3) -- (3.5,-8.5);
        \draw [line width=0.5mm, blue] (3.5,-8.5)  -- (0,-10) node[currarrow,blue,pos=0.5,xscale=-1,sloped,scale=0.75] {};
        \draw[dashed] [line width=0.5mm, blue] [-latex] (3.5,-8.5) -- (7.1,-7.1);
            
        \end{tikzpicture}
        \caption{A search trajectory for \pe{8}{3}. Trajectory details can be found in Table~\ref{8gonk3tab}.}
        \label{8gonk3fig}
    \end{minipage}
    \hfill
    \begin{minipage}[t]{0.23\textwidth}
        \centering
        \begin{tikzpicture}[scale=0.15]
        \tikzset{
        mid arrow/.style={
            postaction={decorate, decoration={markings, mark=at position 0.6 with {\arrow{latex}}}}
        }
    }
\node (a) at (7.7,6.4) {};
\fill [fill=black] (a) circle (0.105) node [above] {$V_1$};
\node (b) at (1.7,9.8) {};
\fill [fill=black] (b) circle (0.105) node [above] {$V_2$};
\node (c) at (-5,8.7) {};
\fill [fill=black] (c) circle (0.105) node [above] {$V_3$};
\node (d) at (-9.4,3.4) {};
\fill [fill=black] (d) circle (0.105) node [above left] {$V_4$};
\node (e) at (-9.4,-3.4) {};
\fill [fill=black] (e) circle (0.105) node [below left] {$V_5$};
\node (f) at (-5,-8.7) {};
\fill [fill=black] (f) circle (0.105) node [below left] {$V_6$};
\node (g) at (1.7,-9.8) {};
\fill [fill=black] (g) circle (0.105) node [below right] {$V_7$};
\node (h) at (7.7,-6.4) {};
\fill [fill=black] (h) circle (0.105) node [below right] {$V_8$};
\node (i) at (10,0) {};
\fill [fill=black] (i) circle (0.105) node [above right] {$V_9$};

\draw [line width=0.25mm, black] (10,0)--(7.7,6.4)--(1.7,9.8)--(-5,8.7)--(-9.4,3.4)--(-9.4,-3.4)--(-5,-8.7)--(1.7,-9.8)--(7.7,-6.4)--(10,0);

\node[diamond, draw=red, fill=red, minimum size=0.15cm, inner sep=0] (f) at (-9.4,3.4) {};
        

\node[diamond, draw=red, fill=red, minimum size=0.15cm, inner sep=0] (f) at (-5,8.7) {};


\node (a) at (1.7,9.8) {};
\fill [fill=red] (a) circle (0.45) node [below] {};

\draw [line width=0.5mm, red, mid arrow] (-9.4,3.4)  -- (-5,8.7);

\draw [line width=0.5mm, red, mid arrow] (-5,8.7) -- (1.7,9.8);

\node[diamond, draw=red, fill=red, minimum size=0.15cm, inner sep=0] (f) at (-9.4,-3.4) {};


\node[diamond, draw=red, fill=red, minimum size=0.15cm, inner sep=0] (f) at (-5,-8.7) {};


\draw [line width=0.5mm, red, mid arrow] (-9.4,-3.4)  -- (-5,-8.7);

\node (a) at (1.7,-9.8) {};
\fill [fill=red] (a) circle (0.45);

\node (a) at (7.7,-6.4) {};
\fill [fill=red] (a) circle (0.45);

\draw [line width=0.5mm, red, mid arrow] (1.7,-9.8) -- (7.7,-6.4);

\node (a) at (8,0.15) {};
\fill [fill=blue] (a) circle (0.45);

\draw[dashed] [line width=0.5mm, blue] [-latex] (8,0.15) -- (-9.4,3.4);

\draw[dashed] [line width=0.5mm, blue] [-latex] (8,0.15) -- (-9.4,-3.4);

\node (a) at (10,0.03) {};
\fill [fill=blue] (a) circle (0.45);

\node (a) at (3.7,-0.07) {};
\fill [fill=blue] (a) circle (0.45) node [left] {};

\draw[dashed] [line width=0.5mm, blue] [-latex] (3.7,-0.07) -- (-5,8.7);

\draw[dashed] [line width=0.5mm, blue] [-latex] (3.7,-0.07) -- (-5,-8.7);

\draw [line width=0.35mm, blue] (8,0.15) -- (10,0.27)node[currarrow,blue,pos=0.5,xscale=1,sloped,scale=0.75] {};

\draw [line width=0.35mm, blue] (10,-0.2) -- (3.7,-0.2) node[currarrow,blue,pos=0.75,xscale=-1,sloped,scale=1] {};


\node (a) at (3.95,6.84) {};
\fill [fill=blue] (a) circle (0.45);

\node (a) at (7.7,6.4) {};
\fill [fill=blue] (a) circle (0.45);

\draw [line width=0.5mm, blue, mid arrow] (3.7,-0.07)  -- (3.95,6.84);

\draw [line width=0.5mm, blue, mid arrow] (3.95,6.84) -- (7.7,6.4);
        \end{tikzpicture}
        \caption{A search trajectory for \pe{9}{3}. Trajectory details can be found in Table~\ref{9gonk3tab}.}
        \label{9gonk3fig}
    \end{minipage}
\end{figure}

%% file: UpperBoundsTabs/k3servants/4gonk3tab.tex
\begin{table}[h!]
\centering
\begin{small}
    \begin{tabular}{|c|c|c|c|c|c|c|}
        \hline
        $i$ & Agent & Vertex & $t_i$ & $Q^{(i)}$ & $d(Q^{(i)} , Exit)$ & Cost \\
        \hline
        1 & $S_1$ & $V_1$ & 0 & $\mathcal{O}$ & 1 & 1 \\
        \hline
        2 & $S_2$ & $V_2$ & 0 & $\mathcal{O}$ & 1 & 1 \\
        \hline
        3 & $S_3$ & $V_3$ & 0 & $\mathcal{O}$ & 1 & 1 \\
        \hline
        4 & $Q$ & $V_4$ & 1 & $V_4$ & 0 & $1^*$ \\
        \hline
    \end{tabular}
\end{small}
\caption{Trajectories' details for the upper bound of \pe{4}{3}; 
see Figure~\ref{4gonk3fig}.}
\label{4gonk3tab}
\end{table}

%% file: UpperBoundsTabs/k3servants/5gonk3tab.tex
\begin{table}[h!]
\centering
\begin{small}
    \begin{tabular}{|c|c|c|c|c|c|c|}
        \hline
        $i$ & Agent & Vertex & $t_i$ & $Q^{(i)}$ & $d(Q^{(i)} , Exit)$ & Cost \\
        \hline
        1 & $S_1$ & $V_1$ & 0 & $(x_{3,5},0)$ & $ 1.00230$ & $ 1.00230$ \\
        \hline
        2 & $S_2$ & $V_2$ & 0 & $(x_{3,5},0)$ & $y_{3,5}+e_5$ & $y_{3,5}+e_5$ $^*$\\
        \hline
        3 & $S_3$ & $V_3$ & 0 & $(x_{3,5},0)$ & $y_{3,5}+e_5$ & $y_{3,5}+e_5$ $^*$\\
        \hline
        4 & $Q$ & $V_5$ & $y_{3,5}$ & $V_5$ & 0 & $y_{3,5}$\\
        \hline
        5 & $Q$ & $V_4$ & $y_{3,5}+e_5$ & $V_4$ & 0 & $y_{3,5}+e_5$ $^*$\\
        \hline
    \end{tabular}
\end{small}
\caption{Trajectories' details for the upper bound of \pe{5}{3}; 
see Figure~\ref{5gonk3fig}.}
\label{5gonk3tab}
\end{table}

%% file: UpperBoundsTabs/k3servants/6gonk3tab.tex
\begin{table}[h!]
\centering
\begin{small}
    \begin{tabular}{|c|c|c|c|c|c|c|}
        \hline
        $i$ & Agent & Vertex & $t_i$ & $Q^{(i)}$ & $d(Q^{(i)} , Exit)$ & Cost \\
        \hline
        1 & $S_1$ & $V_1$ & 0 & $\mathcal{O}$ & 1 & 1 \\
        \hline
        2 & $S_2$ & $V_3$ & 0 & $\mathcal{O}$ & 1 & 1 \\
        \hline
        3 & $S_3$ & $V_5$ & 0 & $\mathcal{O}$ & 1 & 1 \\
        \hline
        4 & $S_2$ & $V_4$ & $e_6$ & $\mathcal{O}$ & 1 & $2^*$ \\
        \hline
        5 & $S_3$ & $V_6$ & $e_6$ & $\mathcal{O}$ & 1 & $2^*$ \\
        \hline
        6 & $Q$ & $V_2$ & $e_6$ & $\mathcal{O}$ & 1 & $2^*$ \\
        \hline
    \end{tabular}
\end{small}
\caption{Trajectories' details for the upper bound of \pe{6}{3}; 
see Figure~\ref{6gonk3fig}.}
\label{6gonk3tab}
\end{table}

%% file: UpperBoundsTabs/k3servants/7gonk3tab.tex
\begin{table}[h!]
\centering
\begin{small}
    \begin{tabular}{|c|c|c|c|c|c|c|}
        \hline
        $i$ & Agent & Vertex & $t_i$ & $Q^{(i)}$ & $d(Q^{(i)} , Exit)$ & Cost \\
        \hline
        1 & $S_1$ & $V_1$ & 0 & $\cal{O}$ & 1 & 1 \\
        \hline
        2 & $S_2$ & $V_3$ & 0 & $\cal{O}$ & 1 & 1 \\
        \hline
        3 & $S_3$ & $V_6$ & 0 & $\cal{O}$ & 1 & 1 \\
        \hline
        4 & $S_1$ & $V_2$ & $e_7$ & $\cal{O}$ & 1 & $1+e_7$ $^*$\\
        \hline
        5 & $S_2$ & $V_4$ & $e_7$ & $\cal{O}$ & 1 & $1+e_7$ $^*$ \\
        \hline
        6 & $S_3$ & $V_7$ & $e_7$ & $\cal{O}$ & 1 & $1+e_7$ $^*$ \\
        \hline
        7 & $Q$ & $V_5$ & $1+e_7$  & $V_5$ & 0 & $1+e_7$ $^*$ \\
        \hline
    \end{tabular}
\end{small}
\caption{Trajectories' details for the upper bound of \pe{7}{3}; 
see Figure~\ref{7gonk3fig}.}
\label{7gonk3tab}
\end{table}

%% file: UpperBoundsTabs/k3servants/8gonk3tab.tex
\begin{table}[h!]
\centering
\begin{small}
    \begin{tabular}{|c|c|c|c|c|c|c|}
        \hline
        $i$ & Agent & Vertex & $t_i$ & $Q^{(i)}$ & $d(Q^{(i)} , Exit)$ & Cost \\
        \hline
        1 & $S_1$ & $V_1$ & 0 & $\mathcal{O}$ & 1 & 1 \\
        \hline
        2 & $S_2$ & $V_5$ & 0 & $\mathcal{O}$ & 1 & 1 \\
        \hline
        3 & $S_3$ & $V_3$ & 0 & $\mathcal{O}$ & 1 & 1 \\
        \hline
        4 & $S_1$ & $V_8$ & $e_8$ & $(0.13014,-0.13083)$ & 0.87964 & 1.64501 \\
        \hline
        5 & $S_2$ & $V_4$ & $e_8$ & $(0.13014,-0.13083)$ & 1.13768 & 1.90305 \\
        \hline
        6 & $S_3$ & $V_2$ & $e_8$ & $(0.13014,-0.13083)$ & 1.13829 & 1.90366 \\
        \hline
        7 & $S_1$ & $V_7$ & $2e_8$ & $(0.35355,-0.85355)$ & 0.38268 & ${1.91342}^*$ \\
        \hline
        8 & $Q$ & $V_6$ & 1.91342 & $V_6$ & 0 & ${1.91342}^*$ \\
        \hline
    \end{tabular}
\end{small}
\caption{Trajectories' details for the upper bound of \pe{8}{3}; 
see Figure~\ref{8gonk3fig}.}
\label{8gonk3tab}
\end{table}

%% file: UpperBoundsTabs/k3servants/9gonk3tab.tex
\begin{table}[h!]
\centering
\begin{small}
    \begin{tabular}{|c|c|c|c|c|c|c|}
        \hline
        $i$ & Agent & Vertex & $t_i$ & $Q^{(i)}$ & $d(Q^{(i)} , Exit)$ & Cost \\
        \hline
        1 & $S_1$ & $V_4$ & 0 & $(x_{3,9},0)$ & $1.91361$ & $1.91361$ $^*$\\
        \hline
        2 & $S_2$ & $V_5$ & 0 & $(x_{3,9},0)$ & $1.91361$ & $1.91361$ $^*$\\
        \hline
        3 & $S_3$ & $V_7$ & 0 & $(x_{3,9},0)$ & $1.24976$ & $1.24976$\\
        \hline
        4 & $Q$ & $V_9$ & $1-x_{3,9}$ & $V_9$ & 0 & $ 0.05688$\\
        \hline
        5 & $S_1$ & $V_3$ & $e_9$ & $(y_{3,9},0)$ & $ 1.22957$ & $ 1.91361$ $^*$\\
        \hline
        6 & $S_2$ & $V_6$ & $e_9$ & $(y_{3,9},0)$ & $ 1.22957$ & $ 1.91361$ $^*$\\
        \hline
        7 & $S_3$ & $V_8$ & $e_9$ & $(y_{3,9},0)$ & $ 0.75351$ & $ 1.43755$\\
        \hline
        8 & $S_1$ & $V_2$ & $2e_9$ & $(0.39472,0.68369)$ & $0.37356$ & $1.74164$\\
        \hline
        9 & $Q$ & $V_1$ & $2e_9+0.37356$ & $V_1$ & 0 & $1.74164$\\
        \hline 
    \end{tabular}
\end{small}
    \caption{Trajectories' details for the upper bound of \pe{9}{3}; 
see Figure~\ref{9gonk3fig}.}
\label{9gonk3tab}
\end{table}

%% file: UpperBoundsFigs/k4servants/ALLgonk4fig.tex

\begin{figure}[H]
    \centering
    \begin{minipage}[t]{0.23\textwidth}
        \centering
        \begin{tikzpicture}[scale=0.15]
        \tikzset{
        mid arrow/.style={
            postaction={decorate, decoration={markings, mark=at position 0.6 with {\arrow{latex}}}}
        }
    }
\node (a) at (3.1,9.5) {};
        \fill [fill=black] (a) circle (0.105) node [above right] {$V_1$};
        \node (b) at (-8.1,5.9) {};
        \fill [fill=black] (b) circle (0.105) node [above left] {$V_2$};
        \node (c) at (-8.1,-5.9) {};
        \fill [fill=black] (c) circle (0.105) node [below left] {$V_3$};
        \node (d) at (3.1,-9.5) {};
        \fill [fill=black] (d) circle (0.105) node [below right] {$V_4$};
        \node (e) at (10,0) {};
        \fill [fill=black] (e) circle (0.105) node [right] {$V_5$};

        \draw [line width=0.25mm, black] (10,0)--(3.1,9.5)--(-8.1,5.9)--(-8.1,-5.9)--(3.1,-9.5)--(10,0);

\node[diamond, draw=red, fill=red, minimum size=0.15cm, inner sep=0] (f) at (3.1,9.5) {};
 \node[above left] at (f.south east) {};

\node[diamond, draw=red, fill=red, minimum size=0.15cm, inner sep=0] (f) at (-8.1,5.9) {};
 \node[below left] at (f.south east) {};

\node[diamond, draw=red, fill=red, minimum size=0.15cm, inner sep=0] (f) at (-8.1,-5.9) {};
 \node[above left] at (f.south east) {};

\node[diamond, draw=red, fill=red, minimum size=0.15cm, inner sep=0] (f) at (3.1,-9.5) {};
 \node[below right] at (f.south east) {};

\node (a) at (0,0) {};
\fill [fill=blue] (a) circle (0.45) node [above right] {};

\node[diamond, draw=blue, fill=blue, minimum size=0.15cm, inner sep=0] (f) at (10,0) {};
 \node[below right] at (f.south east) {};

\draw [line width=0.5mm, blue, mid arrow] (0,0) -- (10,0);

\draw[dashed] [line width=0.5mm, blue] [-latex] (0,0) -- (3.1,9.5);

\draw[dashed] [line width=0.5mm, blue] [-latex] (0,0) -- (3.1,-9.5);
       
\draw[dashed] [line width=0.5mm, blue] [-latex] (0,0) -- (-8.1,5.9);
       
\draw[dashed] [line width=0.5mm, blue] [-latex] (0,0) -- (-8.1,-5.9);
        \end{tikzpicture}
\caption{A search trajectory for \pe{5}{4}. 
Trajectory details can be found in Table~\ref{5gonk4tab}.}
\label{5gonk4fig}
    \end{minipage}
    \hfill
    \begin{minipage}[t]{0.23\textwidth}
        \centering
        \begin{tikzpicture}[scale=0.15]
        \tikzset{
        mid arrow/.style={
            postaction={decorate, decoration={markings, mark=at position 0.6 with {\arrow{latex}}}}
        }
    }
  \node (a) at (5,8.7) {};
\fill [fill=black] (a) circle (0.105) node [above right] {$V_1$};
\node (b) at (-5,8.7) {};
\fill [fill=black] (b) circle (0.105) node [above left] {$V_2$};
\node (c) at (-10,0) {};
\fill [fill=black] (c) circle (0.105) node [left] {$V_3$};
\node (d) at (-5,-8.7) {};
\fill [fill=black] (d) circle (0.105) node [below left] {$V_4$};
\node (e) at (5,-8.7) {};
\fill [fill=black] (e) circle (0.105) node [below right] {$V_5$};
\node (f) at (10,0) {};
\fill [fill=black] (f) circle (0.105) node [right] {$V_6$};

\draw [line width=0.25mm, black] (10,0)--(5,8.7)--(-5,8.7)--(-10,0)--(-5,-8.7)--(5,-8.7)--(10,0);

\node (a) at (5,8.7) {};
\fill [fill=red] (a) circle (0.45) node [below left] {};

\node[diamond, draw=red, fill=red, minimum size=0.15cm, inner sep=0] (f) at (-5,8.7) {};
 \node[above right] at (f.south east) {};

\draw [line width=0.5mm, red, mid arrow] (5,8.7)  -- (-5,8.7);

\node[diamond, draw=red, fill=red, minimum size=0.15cm, inner sep=0] (f) at (-10,0) {};
 \node[below left] at (f.south east) {};

\node (a) at (-5,-8.7) {};
\fill [fill=red] (a) circle (0.45) node [above right] {};

\draw [line width=0.5mm, red, mid arrow] (-5,-8.7) -- (-10,0);

\node (a) at (5,-8.7) {};
\fill [fill=red] (a) circle (0.45) node [below left] {};

\node (a) at (10,0) {};
\fill [fill=red] (a) circle (0.45) node [below right] {};

\node (a) at (0,0) {};
\fill [fill=blue] (a) circle (0.45) node [above right] {};

\node (a) at (-7.5,4.3) {};
\fill [fill=blue] (a) circle (0.45) node [above left] {};

\draw [line width=0.5mm, blue, mid arrow] (0,0) -- (-7.5,4.3);

\draw[dashed] [line width=0.5mm, blue] [-latex] (-7.65,4.3) -- (-10.15,0);

\draw[dashed] [line width=0.5mm, blue] [-latex] (-7.65,4.3) -- (-5.15,8.7);  

        \end{tikzpicture}
\caption{A search trajectory for \pe{6}{4}. 
Trajectory details can be found in Table~\ref{6gonk4tab}.}
\label{6gonk4fig}
    \end{minipage}
    \hfill
    \begin{minipage}[t]{0.23\textwidth}
        \centering
        \begin{tikzpicture}[scale=0.15]
        \tikzset{
        mid arrow/.style={
            postaction={decorate, decoration={markings, mark=at position 0.6 with {\arrow{latex}}}}
        }
    }
       \node (a) at (6.2,7.8) {};
        \fill [fill=black] (a) circle (0.105) node [above right] {$V_1$};
        \node (b) at (-2.2,9.7) {};
        \fill [fill=black] (b) circle (0.105) node [above] {$V_2$};
        \node (c) at (-9,4.3) {};
        \fill [fill=black] (c) circle (0.105) node [above left] {$V_3$};
        \node (d) at (-9,-4.3) {};
        \fill [fill=black] (d) circle (0.105) node [below left] {$V_4$};
        \node (e) at (-2.2,-9.7) {};
        \fill [fill=black] (e) circle (0.105) node [below] {$V_5$};
        \node (f) at (6.2,-7.8) {};
        \fill [fill=black] (f) circle (0.105) node [below right] {$V_6$};
        \node (g) at (10,0) {};
        \fill [fill=black] (g) circle (0.105) node [right] {$V_7$};

        \draw [line width=0.25mm, black] (10,0)--(6.2,7.8)--(-2.2,9.7)--(-9,4.3)--(-9,-4.3)--(-2.2,-9.7)--(6.2,-7.8)--(10,0);
        
\node (a) at (-9,4.3) {};
\fill [fill=red] (a) circle (0.45) node [below right] {};

\node (a) at (-2.2,9.7) {};
\fill [fill=red] (a) circle (0.45) node [above left] {};

\node[diamond, draw=red, fill=red, minimum size=0.15cm, inner sep=0] (f) at (6.2,7.8) {};
 \node[right] at (f.south east) {};

\draw [line width=0.5mm, red, mid arrow] (-2.2,9.7)  -- (6.2,7.8);

\node (a) at (-9,-4.3) {};
\fill [fill=red] (a) circle (0.45) node [above right] {};

\node (a) at (-2.2,-9.7) {};
\fill [fill=red] (a) circle (0.45) node [below right] {};

\node (a) at (6.2,-7.8) {};
\fill [fill=red] (a) circle (0.45) node [right] {};

\draw [line width=0.5mm, red, mid arrow] (-2.2,-9.7)  -- (6.2,-7.8);

\node (a) at (5.1,0.07) {};
\fill [fill=blue] (a) circle (0.45) node [left] {};

\node (a) at (10,0.03) {};
\fill [fill=blue] (a) circle (0.45) node [below right] {};

\node (a) at (6.2,-0.07) {};
\fill [fill=blue] (a) circle (0.45) node [below right] {};

\draw[dashed] [line width=0.5mm, blue] [-latex] (6.2,-0.07) -- (6.2,7.8);

\draw[dashed] [line width=0.5mm, blue] [-latex] (6.2,-0.07) -- (6.2,-7.8);

\draw [line width=0.5mm, blue] (5.1,0.3) -- (10,0.3);

\draw [line width=0.5mm, blue] (10,-0.2)  -- (6.2,-0.2);

        \end{tikzpicture}
\caption{A search trajectory for \pe{7}{4}. 
Trajectory details can be found in Table~\ref{7gonk4tab}.}
\label{7gonk4fig}
    \end{minipage}
   \end{figure}


\begin{figure}[H]
    \centering
    \begin{minipage}[t]{0.23\textwidth}
        \centering
        \begin{tikzpicture}[scale=0.15]
        \tikzset{
        mid arrow/.style={
            postaction={decorate, decoration={markings, mark=at position 0.6 with {\arrow{latex}}}}
        }
    }
\node (a) at (7.1,7.1) {};
        \fill [fill=black] (a) circle (0.105) node [above right] {$V_1$};
        \node (b) at (0,10) {};
        \fill [fill=black] (b) circle (0.105) node [above] {$V_2$};
        \node (c) at (-7.1,7.1) {};
        \fill [fill=black] (c) circle (0.105) node [above left] {$V_3$};
        \node (d) at (-10,0) {};
        \fill [fill=black] (d) circle (0.105) node [left] {$V_4$};
        \node (e) at (-7.1,-7.1) {};
        \fill [fill=black] (e) circle (0.105) node [below left] {$V_5$};
        \node (f) at (0, -10) {};
        \fill [fill=black] (f) circle (0.105) node [below] {$V_6$};
        \node (g) at (7.1,-7.1) {};
        \fill [fill=black] (g) circle (0.105) node [below right] {$V_7$};
        \node (h) at (10,0) {};
        \fill [fill=black] (h) circle (0.105) node [right] {$V_8$};

        \draw [line width=0.25mm, black] (10,0)--(7.1,7.1)--(0,10)--(-7.1,7.1)--(-10,0)--(-7.1,-7.1)--(0, -10)--(7.1,-7.1)--(10,0);

\node (a) at (7.1,7.1) {};
\fill [fill=red] (a) circle (0.45) node [below left] {};

\node (a) at (-7.1,7.1) {};
\fill [fill=red] (a) circle (0.45) node [below right] {};

\node[diamond, draw=red, fill=red, minimum size=0.15cm, inner sep=0] (f) at (0,10) {};
 \node[above left] at (f.south east) {};

\draw[dashed] [line width=0.5mm, blue] [-latex] (0,0) -- (0,10);

\draw [line width=0.5mm, red, mid arrow] (-7.1,7.1)  -- (0,10);

\node (a) at (10,0) {};
\fill [fill=red] (a) circle (0.45) node [below right] {};

\node[diamond, draw=red, fill=red, minimum size=0.15cm, inner sep=0] (f) at (7.1,-7.1) {};
 \node[right] at (f.south east) {};

\draw[dashed] [line width=0.5mm, blue] [-latex] (0,0) -- (7.1,-7.1);

\draw [line width=0.5mm, red, mid arrow] (10,0) -- (7.1,-7.1);

\node (a) at (0,-10) {};
\fill [fill=red] (a) circle (0.45) node [below left] {};

\node[diamond, draw=red, fill=red, minimum size=0.15cm, inner sep=0] (f) at (-7.1,-7.1) {};
 \node[left] at (f.south east) {};

\draw[dashed] [line width=0.5mm, blue] [-latex] (0,0) -- (-7.1,-7.1);

\draw [line width=0.5mm, red, mid arrow] (0,-10)  -- (-7.1,-7.1);

\node (a) at (0,0) {};
\fill [fill=blue] (a) circle (0.45) node [above right] {};

\node[diamond, draw=blue, fill=blue, minimum size=0.15cm, inner sep=0] (f) at (-10,0) {};
 \node[below left] at (f.south east) {};

\draw [line width=0.5mm, blue, mid arrow] (0,0) -- (-10,0);
        \end{tikzpicture}
\caption{A search trajectory for \pe{8}{4}. 
Trajectory details can be found in Table~\ref{8gonk4tab}.}
\label{8gonk4fig}
    \end{minipage}
    \hfill
    \begin{minipage}[t]{0.23\textwidth}
        \centering
        \begin{tikzpicture}[scale=0.15]
        \tikzset{
        mid arrow/.style={
            postaction={decorate, decoration={markings, mark=at position 0.6 with {\arrow{latex}}}}
        }
    }
\node (a) at (7.7,6.4) {};
        \fill [fill=black] (a) circle (0.105) node [above right] {$V_1$};
        \node (b) at (1.7,9.8) {};
        \fill [fill=black] (b) circle (0.105) node [above] {$V_2$};
        \node (c) at (-5,8.7) {};
        \fill [fill=black] (c) circle (0.105) node [above left] {$V_3$};
        \node (d) at (-9.4,3.4) {};
        \fill [fill=black] (d) circle (0.105) node [left] {$V_4$};
        \node (e) at (-9.4,-3.4) {};
        \fill [fill=black] (e) circle (0.105) node [left] {$V_5$};
        \node (f) at (-5,-8.7) {};
        \fill [fill=black] (f) circle (0.105) node [below left] {$V_6$};
        \node (g) at (1.7,-9.8) {};
        \fill [fill=black] (g) circle (0.105) node [below] {$V_7$};
        \node (h) at (7.7,-6.4) {};
        \fill [fill=black] (h) circle (0.105) node [below right] {$V_8$};
        \node (i) at (10,0) {};
        \fill [fill=black] (i) circle (0.105) node [right] {$V_9$};

        \draw [line width=0.25mm, black] (10,0)--(7.7,6.4)--(1.7,9.8)--(-5,8.7)--(-9.4,3.4)--(-9.4,-3.4)--(-5,-8.7)--(1.7,-9.8)--(7.7,-6.4)--(10,0);

\node[diamond, draw=red, fill=red, minimum size=0.15cm, inner sep=0] (f) at (-9.4,3.4) {};
 \node[below left] at (f.south east) {};

 \draw[dashed] [line width=0.5mm, blue] [-latex] (0,0) -- (-9.4,3.4);

\node (a) at (-5,8.7) {};
\fill [fill=red] (a) circle (0.45) node [below right] {};

\draw [line width=0.5mm, red, mid arrow] (-5,8.7) -- (-9.4,3.4);

\node[diamond, draw=red, fill=red, minimum size=0.15cm, inner sep=0] (f) at (1.7,9.8) {};
 \node[above right] at (f.south east) {};

\draw[dashed] [line width=0.5mm, blue] [-latex] (0,0) -- (1.7,9.8);

\node (a) at (7.7,6.4) {};
\fill [fill=red] (a) circle (0.45) node [right] {};

\draw [line width=0.5mm, red, mid arrow] (7.7,6.4) -- (1.7,9.8);

\node (a) at (-9.4,-3.4) {};
\fill [fill=red] (a) circle (0.45) node [above right] {};

\node[diamond, draw=red, fill=red, minimum size=0.15cm, inner sep=0] (f) at (-5,-8.7) {};
 \node[below left] at (f.south east) {};

\draw[dashed] [line width=0.5mm, blue] [-latex] (0,0) -- (-5,-8.7);

\draw [line width=0.5mm, red, mid arrow] (-9.4,-3.4)  -- (-5,-8.7);

\node (a) at (1.7,-9.8) {};
\fill [fill=red] (a) circle (0.45) node [above] {};

\node[diamond, draw=red, fill=red, minimum size=0.15cm, inner sep=0] (f) at (7.7,-6.4) {};
 \node[right] at (f.south east) {};

\draw[dashed] [line width=0.5mm, blue] [-latex] (0,0) -- (7.7,-6.4);

\draw [line width=0.5mm, red, mid arrow] (1.7,-9.8) -- (7.7,-6.4);

\node (a) at (0,0) {};
\fill [fill=blue] (a) circle (0.45) node [above right] {};

\node[diamond, draw=blue, fill=blue, minimum size=0.15cm, inner sep=0] (f) at (10,0) {};
 \node[below right] at (f.south east) {};

\draw [line width=0.5mm, blue, mid arrow] (0,0)  -- (10,0);
       \end{tikzpicture}
\caption{A search trajectory for \pe{9}{4}. 
Trajectory details can be found in Table~\ref{9gonk4tab}.}
\label{9gonk4fig}
    \end{minipage}
    \hfill
    \begin{minipage}[t]{0.23\textwidth}
        \centering
        \begin{tikzpicture}[scale=0.15]
        \tikzset{
        mid arrow/.style={
            postaction={decorate, decoration={markings, mark=at position 0.6 with {\arrow{latex}}}}
        }
    }
   \node (a) at (8.1,5.9) {};
        \fill [fill=black] (a) circle (0.105) node [above right] {$V_1$};
        \node (b) at (3.1,9.5) {};
        \fill [fill=black] (b) circle (0.105) node [above] {$V_2$};
        \node (c) at (-3.1,9.5) {};
        \fill [fill=black] (c) circle (0.105) node [above] {$V_3$};
        \node (d) at (-8.1,5.9) {};
        \fill [fill=black] (d) circle (0.105) node [above left] {$V_4$};
        \node (e) at (-10,0) {};
        \fill [fill=black] (e) circle (0.105) node [left] {$V_5$};
        \node (f) at (-8.1,-5.9) {};
        \fill [fill=black] (f) circle (0.105) node [below left] {$V_6$};
        \node (g) at (-3.1,-9.5) {};
        \fill [fill=black] (g) circle (0.105) node [below] {$V_7$};
        \node (h) at (3.1,-9.5) {};
        \fill [fill=black] (h) circle (0.105) node [below] {$V_8$};
        \node (i) at (8.1,-5.9) {};
        \fill [fill=black] (i) circle (0.105) node [below right] {$V_9$};
        \node (j) at (10,0) {};
        \fill [fill=black] (j) circle (0.105) node [right] {$V_{10}$};

        \draw [line width=0.25mm, black] (10,0)--(8.1,5.9)--(3.1,9.5)--(-3.1,9.5)--(-8.1,5.9)--(-10,0)--(-8.1,-5.9)--(-3.1,-9.5)--(3.1,-9.5)--(8.1,-5.9)--(10,0);

\node[diamond, draw=red, fill=red, minimum size=0.15cm, inner sep=0] (f) at (-8.1,5.9) {};
 \node[below right] at (f.south east) {};
 

\node (a) at (-3.1,9.5) {};
\fill [fill=red] (a) circle (0.45) node [above left] {};

\draw [line width=0.5mm, red, mid arrow] (-3.1,9.5) -- (-8.1,5.9);

\node (a) at (-10,0) {};
\fill [fill=red] (a) circle (0.45) node [below left] {};

\node (a) at (-8.1,-5.9) {};
\fill [fill=red] (a) circle (0.45) node [above right] {};

\node[diamond, draw=red, fill=red, minimum size=0.15cm, inner sep=0] (f) at (-3.1,-9.5) {};
 \node[below right] at (f.south east) {};

\draw [line width=0.5mm, red, mid arrow] (-10,0) -- (-8.1,-5.9);

\draw [line width=0.5mm, red, mid arrow] (-8.1,-5.9)  -- (-3.1,-9.5);

\node (a) at (8.1,-5.9) {};
\fill [fill=red] (a) circle (0.45) node [above left] {};

\node (a) at (10,0) {};
\fill [fill=red] (a) circle (0.45) node [below right] {};

\node[diamond, draw=red, fill=red, minimum size=0.15cm, inner sep=0] (f) at (8.1,5.9) {};
 \node[below right] at (f.south east) {};


\node (a) at (3.1,9.5) {};
\fill [fill=red] (a) circle (0.45) node [above left] {};

\draw [line width=0.5mm, red, mid arrow] (10,0) -- (8.1,-5.9);

\draw [line width=0.5mm, red, mid arrow]  (3.1,9.5) -- (8.1,5.9);

\node (a) at (0,-0.55) {};
\fill [fill=blue] (a) circle (0.45) node [above] {};

\node (a) at (0,-6.73) {};
\fill [fill=blue] (a) circle (0.45) node [above right] {};

\node[diamond, draw=blue, fill=blue, minimum size=0.15cm, inner sep=0] (f) at (3.1,-9.5) {};
 \node[below right] at (f.south east) {};

\draw [line width=0.5mm, blue, mid arrow] (0,-0.55) -- (0,-6.73);

\draw [line width=0.5mm, blue, mid arrow] (0,-6.73) -- (3.1,-9.5);

\draw[dashed] [line width=0.5mm, blue] [-latex] (0,-6.73) -- (-3.1,-9.5);

\draw[dashed] [line width=0.5mm, blue] [-latex] (0,-0.55) -- (8.1,5.9);
\draw[dashed] [line width=0.5mm, blue] [-latex] (0,-0.55) -- (-8.1,5.9);
        \end{tikzpicture}
\caption{A search trajectory for \pe{10}{4}. 
Trajectory details can be found in Table~\ref{10gonk4tab}.}
\label{10gonk4fig}
    \end{minipage}
       \end{figure}

%% file: UpperBoundsTabs/k4servants/5gonk4tab.tex
\begin{table}[h!]
\centering
\begin{small}
    \begin{tabular}{|c|c|c|c|c|c|c|}
        \hline
        $i$ & Agent & Vertex & $t_i$ & $Q^{(i)}$ & $d(Q^{(i)} , Exit)$ & Cost \\
        \hline
        1 & $S_1$ & $V_1$ & 0 & $\cal{O}$ & 1 & 1 $^*$\\
        \hline
        2 & $S_2$ & $V_2$ & 0 & $\cal{O}$ & 1 & 1 $^*$\\
        \hline
        3 & $S_3$ & $V_3$ & 0 & $\cal{O}$ & 1 & 1 $^*$\\
        \hline
        4 & $S_4$ & $V_4$ & 0 & $\cal{O}$ & 1 & 1 $^*$\\
        \hline
        5 & $Q$ & $V_5$ & 1 & $V_5$ & 0 & 1 $^*$\\
        \hline
    \end{tabular}
\end{small}
\caption{Trajectories' details for the upper bound of \pe{5}{4}; 
see Figure~\ref{5gonk4fig}.}
\label{5gonk4tab}
\end{table}

%% file: UpperBoundsTabs/k4servants/6gonk4tab.tex
\begin{table}[h!]
\centering
\begin{small}
    \begin{tabular}{|c|c|c|c|c|c|c|}
        \hline
        $i$ & Agent & Vertex & $t_i$ & $Q^{(i)}$ & $d(Q^{(i)} , Exit)$ & Cost \\
        \hline
        1 & $S_1$ & $V_1$ & 0 & $\cal{O}$ & 1 & 1 \\
        \hline
        2 & $S_2$ & $V_4$ & 0 &  $\cal{O}$ & 1 & 1 \\
        \hline
        3 & $S_3$ & $V_6$ & 0 &  $\cal{O}$ & 1 & 1 \\
        \hline
        4 & $S_4$ & $V_5$ & 0 &  $\cal{O}$ & 1 & 1 \\
        \hline
        5 & $S_1$ & $V_2$ & 1 & $\left(-\tfrac{3}{4},\tfrac{\sqrt{3}}{4}\right)$ & 1/2 & 3/2 $^*$ \\
        \hline
        6 & $S_2$ & $V_3$ & 1 & $\left(-\tfrac{3}{4},\tfrac{\sqrt{3}}{4}\right)$ & 1/2 & 3/2 $^*$ \\
        \hline
    \end{tabular}
\end{small}
\caption{Trajectories' details for the upper bound of \pe{6}{4}; see Figure~\ref{6gonk4fig}.}
\label{6gonk4tab}
\end{table}

%% file: UpperBoundsTabs/k4servants/7gonk4tab.tex
\begin{table}[h!]
\centering
\begin{small}
    \begin{tabular}{|c|c|c|c|c|c|c|}
        \hline
        $i$ & Agent & Vertex & $t_i$ & $Q^{(i)}$ & $d(Q^{(i)} , Exit)$ & Cost \\
        \hline
        1 & $S_1$ & $V_3$ & 0 & $(x_{4,7},0)$ & $\sqrt{x_{4,7}^2-2\cos\left(\tfrac{6\pi}{7}\right)x_{4,7}+1}$ & $\sqrt{x_{4,7}^2-2\cos\left(\tfrac{6\pi}{7}\right)x_{4,7}+1}$ \\
        \hline
        2 & $S_2$ & $V_5$ & 0 & $(x_{4,7},0)$ & $\sqrt{x_{4,7}^2-2\cos\left(\tfrac{4\pi}{7}\right)x_{4,7}+1}$ & $\sqrt{x_{4,7}^2-2\cos\left(\tfrac{4\pi}{7}\right)x_{4,7}+1}$ \\
        \hline
        3 & $S_3$ & $V_2$ & 0 & $(x_{4,7},0)$ & $\sqrt{x_{4,7}^2-2\cos\left(\tfrac{4\pi}{7}\right)x_{4,7}+1}$ & $\sqrt{x_{4,7}^2-2\cos\left(\tfrac{4\pi}{7}\right)x_{4,7}+1}$ \\
        \hline
        4 & $S_4$ & $V_4$ & 0 & $(x_{4,7},0)$ & $\sqrt{x_{4,7}^2-2\cos\left(\tfrac{6\pi}{7}\right)x_{4,7}+1}$ & $\sqrt{x_{4,7}^2-2\cos\left(\tfrac{6\pi}{7}\right)x_{4,7}+1}$ \\
        \hline
        5 & $Q$ & $V_7$ & $1-x_{4,7}$ & $(1,0)$ & 0 & $1-x_{4,7}$ \\
        \hline
        6 & $S_2$ & $V_6$ & $e_7$ & $\left(\cos\left(\tfrac{2\pi}{7}\right),0\right)$ & $\sin\left(\tfrac{2\pi}{7}\right)$ & $e_7+\sin\left(\tfrac{2\pi}{7}\right)$ $^*$ \\
        \hline
        7 & $S_3$ & $V_1$ & $e_7$ & $\left(\cos\left(\tfrac{2\pi}{7}\right),0\right)$ & $\sin\left(\tfrac{2\pi}{7}\right)$ & $e_7+\sin\left(\tfrac{2\pi}{7}\right)$ $^*$ \\
        \hline
    \end{tabular}
\end{small}
\caption{Trajectories' details for the upper bound of \pe{7}{4}; 
see Figure~\ref{7gonk4fig}.}
\label{7gonk4tab}
\end{table}

%% file: UpperBoundsTabs/k4servants/8gonk4tab.tex
\begin{table}[h!]
\centering
\begin{small}
    \begin{tabular}{|c|c|c|c|c|c|c|}
        \hline
        $i$ & Agent & Vertex & $t_i$ & $Q^{(i)}$ & $d(Q^{(i)} , Exit)$ & Cost \\
        \hline
        1 & $S_1$ & $V_1$ & 0 & $\cal{O}$ & 1 & 1 \\
        \hline
        2 & $S_2$ & $V_3$ & 0 &  $\cal{O}$ & 1 & 1 \\
        \hline
        3 & $S_3$ & $V_8$ & 0 &  $\cal{O}$ & 1 & 1 \\
        \hline
        4 & $S_4$ & $V_6$ & 0 &  $\cal{O}$ & 1 & 1 \\
        \hline
        5 & $S_2$ & $V_2$ & $\sqrt{2-\sqrt{2}}$ & $\cal{O}$ & 1 & $1+\sqrt{2-\sqrt{2}}$ $^*$\\
        \hline
        6 & $S_3$ & $V_7$ & $\sqrt{2-\sqrt{2}}$ & $\cal{O}$ & 1 & $1+\sqrt{2-\sqrt{2}}$ $^*$\\
        \hline
        7 & $S_4$ & $V_5$ & $\sqrt{2-\sqrt{2}}$ & $\cal{O}$ & 1 & $1+\sqrt{2-\sqrt{2}}$ $^*$\\
        \hline
        8 & $Q$ & $V_4$ & $1+\sqrt{2-\sqrt{2}}$ & $V_4$ & $0$ & $1+\sqrt{2-\sqrt{2}}$ $^*$ \\
        \hline
    \end{tabular}
\end{small}
\caption{Trajectories' details for the upper bound of \pe{8}{4}; 
see Figure~\ref{8gonk4fig}.}
\label{8gonk4tab}
\end{table}

%% file: UpperBoundsTabs/k4servants/9gonk4tab.tex
\begin{table}[h!]
\centering
\begin{small}
    \begin{tabular}{|c|c|c|c|c|c|c|}
        \hline
        $i$ & Agent & Vertex & $t_i$ & $Q^{(i)}$ & $d(Q^{(i)} , Exit)$ & Cost \\
        \hline
        1 & $S_1$ & $V_1$ & 0 & $\cal{O}$ & 1 & 1 \\
        \hline
        2 & $S_2$ & $V_3$ & 0 & $\cal{O}$ & 1 & 1 \\
        \hline
        3 & $S_3$ & $V_5$ & 0 & $\cal{O}$ & 1 & 1 \\
        \hline
        4 & $S_4$ & $V_7$ & 0 & $\cal{O}$ & 1 & 1 \\
        \hline
        5 & $S_1$ & $V_2$ & $e_9$ & $\cal{O}$ & 1 & $1+e_9$ $^*$\\
        \hline
        6 & $S_2$ & $V_4$ & $e_9$ & $\cal{O}$ & 1 & $1+e_9$ $^*$\\
        \hline
        7 & $S_3$ & $V_6$ & $e_9$ & $\cal{O}$ & 1 & $1+e_9$ $^*$\\
        \hline
        8 & $S_4$ & $V_8$ & $e_9$ & $\cal{O}$ & 1 & $1+e_9$ $^*$\\
        \hline
        9 & $Q$ & $V_9$ & $1+e_9$ & $V_9$ & 0 & $1+e_9$ $^*$\\
        \hline 
    \end{tabular}
\end{small}
\caption{Trajectories' details for the upper bound of \pe{9}{4}; 
see Figure~\ref{9gonk4fig}.}
\label{9gonk4tab}
\end{table}

%% file: UpperBoundsTabs/k4servants/10gonk4tab.tex
\begin{table}[h!]
\centering
\begin{small}
\begin{tabular}{|c|c|c|c|c|c|c|}
        \hline
        $i$ & Agent & Vertex & $t_i$ & $Q^{(i)}$ & $d(Q^{(i)} , Exit)$ & Cost \\
        \hline
        1 & $S_1$ & $V_2$ & 0 & $(0,y_{4,10})$ & 1.05276 & 1.05276 \\
        \hline
        2 & $S_2$ & $V_3$ & 0 & $(0,y_{4,10})$ & 1.05276 & 1.052760 \\
        \hline
        3 & $S_3$ & $V_5$ & 0 & $(0,y_{4,10})$ & 1.00153 & 1.00153 \\
        \hline
        4 & $S_4$ & $V_{10}$ & 0 & $(0,y_{4,10})$ & 1.00153 & 1.00153 \\
        \hline
        5 & $S_1$ & $V_1$ & $e_{10}$ & $(0,y_{4,10})$ & 1.03349 & $1.65152$ $^*$ \\
        \hline
        6 & $S_2$ & $V_4$ & $e_{10}$ & $(0,y_{4,10})$ & 1.03349 & $1.65152$ $^*$  \\
        \hline
        7 & $S_3$ & $V_6$ & $e_{10}$ & $(0,y_{4,10})$ & $0.96851$ & $1.58654$  \\
        \hline
        8 & $S_4$ & $V_9$ & $e_{10}$ & $(0,y_{4,10})$ & $0.96851$ & $1.58654$  \\
        \hline
        9 & $S_3$ & $V_7$ & $2e_{10}$ & $(0,y_{4,10}-e_{10})$ & $0.41545$ & $1.65152$ $^*$ \\
        \hline 
        10 & $Q$ & $V_8$ & $1.65152$ & $V_8$ & 0 & $1.65152$ $^*$ \\
        \hline 
    \end{tabular}
\end{small}
\caption{Trajectories' details for the upper bound of \pe{10}{4}; 
see Figure~\ref{10gonk4fig}.}
\label{10gonk4tab}
\end{table}

%% file: GJL24_arxiv.bbl
\begin{thebibliography}{10}

\bibitem{ahlswede1987search}
R.~Ahlswede and I.~Wegener.
\newblock {\em Search problems}.
\newblock John Wiley \& Sons, Inc., 1987.

\bibitem{alpern2013search}
S.~Alpern, R.~Fokkink, L.~Gasieniec, R.~Lindelauf, and V.~Subrahmanian.
\newblock {\em Search Theory: A Game Theoretic Perspective}.
\newblock Springer, 2013.

\bibitem{AlpGal03}
S.~Alpern and S.~Gal.
\newblock {\em The Theory of Search Games and Rendezvous}.
\newblock Kluwer, 2003.

\bibitem{0001DJ19}
S.~Angelopoulos, C.~D{\"u}rr, and S.~Jin.
\newblock Best-of-two-worlds analysis of online search.
\newblock In R.~Niedermeier and C.~Paul, editors, {\em 36th International
  Symposium on Theoretical Aspects of Computer Science, STACS 2019, March
  13-16, 2019, Berlin, Germany}, volume 126 of {\em LIPIcs}, pages 7:1--7:17.
  Schloss Dagstuhl - Leibniz-Zentrum f{\"u}r Informatik, 2019.

\bibitem{AngelopoulosDL19}
S.~Angelopoulos, C.~D{\"u}rr, and T.~Lidbetter.
\newblock The expanding search ratio of a graph.
\newblock {\em Discret. Appl. Math}, 260:51--65, 2019.

\bibitem{baezayates1993searching}
R.~Baeza~Yates, J.~Culberson, and G.~Rawlins.
\newblock Searching in the plane.
\newblock {\em Information and Computation}, 106(2):234--252, 1993.

\bibitem{BagheriNO19}
I.~Bagheri, L.~Narayanan, and J.~Opatrny.
\newblock Evacuation of equilateral triangles by mobile agents of limited
  communication range.
\newblock In F.~Dressler and C.~Scheideler, editors, {\em ALGOSENSORS 2019},
  volume 11931 of {\em Lecture Notes in Computer Science}, pages 3--22.
  Springer, 2019.

\bibitem{BampasCGIKKP19}
E.~Bampas, J.~Czyzowicz, L.~Gasieniec, D.~Ilcinkas, R.~Klasing, T.~Kociumaka,
  and D.~Pajak.
\newblock Linear search by a pair of distinct-speed robots.
\newblock {\em Algorithmica}, 81(1):317--342, 2019.

\bibitem{beck1964linear}
A.~Beck.
\newblock On the linear search problem.
\newblock {\em Israel Journal of Mathematics}, 2(4):221--228, 1964.

\bibitem{behrouz2023byzantine}
P.~Behrouz, O.~Konstantinidis, N.~Leonardos, A.~Pagourtzis, I.~Papaioannou, and
  M.~Spyrakou.
\newblock Byzantine fault-tolerant protocols for (n, f)-evacuation from a
  circle.
\newblock In {\em International Symposium on Algorithmics of Wireless
  Networks}, pages 87--100. Springer, 2023.

\bibitem{BGMP2022pfaulty}
A.~Bonato, K.~Georgiou, C.~MacRury, and P.~Pra{\l}at.
\newblock Algorithms for p-faulty search on a half-line.
\newblock {\em Algorithmica}, pages 1--30, 2022.

\bibitem{Borowiecki0DK16}
P.~Borowiecki, S.~Das, D.~Dereniowski, and L.~Kuszner.
\newblock Distributed evacuation in graphs with multiple exits.
\newblock In J.~Suomela, editor, {\em Structural Information and Communication
  Complexity - 23rd International Colloquium, SIROCCO 2016, Helsinki, Finland,
  July 19-21, 2016, Revised Selected Papers}, volume 9988 of {\em Lecture Notes
  in Computer Science}, pages 228--241, 2016.

\bibitem{bose2017general}
P.~Bose and J.-L. De~Carufel.
\newblock A general framework for searching on a line.
\newblock {\em Theoretical Computer Science}, 703:1--17, 2017.

\bibitem{BrandtFRW20}
S.~Brandt, K.-T. Foerster, B.~Richner, and R.~Wattenhofer.
\newblock Wireless evacuation on m rays with k searchers.
\newblock {\em Theor. Comput. Sci}, 811:56--69, 2020.

\bibitem{brandt2017collaboration}
S.~Brandt, F.~Laufenberg, Y.~Lv, D.~Stolz, and R.~Wattenhofer.
\newblock Collaboration without communication: Evacuating two robots from a
  disk.
\newblock In {\em International Conference on Algorithms and Complexity}, pages
  104--115. Springer, 2017.

\bibitem{ChrobakGGM15}
M.~Chrobak, L.~Gasieniec, T.~Gorry, and R.~Martin.
\newblock Group search on the line.
\newblock In G.~F. Italiano, T.~Margaria-Steffen, J.~Pokorn{\'y}, J.-J.
  Quisquater, and R.~Wattenhofer, editors, {\em SOFSEM}, volume 8939 of {\em
  Lecture Notes in Computer Science}, pages 164--176. Springer, 2015.

\bibitem{chuangpishit2020multi}
H.~Chuangpishit, K.~Georgiou, and P.~Sharma.
\newblock A multi-objective optimization problem on evacuating 2 robots from
  the disk in the face-to-face model; trade-offs between worst-case and
  average-case analysis.
\newblock {\em Information}, 11(11):506, 2020.

\bibitem{ChuangpishitMNO20}
H.~Chuangpishit, S.~Mehrabi, L.~Narayanan, and J.~Opatrny.
\newblock Evacuating equilateral triangles and squares in the face-to-face
  model.
\newblock {\em Comput. Geom}, 89:101624, 2020.

\bibitem{Clp}
{COIN-OR}.
\newblock Clp: Coin-or linear programming solver.
\newblock \url{https://github.com/coin-or/Clp}.
\newblock Accessed: 2024-06-04.

\bibitem{Ipopt}
{COIN-OR}.
\newblock Ipopt: Interior point optimizer.
\newblock \url{https://github.com/coin-or/Ipopt}.
\newblock Accessed: 2024-06-19.

\bibitem{CzyzowiczDGKM16}
J.~Czyzowicz, S.~Dobrev, K.~Georgiou, E.~Kranakis, and F.~MacQuarrie.
\newblock Evacuating two robots from multiple unknown exits in a circle.
\newblock In {\em ICDCN}, pages 28:1--28:8. ACM, 2016.

\bibitem{CzyzowiczGGKMP14}
J.~Czyzowicz, L.~Gasieniec, T.~Gorry, E.~Kranakis, R.~Martin, and D.~Pajak.
\newblock Evacuating robots via unknown exit in a disk.
\newblock In F.~Kuhn, editor, {\em DISC 2014}, volume 8784 of {\em Lecture
  Notes in Computer Science}, pages 122--136. Springer, 2014.

\bibitem{CzyzowiczGGKKRW17}
J.~Czyzowicz, K.~Georgiou, M.~Godon, E.~Kranakis, D.~Krizanc, W.~Rytter, and
  M.~Wlodarczyk.
\newblock Evacuation from a disc in the presence of a faulty robot.
\newblock In S.~Das and S.~Tixeuil, editors, {\em SIROCCO 2017}, volume 10641
  of {\em Lecture Notes in Computer Science}, pages 158--173. Springer, 2017.

\bibitem{czyzowiczICALP}
J.~Czyzowicz, K.~Georgiou, R.~Killick, E.~Kranakis, D.~Krizanc, M.~Lafond,
  L.~Narayanan, J.~Opatrny, and S.~Shende.
\newblock {Energy Consumption of Group Search on a Line}.
\newblock In C.~Baier, I.~Chatzigiannakis, P.~Flocchini, and S.~Leonardi,
  editors, {\em 46th International Colloquium on Automata, Languages, and
  Programming (ICALP 2019)}, volume 132 of {\em Leibniz International
  Proceedings in Informatics (LIPIcs)}, pages 137:1--137:15, Dagstuhl, Germany,
  2019. Schloss Dagstuhl -- Leibniz-Zentrum f{\"u}r Informatik.

\bibitem{czyzowicz2021energy}
J.~Czyzowicz, K.~Georgiou, R.~Killick, E.~Kranakis, D.~Krizanc, M.~Lafond,
  L.~Narayanan, J.~Opatrny, and S.~Shende.
\newblock Time-energy tradeoffs for evacuation by two robots in the wireless
  model.
\newblock {\em Theoretical Computer Science}, 852:61--72, 2021.

\bibitem{czyzowicz2020priority123}
J.~Czyzowicz, K.~Georgiou, R.~Killick, E.~Kranakis, D.~Krizanc, L.~Narayanan,
  J.~Opatrny, and S.~Shende.
\newblock Priority evacuation from a disk: The case of n= 1, 2, 3.
\newblock {\em Theoretical Computer Science}, 806:595--616, 2020.

\bibitem{czyzowicz2020priority4}
J.~Czyzowicz, K.~Georgiou, R.~Killick, E.~Kranakis, D.~Krizanc, L.~Narayanan,
  J.~Opatrny, and S.~Shende.
\newblock Priority evacuation from a disk: The case of $n\geq 4$.
\newblock {\em Theoretical Computer Science}, 846:91--102, 2020.

\bibitem{czyzowicz2019groupkos}
J.~Czyzowicz, K.~Georgiou, and E.~Kranakis.
\newblock Group search and evacuation.
\newblock In {\em Distributed Computing by Mobile Entities}, pages 335--370.
  Springer, 2019.

\bibitem{czyzowicz2021searchbyz}
J.~Czyzowicz, K.~Georgiou, E.~Kranakis, D.~Krizanc, L.~Narayanan, J.~Opatrny,
  and S.~Shende.
\newblock Search on a line by byzantine robots.
\newblock {\em International Journal of Foundations of Computer Science},
  32(04):369--387, 2021.

\bibitem{CGKNOV20}
J.~Czyzowicz, K.~Georgiou, E.~Kranakis, L.~Narayanan, J.~Opatrny, and
  B.~Vogtenhuber.
\newblock {Evacuating Robots from a Disk Using Face-to-Face Communication}.
\newblock {\em {Discrete Mathematics \& Theoretical Computer Science}}, {vol.
  22 no. 4}, Aug. 2020.

\bibitem{czyzowicz2021groupevac}
J.~Czyzowicz, R.~Killick, E.~Kranakis, D.~Krizanc, L.~Narayanan, J.~Opatrny,
  D.~Pankratov, and S.~Shende.
\newblock Group evacuation on a line by agents with different communication
  abilities.
\newblock {\em ISAAC 2021}, pages 57:1--57:24, 2021.

\bibitem{czyzowicz2021searchnew}
J.~Czyzowicz, R.~Killick, E.~Kranakis, and G.~Stachowiak.
\newblock Search and evacuation with a near majority of faulty agents.
\newblock In {\em SIAM Conference on Applied and Computational Discrete
  Algorithms (ACDA21)}, pages 217--227. SIAM, 2021.

\bibitem{CzyzowiczKKNO19}
J.~Czyzowicz, E.~Kranakis, D.~Krizanc, L.~Narayanan, and J.~Opatrny.
\newblock Search on a line with faulty robots.
\newblock {\em Distributed Comput}, 32(6):493--504, 2019.

\bibitem{CzyzowiczKKNOS15}
J.~Czyzowicz, E.~Kranakis, D.~Krizanc, L.~Narayanan, J.~Opatrny, and S.~M.
  Shende.
\newblock Wireless autonomous robot evacuation from equilateral triangles and
  squares.
\newblock In S.~Papavassiliou and S.~Ruehrup, editors, {\em 14th International
  Conference, ADHOC-NOW}, volume 9143 of {\em Lecture Notes in Computer
  Science}, pages 181--194. Springer, 2015.

\bibitem{CzyzowiczKKNOS17}
J.~Czyzowicz, E.~Kranakis, D.~Krizanc, L.~Narayanan, J.~Opatrny, and S.~M.
  Shende.
\newblock Linear search with terrain-dependent speeds.
\newblock In D.~Fotakis, A.~Pagourtzis, and V.~T. Paschos, editors, {\em
  Algorithms and Complexity - 10th International Conference, CIAC 2017, Athens,
  Greece, May 24-26, 2017, Proceedings}, volume 10236 of {\em Lecture Notes in
  Computer Science}, pages 430--441, 2017.

\bibitem{disser2019evacuating}
Y.~Disser and S.~Schmitt.
\newblock Evacuating two robots from a disk: a second cut.
\newblock In {\em International Colloquium on Structural Information and
  Communication Complexity}, pages 200--214. Springer, 2019.

\bibitem{dunning2017jump}
I.~Dunning, J.~Huchette, and M.~Lubin.
\newblock Jump: A modeling language for mathematical optimization.
\newblock {\em SIAM review}, 59(2):295--320, 2017.

\bibitem{feinerman2017ants}
O.~Feinerman and A.~Korman.
\newblock The ants problem.
\newblock {\em Distributed Computing}, 30(3):149--168, 2017.

\bibitem{FeketeGK10}
S.~P. Fekete, C.~Gray, and A.~Kr{\"o}ller.
\newblock Evacuation of rectilinear polygons.
\newblock In W.~Wu and O.~Daescu, editors, {\em Combinatorial Optimization and
  Applications - 4th International Conference, COCOA 2010, Kailua-Kona, HI,
  USA, December 18-20, 2010, Proceedings, Part I}, volume 6508 of {\em Lecture
  Notes in Computer Science}, pages 21--30. Springer, 2010.

\bibitem{11340}
P.~Flocchini, G.~Prencipe, and N.~Santoro, editors.
\newblock {\em Distributed Computing by Mobile Entities, Current Research in
  Moving and Computing}, volume 11340 of {\em Lecture Notes in Computer
  Science}.
\newblock Springer, 2019.

\bibitem{GGK2022asym}
K.~Georgiou, N.~Giachoudis, and E.~Kranakis.
\newblock Evacuation from a disk for robots with asymmetric communication.
\newblock In {\em 33rd International Symposium on Algorithms and Computation
  (ISAAC 2022)}. Schloss Dagstuhl-Leibniz-Zentrum f{\"u}r Informatik, 2022.

\bibitem{georgiou2022asymmetricevacuation}
K.~Georgiou, N.~Giachoudis, and E.~Kranakis.
\newblock Evacuation from a disk for robots with asymmetric communication.
\newblock In {\em 33rd International Symposium on Algorithms and Computation
  (ISAAC 2022)}. Schloss Dagstuhl-Leibniz-Zentrum f{\"u}r Informatik, 2022.

\bibitem{GJ22-triangle-algosensros}
K.~Georgiou and W.~Jang.
\newblock Triangle evacuation of 2 agents in the wireless model.
\newblock In {\em International Symposium on Algorithms and Experiments for
  Wireless Sensor Networks}, pages 77--90. Springer, 2022.

\bibitem{georgiou2022triangle}
K.~Georgiou and W.~Jang.
\newblock Triangle evacuation of 2 agents in the wireless model.
\newblock In {\em Algorithmics of Wireless Networks: 18th International
  Symposium on Algorithmics of Wireless Networks, ALGOSENSORS 2022, Potsdam,
  Germany, September 8--9, 2022, Proceedings}, pages 77--90. Springer, 2022.

\bibitem{kranakis2019search}
K.~Georgiou, G.~Karakostas, and E.~Kranakis.
\newblock Search-and-fetch with 2 robots on a disk: Wireless and face-to-face
  communication models.
\newblock {\em Discrete Mathematics \& Theoretical Computer Science}, 21, 2019.

\bibitem{GEORGIOU202118fetch1}
K.~Georgiou, G.~Karakostas, and E.~Kranakis.
\newblock Treasure evacuation with one robot on a disk.
\newblock {\em Theoretical Computer Science}, 852:18--28, 2021.

\bibitem{GeorgiouKLPP19}
K.~Georgiou, E.~Kranakis, N.~Leonardos, A.~Pagourtzis, and I.~Papaioannou.
\newblock Optimal circle search despite the presence of faulty robots.
\newblock In F.~Dressler and C.~Scheideler, editors, {\em ALGOSENSORS 2019},
  volume 11931 of {\em Lecture Notes in Computer Science}, pages 192--205.
  Springer, 2019.

\bibitem{GLLKllp2023}
K.~Georgiou, S.~Leizerovich, J.~Lucier, and S.~Kundu.
\newblock Evacuating from $\ell_p$ unit disks in the wireless model.
\newblock {\em Theoretical Computer Science}, 944:113675, 2023.

\bibitem{GLweightedLine2023}
K.~Georgiou and J.~Lucier.
\newblock Weighted group search on a line \& implications to the priority
  evacuation problem.
\newblock {\em Theoretical Computer Science}, 939:1--17, 2023.

\bibitem{GW24-iwoca}
K.~Georgiou and X.~Wang.
\newblock Weighted group search on the disk \& improved lower bounds for
  priority evacuation.
\newblock In {\em 35th International Workshop on Combinatorial Algorithms
  (IWOCA'24)}, Lecture Notes in Computer Science, page to appear, 2024.

\bibitem{hohzaki2016search}
R.~Hohzaki.
\newblock Search games: Literature and survey.
\newblock {\em Journal of the Operations Research Society of Japan},
  59(1):1--34, 2016.

\bibitem{kleinberg1994line}
J.~M. Kleinberg.
\newblock On-line search in a simple polygon.
\newblock In {\em SODA}, volume~94, pages 8--15. Citeseer, 1994.

\bibitem{MillerP15}
A.~Miller and A.~Pelc.
\newblock Tradeoffs between cost and information for rendezvous and treasure
  hunt.
\newblock {\em J. Parallel Distributed Comput}, 83:159--167, 2015.

\bibitem{pattanayak2018evacuating}
D.~Pattanayak, H.~Ramesh, P.~S. Mandal, and S.~Schmid.
\newblock Evacuating two robots from two unknown exits on the perimeter of a
  disk with wireless communication.
\newblock In {\em Proceedings of the 19th International Conference on
  Distributed Computing and Networking}, pages 1--4, 2018.

\bibitem{Sun2020}
X.~Sun, Y.~Sun, and J.~Zhang.
\newblock Better upper bounds for searching on a line with byzantine robots.
\newblock In {\em Complexity and Approximation}, pages 151--171. Springer,
  2020.

\end{thebibliography}
